\documentclass[11pt]{article}
\usepackage{amsmath}
\usepackage{graphicx,psfrag,epsf}
\usepackage{enumerate}
\usepackage{url} 

\usepackage{color, framed, alltt, amssymb,bbm,amsthm,tabularx,caption,subcaption,amsmath,mathtools,cancel,lipsum,subcaption,float}
\usepackage[outercaption]{sidecap}
\usepackage[mathscr]{euscript}
\usepackage[normalem]{ulem}
\usepackage{arydshln}
\usepackage{etoolbox}

\useunder{\uline}{\ul}{}
\DeclarePairedDelimiter{\norm}{\lVert}{\rVert}

\newcommand{\empavg}{\frac{1}{n}\sum_{i=1}^{n}}
\newcommand{\ltheta}{ \hat{\theta} }

\DeclareMathOperator*{\argmin}{\arg\!\min}
\usepackage[linesnumbered,ruled]{algorithm2e}

\newtheorem{theorem}{Theorem}[section]

\newtheorem{lemma}[theorem]{Lemma}

\newtheorem{proposition}{Proposition}[section]
\newtheorem{definition}{Definition}[section]

\newtheorem{assumption}{Assumption}

\usepackage[utf8]{inputenc}
\usepackage[round]{natbib}
\usepackage[resetlabels]{multibib}
\newcites{Supp}{References}
\usepackage[hypertexnames=false]{hyperref}

\newcommand{\blind}{1}

\begin{document}
\def\spacingset#1{\renewcommand{\baselinestretch}%
{#1}\small\normalsize} \spacingset{1}


\if1\blind
{
\title{\bf PUlasso: High-dimensional variable selection with presence-only data}
\author{Hyebin Song\\
Department of Statistics, University of Wisconsin-Madison\\
and \\
Garvesh Raskutti \thanks{
Both HS and GR were partially supported by NSF-DMS 1407028. GR was also partially supported by ARO W911NF-17-1-0357.}\hspace{.2cm}\\
Department of Statistics, University of Wisconsin-Madison}
\maketitle
} \fi

\if0\blind
{
\bigskip
\bigskip
\bigskip
\begin{center}
{\LARGE\bf PUlasso: High-dimensional variable selection with presence-only data}
\end{center}
\medskip
} \fi


\begin{abstract}
In various real-world problems, we are presented with classification problems with \emph{positive and unlabeled data}, referred to as presence-only responses. In this paper, we study variable selection in the context of presence only responses where the number of features or covariates $p$ is large. The combination of \emph{presence-only responses} and \emph{high dimensionality} presents both statistical and computational challenges. In this paper, we develop the \emph{PUlasso} algorithm for variable selection and classification with positive and unlabeled responses. Our algorithm involves using the majorization-minimization (MM) framework which is a generalization of the well-known expectation-maximization (EM) algorithm. In particular to make our algorithm scalable, we provide two computational speed-ups to the standard EM algorithm. We provide a theoretical guarantee where we first show that our algorithm converges to a stationary point, and then prove that any stationary point within a local neighborhood of the true parameter achieves the minimax optimal mean-squared error under both strict sparsity and group sparsity assumptions. We also demonstrate through simulations that our algorithm out-performs state-of-the-art algorithms in the moderate $p$ settings in terms of classification performance. Finally, we demonstrate that our PUlasso algorithm performs well on a biochemistry example.
\end{abstract}

\noindent%
\vfill

\newpage
\spacingset{1.45} 

\section{Introduction}
\label{sec:intro}

In many classification problems, we are presented with the problem where it is either prohibitively expensive or impossible to obtain negative responses and we only have positive and unlabeled \emph{presence-only} responses (see e.g.~\cite{ward_presence-only_2009}). For example, presence-only data is prevalent in geographic species distribution modeling in ecology where presences of species in specific locations are easily observed but absences are difficult to track (see e.g.~\cite{ward_presence-only_2009}), text mining (see e.g.~\cite{liu_building_2003}), bioinformatics (see e.g.~\cite{elkan_learning_2008}) and many other settings. Classification with presence-only data is sometimes referred to as PU-learning (see e.g.~\cite{liu_building_2003,elkan_learning_2008}). In this paper we address the problem of variable selection with presence-only responses.

\subsection{Motivating application: Biotechnology}

Although the theory and methodology we develop apply generally, a concrete application that motivates this work arises from biological systems engineering. In particular, recent high-throughput technologies generate millions of biological sequences from a library for a protein or enzyme of interest (see e.g.~\cite{Fowler2014,Hietpas2011}). In Section~\ref{SecExperiment} the enzyme of interest is beta-glucosidase (BGL) which is used to decompose disaccharides into glucose which is an important step in the process of converting plant matter to bio-fuels~(\cite{Romero2015}). The performance of the BGL enzyme is measured by the concentration of glucose that is produced and a positive response arises when the disaccharide is decomposed to glucose and a negative response arises otherwise. Hence there are two scientific goals: firstly to determine how the sequence structure influences the biochemical functionality; secondly, using this relationship to engineer and design BGL sequences with improved functionality.

Given these two scientific goals, we are interested in both the \emph{variable selection} and \emph{classification} problem since we want to determine which positions in the sequence most influence positive responses as well as classify which protein sequences are functional. Furthermore the number of variables here is large since we need to model long and complex biological sequences. Hence our variable selection problem is \emph{high-dimensional}. In Section~\ref{SecExperiment} we demonstrate the success of our algorithm in this application context.

\subsection{Problem setup}
To state the problem formally, let $x \in \mathbb{R}^p$ be a $p$-dimensional covariate such that $x \sim \mathbb{P}_X$, $y \in \{0,1\}$ an associated response, and $z \in \{0,1\}$ an associated label. If a sample is labeled ($z=1$), its associated outcome is positive ($y=1$). On the other hand, if a sample is unlabeled ($z=0$), it is assumed to be randomly drawn from the population with only covariates $x$ not the response $y$ being observed. Given $n_{\ell}$ labeled and $n_u$ unlabeled samples, the goal is to draw inferences about the relationship between $y$ and $x$.
We model the relationship between the probability of a response $y$ being positive and $(x, \theta)$ using the standard logistic regression model:\begin{equation}
\label{eq:py|x}
\mathbb{P}(y=1|x;\theta) =  \frac{e^{\eta_\theta(x)}}{1+e^{\eta_\theta(x)}}, \qquad \eta_\theta(x) = \theta^Tx
\end{equation}
and $y|x\sim \mathbb{P}(\cdot|x;\theta^*)$ where $\theta^* \in \mathbb{R}^p$  refers to the unknown true parameter. Also, we assume the label $z$ is assigned only based on the latent response $y$ independent from $x$. Viewing $z$ as a noisy observation of latent $y$, this assumption corresponds to a missing at random assumption, a classical assumption in latent variable problems.

Given such $z$, we select $n_l$ labeled and $n_u$ unlabeled samples from samples with $z=1$ and $z=0$ respectively. An important issue is how  positive and unlabeled samples are selected. In this paper we adopt a case-control approach (for example,~\cite{mccullagh_generalized_1989}) which is suitable for our biotechnology application and many others. In particular we introduce another binary random variable $s \in \{0,1\}$ representing whether a sample is selected ($s=1$) or not ($s=0$) to model different sampling rates in selecting labeled and unlabeled samples. Since there are $n_\ell$ labeled and $n_u$ unlabeled samples, we have
\begin{equation*}
	\dfrac{\mathbb{P}(z=1|s=1)}{\mathbb{P}(z=0|s=1)} = \dfrac{n_\ell}{n_u},
\end{equation*}
and we see only selected samples, $(x_i,z_i,s_i=1)_{i=1}^{n_\ell +n_u}$. It is further assumed that the selection is only based on the label $z$, independent of $x$ and $y$. We note that this case-control scheme~(\cite{lancaster_case-control_1996,ward_presence-only_2009}), opposed to the single-training sampling scheme~(\cite{elkan_learning_2008}) is needed to model the case where unlabeled samples are random draws from the original population, since positive samples have to be over-represented in the dataset to satisfy such model assumption. 

In our biotechnology application the case-control setting is appropriate since the high-throughput technology leads to the unlabeled samples being drawn randomly from the original population~(see~\cite{Romero2015} for details). As is displayed in Fig.~\ref{fig:htpseqfcn}, sequences are selected randomly from a library and positive samples are generated through a screening step. Hence the positive sequences are sampled randomly from the positive sequences while the unlabeled sequences are based on random sampling from the original sequence library. This experiment corresponds exactly to the case-control sampling scheme discussed.

\begin{figure}[htbp]
    \centering
    \includegraphics{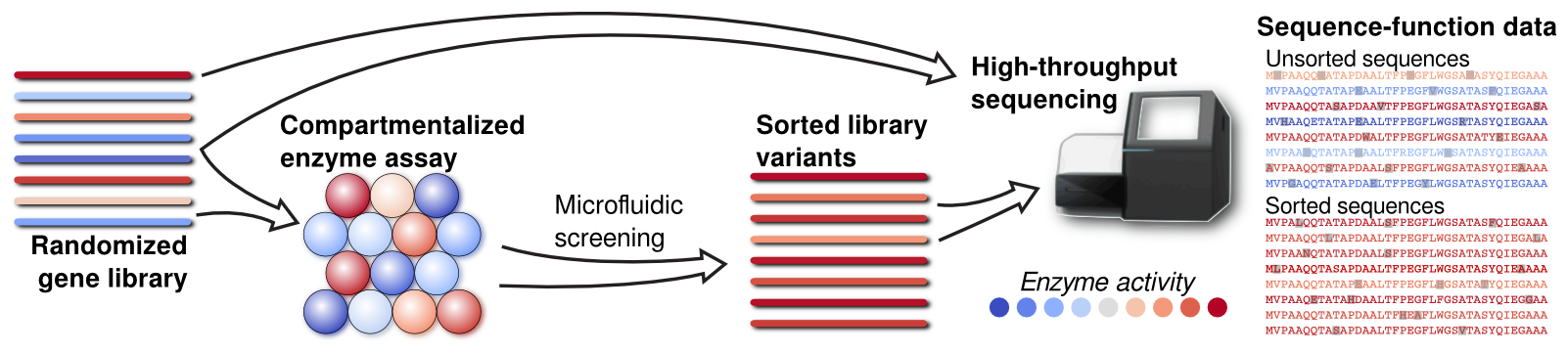}
    \caption{High-throughput sequencing diagram}
    \label{fig:htpseqfcn}
\end{figure}
Furthermore, the true positive prevalence is
\begin{equation}\label{def:pi}
	\pi := \mathbb{P}(y=1) = \int \dfrac{e^{\eta_{\theta^*}(x)}}{1+e^{\eta_{\theta^*}(x)}}d \mathbb{P}_X(x) \in (0,1)
\end{equation}
and $\pi$ is assumed known. In our biotechnology application, $\pi$ is estimated precisely using an alternative experiment~(\cite{Romero2015}). 

In the biological sequence engineering example, $(x_i)_{i=1}^{n_\ell +n_u}$ correspond to binary covariates of biological sequences. In the BGL example, for each of the $d$ positions, there are $M$ possible categories of amino acids. Therefore the covariates correspond to the indicator of an amino acid appearing in a given position~($p = O(dM)$) as well as pairs of amino acids~($p = O(d^2M^2)$), and so on. Here $d = O(1000)$ and $M \approx 20$ make the problem high-dimensional.

High-dimensional PU-learning presents computational challenges since the standard logistic regression objective leads to a non-convex likelihood when we have positive and unlabeled data. To address this challenge, we build on the expectation-maximization (EM) procedure developed in \cite{ward_presence-only_2009}  and provide two computational speed-ups. In particular we introduce the \emph{PUlasso} for high-dimensional variable selection with positive and unlabeled data. Prior work that involves the EM algorithm in the low-dimensional setting in \cite{ward_presence-only_2009} involves solving a logistic regression model at the M-step. To adapt to the high-dimensional setting and make the problem scalable, we include an $\ell_1$-sparsity or
$\ell_1/\ell_2$-group sparsity penalty and provide two speed-ups. Firstly we use a quadratic majorizer of the logistic regression objective, and secondly we use techniques in linear algebra to exploit sparsity of the design matrix $X$ which commonly arises in the applications we are dealing with. Our PUlasso algorithm fits into the majorization-minimization~(MM) framework~(see e.g.~\cite{lange_optimization_2000, Ortega}) for which the EM algorithm is a special case.

\subsection{Our contributions}

In this paper, we make the following major contributions:

\begin{itemize}
\item Develop the PUlasso algorithm for doing variable selection and classification with presence-only data. In particular we build on the existing EM algorithm developed in~\cite{ward_presence-only_2009} and add two computational speed-ups, quadratic majorization and exploiting sparse matrices. These two speed-ups improve speed by several orders of magnitude and allows our algorithm to scale to datasets with millions of samples and covariates.
\item Provide theoretical guarantees for our algorithm. First we show that our algorithm converges to a stationary point of the non-convex objective, and then show that any stationary point within a local neighborhood of $\theta^*$ achieves the minimax optimal mean-squared error for sparse vectors. To provide statistical guarantees we extend the existing results of generalized linear model with a canonical link function~(\cite{negahban_unified_2012,loh_regularized_2013}) to a non-canonical link function and show optimality of stationary points of non-convex objectives in high-dimensional statistics.
To the best of our knowledge the PUlasso is the first algorithm where PU-learning is provably optimal in the high-dimensional setting.
\item Demonstrate through a simulation study that our algorithm performs well in terms of classification compared to state-of-the-art PU-learning methods in~\cite{du_marthinus_convex_2015,elkan_learning_2008,liu_building_2003}, both for low-dimensional and high-dimensional problems. 
\item Demonstrate that our PUlasso algorithm allows us to develop improved protein-engineering approaches. In particular we apply our PUlasso algorithm to sequences of BGL~(beta-glucosidase) enzymes to determine which sequences are functional. We demonstrate that sequences selected by our algorithm have a good predictive accuracy and we also provide a scientific experiment which shows that the variables selected lead to BGL proteins that are engineered with improved functionality. 
\end{itemize}

The remainder of the paper is organized as follows: in Section~\ref{SecPU} we provide the background and introduce the PUlasso algorithm, including our two computational speed-ups and provide an algorithmic guarantee that our algorithm converges to a stationary point;  in Section~\ref{SecStat} we provide statistical mean-squared error guarantees which show that our PUlasso algorithm achieves the minimax rate; Section~\ref{SecSim} provides a comparison in terms of classification performance of our PUlasso algorithm to state-of-the-art PU-learning algorithms; finally in Section~\ref{SecExperiment}, we apply our PUlasso algorithm to the BGL data application and provide both a statistical validation and simple scientific validation for our selected variables.  

{\bf Notation: } For scalars $a,b \in \mathbb{R}$, we denote $a \wedge b  =\min\{a,b\}, a\vee b = \max \{a,b\}$. Also, we denote $a\gtrsim b$ if there exists a universal constant $c>0$ such that $a \geq c b$. For $v,w \in \mathbb{R}^p$, we denote $\ell_1$, $\ell_2$, and $\ell_\infty$ norm as $\norm{v}_1 = \sum_{i=1}^n |v_i| $, $\norm{v}_2 = \sqrt{v^Tv}$, and  $\norm{v}_\infty = \sup_j |v_j|$ and use $v \circ w \in \mathbb{R}^p$ to denote Hadamard product (entry-wise product) of $v,w$. For a set $S$, we use $|S|$ to denote the cardinality of $S$. For any subset $S \subseteq \{1,\dots,p\}$, $v_S \in \mathbb{R}^{|S|}$ denotes the sub-vector of the vector $v$ by selecting the components with indices in $S$. Likewise for matrix $A \in \mathbb{R}^{n \times p}$, $A_S \in \mathbb{R}^{n\times|S|}$ denotes a sub-matrix by selecting columns with indices in $S$. For a group $\ell_1/\ell_2$ norm, the norm is characterized by a partition $\mathcal{G}:= (g_1,\dots,g_J)$ of $\{1,\dots,p\}$ and associated weights $(w_j)_1^J$. We let $\mathscr{G}:= (\mathcal{G},(w_j)_1^J)$ and define the $\ell_1/\ell_2$ norm as $\norm{v}_{\mathscr{G},2,1} := \sum_j w_j \norm{v_{g_j}}_2$. We often need a dual norm of $\norm{\cdot}_{\mathscr{G},2,1} $. We use $\bar{\mathscr{G}}$ to denote $\bar{\mathscr{G}}:= (\mathcal{G},(w_j^{-1})_1^J)$ and write $\norm{v}_{\bar{\mathscr{G}},2,\infty} = \max_j w_j^{-1} \norm{v_{g_j}}_2$. Finally we write $\mathbb{B}_q(r,v)$ for an $\ell_q$ ball with radius $r$ centered at $v \in \mathbb{R}^p$, and denote as $\mathbb{B}_q(r)$ if $v=0$.

For a convex function $f:\mathbb{R}^p \rightarrow \mathbb{R}$, we use $\partial f(x) $ to denote the set of sub-gradients at the point $x$ and $ \triangledown f (x)$ to denote an element of  $\partial f(x) $.
Also for a function $f+g$ such that $f$ is differentiable (but not necessarily convex) and $g$ is convex, we define $\partial (f+g)(x) := \{\triangledown f (x) + h \in \mathbb{R}^p ; h \in \partial g(x)\} $ with a slight abuse of notation. Also, we say $f(n) = O(g(n))$, $f(n) = \Omega(g(n))$, and $f(n) = \Theta(g(n))$ if $|f|$ is asymptotically bounded above, bounded below, and bounded above and below by $g$.

For a random variable $x\in\mathbb{R}$, we say $x$ is a sub-Gaussian random variable with sub-Gaussian parameter $\sigma_x >0 $ if $E[\exp(t(x-E[x]))]\leq \exp(t^2 \sigma_x^2/2)$ for all $t \in \mathbb{R}$ and we denote as $x \sim \mbox{subG}(\sigma_x^2)$  with a slight abuse of notation. Similarly, we say $x$ is a sub-exponential random variable with sub-exponential parameter $(\nu,b) $ if $E[\exp(t(x-E[x]))]\leq \exp(t^2 \nu^2/2)$ for all $|t| \leq 1/b$ and we denote as $x \sim \mbox{subExp}(\nu,b)$. A collection of random variables $(x_1,\dots,x_n)$ is referred to as $x_1^n$.

\section{PUlasso algorithm}
\label{SecPU}

In this section, we introduce our PUlasso algorithm. First, we discuss the prior EM algorithm approach developed in~\cite{ward_presence-only_2009} and apply a simple regularization scheme. We then discuss our two computational speed-ups, the quadratic majorization for the M-step and exploiting sparse matrices. We prove that our algorithm has the descending property and converges to a stationary point, and show that our two speed-ups increase speed by several orders of magnitude.

\subsection{Prior approach: EM algorithm with regularization}
\label{SecPrev}

First we use the prior result in~\cite{ward_presence-only_2009} to determine the observed log-likelihood~(in terms of the $z_i$'s) and the full log-likelihood~(in terms of the unobserved $y_i$'s and $z_i$'s). The following lemma, derived in \cite{ward_presence-only_2009}, gives the form of the observed and the full log-likelihood in the case-control sampling scheme.

\begin{lemma}[\cite{ward_presence-only_2009}]\label{lem:2.1}
The observed log-likelihood $\log L(\theta;x_1^n,z_1^n) $ for our presence-only model in terms of $(x_i, z_i,s_i=1)_{i=1}^n$ is:
\begin{align}\label{lem2.1:logLikObs}
\log L(\theta;{x_1^n,z_1^n})  & = \log \left(\prod_i \mathbb{P}_{\theta}(z_i | x_i,s_i=1)\right)\nonumber\\
& = \sum_{i=1}^n \log  \left( \dfrac{\frac{n_l}{\pi n_u} e^{\theta^Tx}}{1+(1+\frac{n_l}{\pi n_u})e^{\theta^Tx}}\right)^{z_i}\left( \dfrac{1+ e^{\theta^Tx}}{1+(1+\frac{n_l}{\pi n_u})e^{\theta^Tx}}\right)^{1-z_i}
\end{align}
The full log-likelihood $\log L_f(\theta;x_1^n,y_1^n,z_1^n) $ in terms of $(x_i, y_i, z_i,s_i=1)_{i=1}^n$ is
\begin{align}\label{lem2.1:logLikFull}
\log L_f(\theta;{ x_1^n,y_1^n,z_1^n}) &= \log  \left(\prod_i \mathbb{P}_{\theta}(y_i,z_i|x_i,s_i=1)\right) \nonumber \\
& \propto \sum_{i=1}^n [y_i(x_i^T \theta +\log\dfrac{n_\ell+\pi n_u}{\pi n_u}) - \log(1+\exp(x_i^T\theta+\log\dfrac{n_\ell+\pi n_u}{\pi n_u}))]
\end{align}
where $ n_\ell, n_u $ are the number of positive and unlabeled observations, $n = n_\ell +n_u$  and $ \pi $ is defined in \eqref{def:pi}.
\end{lemma}
\noindent The proof can be found in \cite{ward_presence-only_2009}. Our goal is to estimate the parameter  $\theta^*:=\argmin_{\theta  \in \mathbb{R}^p} E[-\log L(\theta; x_1^n,z_1^n)]$, which we assume to be unique. In the setting where $p$ is large, we add a regularization term. We are interested in cases when there exists or does not exist a group structure within covariates. To be general we use the group $\ell_1/\ell_2$-penalty for which $\ell_1$ is a special case. Hence our overall optimization problem is:
\begin{equation}\label{eq:objective}
 \underset{\theta}{\text{minimize}}
\qquad  -\dfrac{1}{n}\sum_{i=1}^{n} \log L(\theta;x_i,z_i) +P_\lambda(\theta)
\end{equation}
where $ \log L(\theta; x_i, z_i) $ is the observed log-likelihood. For a penalty term, we use the group sparsity regularizer
\begin{equation}\label{def:pen}
P_\lambda(\theta) := \lambda \norm{\theta}_{\mathscr{G},2,1} = \lambda \sum_{j=1}^{J}w_j\norm{\theta_{g_j}}_2
\end{equation}
with $\mathscr{G}= (\mathcal{G},(w_j)_{j=1}^J)$, such that $\mathcal{G}:=(g_1,...,g_J)$ is a partition of $(1,\dots,p)$ and $w_j>0$. We note that $\norm{\theta}_{\mathscr{G},2,1} = \norm{\theta}_{1}$ if $J = p$, $g_j = \{j\}$ and $w_j =1$, $\forall j$. For notational convenience we denote the overall objective  $\mathscr{F}_n(\theta)$ as
\begin{align}\label{eq:objfunL}
\mathscr{F}_n(\theta) &:= -\dfrac{1}{n}\sum_{i=1}^{n} \log L(\theta;x_i,z_i) +P_\lambda(\theta)= \mathscr{L}_n(\theta)+P_\lambda(\theta)
\end{align}
where we define the loss function $\mathscr{L}_n(\theta)$ as $\mathscr{L}_n(\theta) := -n^{-1}\sum_{i=1}^{n} \log L(\theta;x_i,z_i) $ and $P_\lambda(\theta) = \lambda \norm{\theta}_{\mathscr{G},2,1} = \lambda \sum_{j=1}^{J}w_j\norm{\theta_{g_j}}_2$.

In the original proposal of the group lasso, \cite{yuan_model_2006} recommended to use \eqref{def:pen} for orthonormal group matrices $ X_{g_j} $, i.e. $X_{g_j}^TX_{g_j}/n = I_{|g_j|\times|g_j|}$. If group matrices are not orthonormal however, it is unclear whether we should orthonormalize group matrices prior to application of the group lasso. This question was addressed in \cite{simon_standardization_2012}, and the authors provide a compelling argument that prior orthonormalization has both theoretical and computational advantages. In particular,~\cite{simon_standardization_2012} demonstrated that the following orthonormalization procedure is intimately connected with the uniformly most powerful invariant testing for inclusion of a group. To describe this orthonormalization explicitly, we obtain standardized group matrices $Q_{g_j} \in \mathbb{R}^{n \times |g_j|}$ and scale matrices $R_{g_j} \in \mathbb{R}^{|g_j| \times |g_j|}$ for $j\geq 2$ using the QR-decomposition such that 
\begin{equation}\label{eq:PXQR}
P_0 X_{g_j} = Q_{g_j} R_{g_j} \text{ and }  Q_{g_j}^T Q_{g_j} = n I_{|g_j| \times |g_j|} 
\end{equation}
where $ P_0 = (I_{n \times n} - \frac{\mathbbm{1}_n\mathbbm{1}_n^T}{n})$ is the projection matrix onto the orthogonal space of $ \mathbbm{1}_n $. Letting $ Q:= [\mathbbm{1}_n,Q_{g_2},\dots,Q_{g_J}] = [q_1^T,\dots,q_n^T]$, the original optimization problem~\eqref{eq:objective} can be expressed in terms of $q_i$'s and becomes:
\begin{equation}\label{eq:newobjective}
\argmin_\nu \left\lbrace-\dfrac{1}{n}\sum_{i=1}^{n} \log L(\nu;q_i,z_i) + \lambda\sum_{j=1}^J w_j\norm{\nu_{g_j}}_2\right\rbrace 
\end{equation}
where we use the transformation $\theta$ to $\nu$:
\begin{equation}\label{eq:transformation}
\theta_{g_j} = 
\begin{cases}
\nu_1 -\sum_{j=2}^J \frac{\mathbbm{1}_n^T}{n} X_{g_j} R_{g_j}^{-1}\nu_{g_j} & j=1\\
R_{g_j}^{-1} \nu_{g_j} & j\geq 2.\\
\end{cases}
\end{equation}
We note that this corresponds to the standard centering and scaling of the predictors in the case of standard lasso. For more discussion about group lasso and standardization, see e.g. \cite{huang_selective_2012}.

A standard approach to performing this minimization is to use the EM-algorithm approach developed in~\cite{ward_presence-only_2009}. In particular we treat $y_1^n$ as hidden variables and estimate them in the E-step. Then use estimated $\hat{y}_1^n$ to obtain the full log-likelihood $\log L_f(\theta;x_1^n,\hat{y}_1^n,z_1^n)$ in the M-step.

\scalebox{0.85}{
\begin{algorithm}[H]\label{alg:em}
\SetAlgoLined
Input: an initialization $ \theta^0$ such that $\mathscr{F}_n(\theta^0)\leq \mathscr{F}_n(\theta_{null})$ \\
	\For{m=0,1,2,\dots,}{
		\begin{itemize}
			\item E-step : estimate $ y_i $ at $ \theta = \theta^m $ by			\begin{equation}\label{alg1:E_step}
			\hat{y_i}(\theta^m)  =  \left( \dfrac{e^{x_i^T\theta^m}}{1+e^{x_i^T\theta^m}}\right)^{1-z_i}
			\end{equation}
			\item M-step : obtain $ \theta^{m+1} $ by
			\begin{equation}\label{alg1:M_step}
			\theta^{m+1} \in  \argmin_\theta \left\lbrace-\dfrac{1}{n}\sum_{i=1}^n \left(\hat{y_i}(\theta^m)\left(x_i^T \theta +b\right)-\log(1+e^{x_i^T\theta+b})\right)+P_\lambda(\theta)\right\rbrace
			\end{equation}
			where  $ b := \log\dfrac{n_\ell+\pi n_u}{\pi n_u}  $
		\end{itemize}
	}
	
\caption{Regularized EM algorithm for the optimization problem \eqref{eq:objective}}
\end{algorithm}
}

The E-step follows from $E_{\theta^m}[y_i|z_i,x_i,s_i=1] = \left(\dfrac{e^{x_i^T\theta^m}}{1+e^{x_i^T\theta^m}} \right)^{1-z_i}$  since $ z_i=1  $ implies $ y_i=1  $ and when $ z_i=0 $, observations in the unlabeled data are random draws from the population. An initialization $\theta^0$ can be any $\mathbb{R}^p$ vector such that $\mathscr{F}_n(\theta^0)\leq \mathscr{F}_n(\theta_{null})$ where $\theta_{null}$ is the parameter corresponding to the intercept-only model. 
If we are provided with no additional information, we may use $\theta_{null}$ for the initialization. We use $\theta^0 = \theta_{null}$ as the initialization for the remainder of the paper. For the M-step it was originally proposed to use a logistic regression solver. We can use a regularized logistic regression solver such as the {\bf  glmnet} R  package to solve \eqref{alg1:M_step}. We discuss a computationally more efficient way of solving~\eqref{alg1:M_step} in the subsequent section.

\subsection{PUlasso : A Quadratic Majorization for the M-step}
Now we develop our PUlasso algorithm which is a faster algorithm for solving \eqref{eq:objective} by using quadratic majorization for the M-step. The main computational bottleneck in algorithm \ref{alg:em} is the M-step which requires minimizing a regularized logistic regression loss at each step. This sub-problem does not have a closed-form solution and needs to be solved iteratively, causing inefficiency in the algorithm. However the most important property of  the objective function in the M-step is that it is a surrogate function of the likelihood which ensures the descending property~(see e.g.~\cite{lange_optimization_2000}). Hence we replace a logistic loss function with a computationally faster quadratic surrogate function. In this aspect, our approach is an example of the more general majorization-minimization~(MM) framework~(see e.g.~\cite{lange_optimization_2000,Ortega}). 

On the other hand, our loss function itself belongs to a generalized linear model family, as we will discuss in more detail in the subsequent section. A number of works have developed methods for efficiently solving regularized generalized linear model problems. A standard approach is to make a quadratic approximation of the log-likelihood and use solvers for a regularized least-square problem. Works include using an exact Hessian~(\cite{lee_efficient_2006, friedman_regularization_2010}), an approximate Hessian~(\cite{meier_group_2008}) or a Hessian bound~(\cite{krishnapuram_sparse_2005,simon_standardization_2012,breheny_group_2013}) for the second order term. Solving a second-order approximation problem amounts to taking a Newton step, thus convergence is not guaranteed without a step-size optimization~(\cite{lee_efficient_2006,meier_group_2008}), unless a global bound of the Hessian matrix is used. Our work can be viewed as in the line of these works where a quadratic approximation of the loss function is made and then an upper bound of the Hessian matrix is used to preserve a majorization property.

A coordinate descent~(CD) algorithm~(\cite{wu_coordinate_2008,friedman_regularization_2010}) or a block coordinate descent~(BCD) algorithm~(\cite{yuan_model_2006,puig_multidimensional_2011,simon_standardization_2012, breheny_group_2013}) has been a very efficient and standard way to solve a quadratic problem with $\ell_1$ penalty or $\ell_1/\ell_2$ penalty and we also take this approach. When a feature matrix $X \in \mathbb{R}^{n\times p}$ is sparse, we can set up the algorithm to exploit such sparsity through a sparse linear algebra calculation. We discuss this implementation strategy in the Section~\ref{subsec:sparseBCD}.

Now we discuss the PUlasso algorithm and the construction of quadratic surrogate functions in more details. Using the MM framework we construct the set of majorization functions $ -\overline{Q}(\theta; \theta^m) $ with the following two properties:
\begin{equation}\label{eq:surrogateProperties}
\overline{Q}(\theta^m;\theta^m) = Q(\theta^m;\theta^m), \quad \overline{Q}(\theta;\theta^m) \leq Q(\theta;\theta^m), \forall \theta
\end{equation}
where our goal is to minimize $ -Q $ where $ Q(\theta ;\theta^m) := n^{-1}E_{\theta^m}[\log L_f (\theta)|{ z_1^n,x_1^n,s_1^n}=1] $.

Using the Taylor expansion of $ Q(\theta;\theta^m) $ at $\theta = \theta^m$, we obtain $Q(\theta;\theta^m)$
\begin{align*}
&=Q(\theta^m;\theta^m)+\dfrac{1}{n}[X^T(\hat{y}(\theta^m)-\mu^*(\theta^m))]^T\Delta_m- \frac{1}{2n}\int_0^1\Delta_m^TX^TW(\theta +s\Delta_m)X\Delta_m ds\\
&\geq Q(\theta^m;\theta^m)+\dfrac{1}{n}(\hat{y}(\theta^m)-\mu^*(\theta^m))^TX\Delta_m - \frac{1}{8n}\Delta_m^TX^TX\Delta_m
\end{align*}
where we define $\Delta_m := \theta - \theta^m$, $\mu^*(\theta^m)_{i} :=\dfrac{e^{x_i^T\theta^m + b}}{1+e^{x_i^T\theta^m + b}}$, $b := \log\dfrac{n_\ell+\pi n_u}{\pi n_u}$ and $W \in \mathbb{R}^{n \times n}$ is a diagonal matrix with $[W(\theta)]_{ii}: = \mu^*(\theta)_i(1-\mu^*(\theta)_i)$. The inequality follows from $W(\theta) \prec \frac{1}{4} I_{n \times n},\; \forall\; \theta$. Thus setting $ \overline{Q} $ as follows:
\begin{equation*}
\overline{Q}(\theta;\theta^m) := Q(\theta^m;\theta^m)+\dfrac{1}{n}(\hat{y}(\theta^m)-\mu^*(\theta^m))^T(X\theta - X\theta^m) - \frac{1}{8n}(\theta-\theta^m)^TX^TX(\theta-\theta^m),
\end{equation*}
$\overline{Q}$ satisfies both conditions in~\eqref{eq:surrogateProperties}. Also with some algebra, it follows that 
\begin{equation*}
\overline{Q}(\theta;\theta^m) = -\frac{1}{8n}(4(\hat{y}(\theta^m)-\mu^*(\theta^m))+X\theta^m - X\theta)^T(4(\hat{y}(\theta^m)-\mu^*(\theta^m))+X\theta^m - X\theta)+c(\theta^m)
\end{equation*}
for some $ c(\theta^m) $ which does not depend on $ \theta $. Hence $ -\overline{Q}$ acts as a quadratic surrogate function of $-Q$ which replaces our M-step for the original EM algorithm. Therefore our PUlasso algorithm can be represented as follows.

\scalebox{0.85}{
\begin{algorithm}[H]\label{alg:mm}
\SetAlgoLined
Input: an initialization $ \theta^0$ such that $\mathscr{F}_n(\theta^0)\leq \mathscr{F}_n(\theta_{null})$ \\
\For{m=0,1,2,\dots,}{
		\begin{itemize}
			\item E-step : estimate $ y_i $ at $ \theta = \theta^m $ by
			\begin{equation}\label{alg2:E_step}
			\hat{y_i}(\theta^m)  =  \left( \dfrac{e^{x_i^T\theta^m}}{1+e^{x_i^T\theta^m}}\right)^{1-z_i}
			\end{equation}
			\item QM-EM step : obtain $ \theta^{m+1} $ by
			\begin{enumerate}
				\item create a working response vector $ u(\theta^m) $ at $ \theta= \theta^m $
				\begin{equation}\label{alg2:M_step}
				u(\theta^m) := 4(\hat{y}(\theta^m)-\mu^*(\theta^m))+X\theta^m
				\end{equation}
				\item solve a quadratic loss problem with a penalty
				\begin{equation}\label{alg2:QLprblm}
				\theta^{m+1} \in  \argmin_\theta\left\lbrace\dfrac{1}{2n}(u (\theta^m)-X\theta)^T(u (\theta^m)-X\theta)+4 P_\lambda(\theta)\right\rbrace
				\end{equation}
			\end{enumerate}
		\end{itemize}
	}
\caption{PUlasso : QM-EM algorithm for the optimization problem \eqref{eq:objective}}
\end{algorithm}
}

Now we state the following proposition to show that both the regularized EM and PUlasso algorithms have the desirable descending property and converge to a stationary point. For convenience we define the feasible region $\widetilde{\Theta_0}$, which contains all $\theta$ whose objective function value is better than that of the intercept-only model, defined as:
\begin{equation}\label{def:Tilde_Theta0}
	\widetilde{\Theta_0} := \{ \theta \in \mathbb{R}^p ; \mathscr{F}_n(\theta) \leq \mathscr{F}_n(\theta_{null})\}
\end{equation}
where $\theta_{null} = [\log\frac{\pi}{1-\pi},0,...,0]^T$, an estimate corresponding to the intercept-only model. We let $ \mathscr{S} $ be the set of stationary points satisfying the first order optimality condition, i.e., 
\begin{equation}\label{eq:stationaryset}
\mathscr{S} := \{\theta; \exists \triangledown \mathscr{F}_n(\theta) \in \partial \mathscr{F}_n(\theta) \mbox{ such that } \triangledown \mathscr{F}_n(\theta)^T(\theta'-\theta) \geq 0, \forall\; \theta' \in \widetilde{\Theta_0}\}.
\end{equation}
\noindent One of the important conditions is to ensure that all iterates of our algorithm lie in $\widetilde{\Theta_0}$ which is trivially satisfied if $\theta^0 = \theta_{null}$.

\begin{proposition}\label{prop:convergence}
The sequence of estimates $(\theta^m) $ obtained by Algorithms~\ref{alg:em} or~\ref{alg:mm} satisfies 
\begin{itemize}\setlength\itemsep{0em}
\item[(i)] $ \mathscr{F}_n(\theta^m) \geq \mathscr{F}_n(\theta^{m+1})  $, and $ \mathscr{F}_n(\theta^m) > \mathscr{F}_n(\theta^{m+1})  $  if $ \theta^m \not\in \mathscr{S} $. 
\item[(ii)] All limit points of $(\theta^m)_1^{\infty}$ are elements of the set $\mathscr{S}$, and $ \mathscr{F}_n(\theta^m) $ converges monotonically to $ \mathscr{F}_n(\widetilde{\theta}) $ for some $ \widetilde{\theta} \in \mathscr{S}$.
\item[(iii)] The sequence $(\theta^m)$ has at least one limit point, which must be a stationary point of $ \mathscr{F}_n(\theta) $ by (ii).
\end{itemize}
\end{proposition}
Proposition~\ref{prop:convergence} shows that we obtain a stationary point of the objective~\eqref{eq:objfunL} as an output of both the regularized EM algorithm and our PUlasso algorithm. The proof uses the standard arguments based on Jensen's inequality, convergence of EM algorithm and MM algorithms and is deferred to the supplement S1.1.

\subsubsection{Block Coordinate Descent Algorithm for M-step and Sparse Calculation}\label{subsec:sparseBCD}

In this section, we discuss the specifics of finding a minimizer for the M-step~\eqref{alg2:QLprblm} for each iteration of our PUlasso algorithm. After pre-processing the design matrix as described in~\eqref{eq:newobjective}, \eqref{eq:transformation}, we solve the following optimization problem using a standard block-wise coordinate descent algorithm. 
\begin{equation}\label{eq:newQLprblm}
\argmin_\nu \left\lbrace \dfrac{1}{2n}\norm{u-Q\nu}_2^2+ 4\lambda\sum_{j=1}^J w_j\norm{\nu_{g_j}}_2\right\rbrace 
\end{equation}
\scalebox{0.85}{
\begin{algorithm}[H]\label{alg3:BCD}\abovedisplayskip=-20pt
Given initial parameter $ \nu = [\nu_1,\nu_{g_2}^T\dots,\nu_{g_J}^T]^T $, a residual vector $r = u - \sum_{j=1}^{J}Q_{g_j} \nu_{g_j}$ \\
\For{j=1}{update $\nu_1$ and $r$ using \eqref{alg3:BCD-1}-\eqref{alg3:BCD-3} }
\Repeat{convergence}{
	\For{j=2,\dots,J}{
		\begin{align}
		z_j &= n^{-1}Q_{g_j}^Tr + \nu_{g_j}\label{alg3:BCD-1}\\
		\nu_{g_j} ' &\leftarrow S(z_j, 4\lambda w_j)\label{alg3:BCD-2}\\
		r' & \leftarrow r + Q_{g_j}(\nu_{g_j} - \nu_{g_j}')\label{alg3:BCD-3}\\
		r &\leftarrow r' , \nu_{g_j} \leftarrow \nu_{g_j}'\nonumber
		\end{align}}}
\caption{Fitting \eqref{eq:newQLprblm} using Block Coordinate Descent}
\end{algorithm}}

$S(.,\lambda)$ is the soft thresholding operator defined as follows:
\[S(z,\lambda) := 
\begin{cases}
(\norm{z}_2 -\lambda)\dfrac{z}{\norm{z}_2} & \text{ if } \norm{z}_2 > \lambda \\
0&\text{otherwise.} \\
\end{cases}
\]
Note that we do not need to keep updating the intercept $\nu_1$ since $ Q_{g_j} , j\geq 2$ are orthogonal to $ Q_{g_1} \equiv \mathbbm{1}_n $. For more details, see e.g. \cite{breheny_group_2013}.

For our biochemistry example and many other examples, $X$ is a sparse matrix since each entry is an indicator of whether an amino acid is in a position. In Algorithm~\ref{alg3:BCD},  we do not exploit this sparsity since $Q$ will not be sparse even when $X$ is sparse. If we want to exploit sparse $X$ we use the following algorithm.\\

\scalebox{0.85}{
\begin{algorithm}[H]\label{alg4:s-BCD}\abovedisplayskip=-20pt
	Given initial parameter $ \nu = [\nu_1,\nu_{g_2}^T\dots,\nu_{g_J}^T]^T $, $r = u - P_0(\sum_{j=1}^{J}X_{g_j}R_{g_j}^{-1} \nu_{g_j})$ \\
\For{j=1}{update $\nu_1$ and $r$ using \eqref{alg3:BCD-1}-\eqref{alg3:BCD-3}.}

	\Repeat{convergence}{
		\For{j=2,\dots,J}{
			\begin{align}
			z_j &= n^{-1} R_{g_j}^{-1}X_{g_j}^Tr -  R_{g_j}^{-1}\left(X_{g_j}^T\mathbbm{1}_n/n\right)\left(\mathbbm{1}_n^Tr/n\right)+ \nu_{g_j}\label{alg4:sBCD-1}\\
			\nu_{g_j} ' &\leftarrow S(z_j, 4\lambda w_j)\label{alg4:sBCD-2}\\
			r' & \leftarrow r+ X_{g_j}R_{g_j}^{-1}(\nu_{g_j} - \nu_{g_j}')\label{alg4:sBCD-3}\\
			a_j  &\leftarrow  \mathbbm{1}_n^TX_{g_j}R_{g_j}^{-1}(\nu_{g_j} - \nu_{g_j}')/n\label{alg4:sBCD-4}\\
			r &\leftarrow r' , \nu_{g_j} \leftarrow \nu_{g_j}'\nonumber
			\end{align}}
		\begin{equation}\label{alg4:sBCD-5}
		r \leftarrow r - (\sum_{j=2}^{J}a_j )\mathbbm{1}_n
		\end{equation}
		
	}
	\caption{Fitting \eqref{eq:newQLprblm} and exploiting  sparse X}
\end{algorithm}
}

To explain the changes to this algorithm, we modify \eqref{alg3:BCD-1} and \eqref{alg3:BCD-3} so that we directly use $X$ rather than $Q$ to exploit the sparsity of $X$. Using \eqref{eq:PXQR}, we first substitute $Q_{g_j}$ with $P_0X_{g_j} R_{g_j}^{-1}$ to obtain
\begin{align}
z_j &= n^{-1} R_{g_j}^{-1}X_{g_j}^TP_0r+ \nu_{g_j}\label{alg4:deriv1}\\
r' & \leftarrow r + P_0X_{g_j} R_{g_j}^{-1}(\nu_{g_j} - \nu_{g_j}').\label{alg4:deriv2}
\end{align}
However carrying out \eqref{alg4:deriv1}-\eqref{alg4:deriv2} instead of \eqref{alg3:BCD-1}-\eqref{alg3:BCD-3} incurs a greater computational cost. Calculating $Q_{g_j}^Tr$ requires $n|g_j|$ operations. On the contrary, the minimal number of operations required to do a matrix multiplication of $R_{g_j}^{-1}X_{g_j}^TP_0r$ is $n^2+n|g_j|+|g_j|^2$, when it is parenthesized as $R_{g_j}^{-1}(X_{g_j}^T(P_0r))$. In many cases $|g_j|$ is small~(for standard lasso, $|g_j|=1, \forall j$ and for our biochemistry example, $|g_j|$ is at most 20), but the additional increase in $n$ can be very costly (especially in our example where $n$ is over 4 million).

For a more efficient calculation, we first exploit the structure of $ P_0 = I_{n \times n} - \frac{\mathbbm{1}_n\mathbbm{1}_n^T}{n}$ when multiplying $P_0 $ with a vector, which reduces the cost from $n^2$ operations to $2n$ operations. Also, we carry out calculations using $X_{g_j}$ instead of $P_0X_{g_j}$ when calculating residuals and do the corrections all at once. 

Before going into detail about \eqref{alg4:sBCD-1}-\eqref{alg4:sBCD-4}, we first discuss the computational complexity. Comparing \eqref{alg4:sBCD-1} with \eqref{alg3:BCD-1}, the first term only requires an additional $|g_j|^2$ operations. The second term $(X_{g_j}^T\mathbbm{1}_n)/n$ can be stored during the initial QR decomposition; thus the only potentially expensive operation is calculating an average of $r$ which requires $n$ operations.
Comparing \eqref{alg4:sBCD-3} with \eqref{alg3:BCD-3}, only $|g_j|^2$ additional operations are needed when we parenthesize as $X_{g_j}(R_{g_j}^{-1}(\nu_{g_j} - \nu_{g_j}'))$. Note that if we had kept $P_0$, there would have been an additional $2n$ operations even though we had used the structure of $P_0$.
In the calculation of \eqref{alg4:sBCD-5}, we note that $n$ operations are involved in subtracting $\sum_{j=2}^{J}a_j$ from $r$ because $ a_j $ are scalars. In summary, we essentially reduce additional computational cost from $O(n^2)$ to $nJ$ per cycle by carrying out \eqref{alg4:sBCD-1}-\eqref{alg4:sBCD-4} instead of \eqref{alg4:deriv1}-\eqref{alg4:deriv2}.

Now we derive/explain the formulas in Algorithm~\ref{alg4:s-BCD}. To make quantities more explicit, we use $r_j$ and $r_j'$ to denote a residual vector before/after update at $j$ using Algorithm \ref{alg3:BCD} and $\tilde{r}_j$ and $\tilde{r}_j'$ using Algorithm \ref{alg4:s-BCD}. By definition, $r_{j+1} = r'_{j}$ and $\tilde{r}_{j+1} = \tilde{r}'_{j}$. Also we note that in the beginning of the cycle $r_2 =\tilde{r}_2$. Equation \eqref{alg4:sBCD-1} can be obtained from \eqref{alg4:deriv1} by replacing $P_0$ with $I_{n \times n} - \frac{\mathbbm{1}_n\mathbbm{1}_n^T}{n}$. Now we show that modified residuals still correctly update coefficients. Starting from $j=2$, a calculated residual $\tilde{r}_{j}'$ is a constant vector off from a correct residual $r_j'$, as we see below:
\begin{align}
r_{j}^{'} &= r_j+P_0X_{g_j}R_{g_j}^{-1}(\nu_{g_j}- \nu_{g_j}')\label{alg4:deriv3}\\
&=r_j+X_{g_j}R_{g_j}^{-1}(\nu_{g_j} - \nu_{g_j}')- \mathbbm{1}_n\frac{\mathbbm{1}_n^T}{n}X_{g_j}R_{g_j}^{-1}(\nu_{g_j} - \nu_{g_j}')\label{alg4:deriv4}\\
&=\tilde{r}_j'- \mathbbm{1}_n a_j\label{alg4:deriv5}
\end{align}
where we recall that $a_j = \frac{\mathbbm{1}_n^T}{n}X_{g_j}R_{g_j}^{-1}(\nu_{g_j} - \nu_{g_j}')$. We note $P_0 r_j' = P_0 \tilde{r}_j'$ because $P_0 \mathbbm{1}_n = 0$. Then the next $z_{j+1}$, thus new $\nu_{g_{j+1}}$, are still correctly calculated since
\begin{align}
z_{j+1} =  n^{-1} R_{g_{j+1}}^{-1}X_{g_{j+1}}^T P_0r_{j+1}+ \nu_{g_{j+1}} =  n^{-1} R_{g_{j+1}}^{-1}X_{g_{j+1}}^T P_0\tilde{r}_{j+1}+ \nu_{g_{j+1}}. 
\end{align}

The next residual $\tilde{r}_{j+1}'$ is again off by a constant from the correct residual $r_{j+1}'$. To see this, $r_{j+1}' = r_{j+1} + P_0X_{g_{j+1}}R_{g_{j+1}}^{-1}(\nu_{g_{j+1}}- \nu_{g_{j+1}}') = \tilde{r}_{j+1} + P_0X_{g_{j+1}}R_{g_{j+1}}^{-1}(\nu_{g_{j+1}}- \nu_{g_{j+1}}')-a_j \mathbbm{1}_n$ by \eqref{alg4:deriv2}. Going through \eqref{alg4:deriv3}-\eqref{alg4:deriv5} with $j$ being replaced by $j+1$, we obtain 
$$ r_{j+1}' = \tilde{r}_{j+1}' -(a_j+a_{j+1}) \mathbbm{1}_n.$$
Inductively, we have correct $z_j$, thus $\nu_{g_j}$ for all $j\geq2$. At the end of the cycle, we correct the residual vector all at once by letting $r \leftarrow r - (\sum_{j=2}^{J}a_j )\mathbbm{1}_n$. 

\subsection{R Package details}

We provide a publicly available R implementation of our algorithm in the {\bf PUlasso} package. For a fast and efficient implementation, all underlying computation is implemented in C++. The package uses warm start and strong rule~(\cite{friedman_pathwise_2007,tibshirani_strong_2012}), and a cross-validation function is provided as well for the selection of the regularization parameter $\lambda$. Our package supports a parallel computation through the R package {\bf parallel}.

\subsection{Run-time improvement}

Now we illustrate the run-time improvements for our two speed-ups. Note that we only include $p$ up to $100$ so that we can compare to the original regularized EM algorithm. For our biochemistry application $p = O(10^4)$ and $n = O(10^6)$ which means the regularized EM algorithm is too slow to run efficiently. Hence we use smaller values of $n$ and $p$ in our run-time comparison. It is clear from our results that the quadratic majorization step is several orders of magnitude faster than the original EM algorithm, and exploiting the sparsity of $X$ provides a further $30 \%$ speed-up.

\begin{table}[htbp]
\centering
\scalebox{0.9}{
\begin{tabular}{ccccc}
	\hline
	&(n,p) & PUlasso & EM & time reduction(\%) \\ 
	\hline
	Dense matrix & n=1000, p=10 & 0.94 & 443.72 & 99.79 \\ 
	&n=5000, p=50 & 2.52 & 1844.98 & 99.86 \\ 
	&n=10000, p=100 & 9.45 & 5066.86 & 99.81 \\ 
	Sparse matrix &n=1000, p=10 & 0.40 & 196.86 & 99.80 \\ 
	&n=5000, p=50 & 2.01 & 614.65 & 99.67 \\ 
	&n=10000, p=100 & 4.29 & 1201.09 & 99.64 \\ 
	\hline
\end{tabular}}
\caption{Timings (in seconds). Sparsity level in $ X $ = 0.95,  $n_\ell/n_u = 0.5$. Total time for $100 \, \lambda$ values, averaged over 3 runs.}
\end{table}
\begin{table}[htbp]
\centering
\scalebox{0.9}{
\begin{tabular}{cccc}
	\hline
	(n,p) & sparse calculation & dense calculation & time reduction(\%) \\ 
	\hline
	n=10000, p=100 & 12.91 & 19.24 & 32.89 \\ 
	n=30000, p=100 & 25.64 & 38.73 & 33.79 \\ 
	n=50000, p=100 & 39.47 & 57.18 & 30.97 \\ 
	\hline
\end{tabular}}
	\caption{Timings (in seconds) using sparse and dense calculation for fitting the same simulated data. Sparsity level in X = 0.95, $n_\ell/n_u = 0.5$. Total time for $ 100 \, \lambda $ values, averaged over 3 runs.}
\end{table}	

\section{Statistical Guarantee}
\label{SecStat}

We now turn our attention to statistical guarantees for our PUlasso algorithm under the statistical model~\eqref{eq:py|x}. In particular we provide error bounds for any stationary point of the non-convex optimization problem \eqref{eq:objective}. Proposition~\ref{prop:convergence} guarantees that we obtain a stationary point from our PUlasso algorithm. 

 We first note that the observed likelihood~\eqref{lem2.1:logLikObs} is a generalized linear model (GLM) with a non-canonical link function. To see this, we rewrite the observed likelihood \eqref{lem2.1:logLikObs} as
 \begin{equation}\label{eq:lik_in_glm}
 	L(\theta;{x_1^n,z_1^n}) =\prod_{i=1}^{n}  \exp\left(z_i\eta_i -A(\eta_i)\right)
 \end{equation}
 after some algebraic manipulations, where we define $\eta_i := \log(n_\ell/\pi n_u)+x_i^T\theta-\log(1+e^{x_i^T\theta})$ and $A(\eta_i) := \log (1+e^{\eta_i})$. Also, we let $\mu(\eta_i) := A'(\eta_i)$, which is the conditional mean of $z_i$ given $x_i$, by the property of exponential families. For the convenience of the reader, we include the derivation from \eqref{lem2.1:logLikObs} to \eqref{eq:lik_in_glm} in the supplementary material S2.1. The mean of $z_i$ is related with $\theta^T x_i$ via the link function $g$ through $g(\mu(\eta_i)) = \theta^T x_i $, where $g$ satisfies $(g\circ\mu)^{-1}(\theta^T x_i) = \log(n_\ell/\pi n_u)+x_i^T\theta-\log(1+e^{x_i^T\theta})$. Because $(g\circ\mu)^{-1}$ is not the identity function, the likelihood is not convex anymore. For a more detailed discussion about the GLM with non-canonical link, see e.g. \cite{mccullagh_generalized_1989, fahrmeir_consistency_1985}.

 A number of works have been devoted to sparse estimation for generalized linear models. A large number of previous works have focused on generalized linear models with convex loss functions~(negative log-likelihood with a canonical link) plus $\ell_1$ or $\ell_1/\ell_2$ penalties. Results with the $\ell_1$ penalty include a risk consistency result (\cite{van_de_geer_high-dimensional_2008}) and estimation consistency in $\ell_2$ or $\ell_1$ norms~(\cite{kakade_learning_2010}). For a group-structured penalty, a probabilistic bound for the prediction error was given in \cite{meier_group_2008}. An $\ell_2$ estimation error bound in the case of the group lasso was given in \cite{blazere_oracle_2014}.

\cite{negahban_unified_2012} re-derived an $\ell_2$ error bound of an $\ell_1$-penalized GLM estimator under the unified framework for M-estimators with a convex loss function. This result about the regularized GLM was generalized in \cite{loh_regularized_2013} where penalty functions are allowed to be non-convex, while the same convex loss function  was used. Since the overall objective function is non-convex, authors discuss error bounds obtained for any stationary point, not a global minimum. In this aspect, our work closely follows this idea. However, our setting differs from \cite{loh_regularized_2013} in two aspects: first, the loss function in our setting is non-convex, in contrast with a convex loss function~(a negative log-likelihood with a canonical link) with non-convex regularizer in \cite{loh_regularized_2013}. Also, an additive penalty function was used in the work of \cite{loh_regularized_2013}, but we consider a group-structured penalty.

After the initial draft of this paper was written, we became aware of two recent papers (\cite{Elsener2018-tm,Mei2018-ec}) which studied non-convex M-estimation problems in various settings including binary linear classification, where the goal is to learn $\theta^*$ such that $E[z_i|x_i] = \sigma (x_i^T\theta^*)$ for a known $\sigma(\cdot)$. The proposed estimators are stationary points of the optimization problem: $\argmin_\theta n^{-1}\sum_{i=1}^n (z_i - \sigma(x_i^T\theta))^2 +\lambda \|\theta\|_1$ in both papers.
%
As the focus of our paper is to learn a model with a structural contamination in responses, 
our choice of mean and loss functions differ from both papers. In particular, our choice of mean function is different from the sigmoid function, which was the representative example of $\sigma(\cdot)$ in both papers, and we use the negative log-likelihood loss in contrast to the squared loss.
We establish error bounds by proving a modified restricted strong convexity condition, which will be discussed shortly, while error bounds of the same rates were established in~\cite{Elsener2018-tm} through a sharp oracle inequality, and a uniform convergence result over population risk in~\cite{Mei2018-ec}.

 Due to the non-convexity in the observed log-likelihood, we limit the feasible region $\Theta_0$ to  
\begin{equation}\label{def:Theta0}
	\Theta_0 := \{\theta\in \mathbb{R}^p; \norm{\theta}_2 \leq r_0, \norm{\theta}_{\mathscr{G},2,1} \leq R_n\}
\end{equation}
for theoretical convenience. Here $r_0,R_n>0$ must be chosen appropriately and we discuss these choices later. Similar restriction is also assumed in~\cite{loh_regularized_2013}.

\subsection{Assumptions}
We impose the following assumptions. First, we define a sub-Gaussian tail condition for a random vector $x \in \mathbb{R}^p$; we say $x$ has a sub-Gaussian tail with parameter $\sigma_x^2$, if for any fixed $v \in \mathbb{R}^p$, there exists $\sigma_x>0$ such that $E[\exp(t (x-E[x])^Tv)] \leq \exp(t^2\norm{v}_2^2 \sigma_x^2/2)$ for any $t \in \mathbb{R}$. We recall that $\theta^*$ is the true parameter vector, which minimizes the population loss.

\begin{assumption}\label{a1}
The rows $x_i \in \mathbb{R}^p$, $i = 1, 2, \dots , n$ of the design matrix are i.i.d. samples from a mean-zero distribution with sub-Gaussian tails with parameter $\sigma_x^2$. Moreover, $\Sigma_x := E[x_ix_i^T]$ is a positive definite and with minimum eigenvalue $ \lambda_{min}(\Sigma_x)\geq K_0$ where $K_0$ is a constant bounded away from $0$. We further assume that $(x_{ij})_{j \in g_j}$ are independent for all $j \in g_j$ and $g_j \in \mathcal{G}$. 
\end{assumption}

 Similar assumptions appear in for e.g. \cite{negahban_unified_2012}.
This restricted minimum eigenvalue condition~(see e.g.~\cite{RasWaiYu10b} for details) is satisfied for weakly correlated design matrices. We further assume independence across covariates \emph{within groups} since sub-Gaussian concentration bound assuming independence within groups is required.
\begin{assumption}\label{a2} 
For any $r>0$, there exists $K_1^r$ such that $\max_i |x_i^T \theta| \leq K_1^r $ a.s. for all $\theta$ in the set $\{\theta\;:\;\|\theta - \theta^*\|_2 \leq r \cap supp(\theta-\theta^*) \subseteq g_j$ for some $g_j \in \mathcal{G}\}$. 
\end{assumption}

Assumption~\ref{a2} ensures that $|x_i^T \theta^*|$ is bounded a.s., which guarantees that the underlying probability $(1+ e^{-x_i^T \theta^*})^{-1}$ is between $0$ and $1$, and $|x_i^T\theta|$ is also bounded within a compact sparse neighborhood of $\theta^*$ which ensures concentration to the population loss. Comparable assumptions are made in \cite{Elsener2018-tm,Mei2018-ec} where similar non-convex M estimation problems are investigated.

\begin{assumption}\label{a3}
The ratio of the number of labeled to unlabeled data , i.e. $ n_\ell/n_u $ is lower bounded away from 0 and upper bounded for all $ n = n_\ell + n_u$, as $n \rightarrow \infty$. Equivalently, there is a constant $K_2$ such that $| \log\left(n_\ell/\pi n_u \right)| \leq K_2$
\end{assumption}
Assumption~\ref{a3} ensures that the number of labeled samples $n_{\ell}$  is not too small or large relative to $n$. The reason why $n_{\ell}$ can not be too large is that the labeled samples are only positives and we need a reasonable number of negative samples which are a part of the unlabeled samples. 
\begin{assumption}[Rate conditions]\label{a4} 
We assume a high-dimensional regime where both $(n,p) \rightarrow \infty$ and $\log p = o(n)$.  For $\mathscr{G} = ((g_1,\dots,g_J), (w_j)_1^J))$ and $m:= \max_j |g_j|$,  we assume $J = \Omega(n^{\beta})$ for some $\beta>0$, $m = o (n \wedge J)$, $\min_j w_j = \Omega(1)$, and $\max_j w_j =o (n \wedge J) $.
\end{assumption}
Assumption \ref{a4} states standard rate conditions in a high-dimensional setting. In terms of the group structure, we assume that growth of $p$ is not totally attributed to the expansion of a few groups; the number of groups $J$ increases with $n$, and the maximum group size $m$ is of small order of both $n$ and $J$. Also we note that a typical choice of $w_j  = \sqrt{ |g_j|} $ satisfies Assumption 4 because $\min_j w_j \geq 1$, $\max_j w_j = \sqrt{m}$ and $\sqrt{m}/n , \sqrt{m}/J =o(1)$.

Finally we define the restricted strong convexity assumption for a loss function following the definition in~\cite{loh_regularized_2013}.

\begin{definition}[Restricted strong convexity] We say $ \mathscr{L}_n $ satisfies a \emph{restricted strong convexity}~(RSC) condition with respect to $\theta^*$ with curvature $ \alpha>0 $ and tolerance function $ \tau $ over $ \Theta_0 $ if the following inequality is satisfied for all $ \theta \in \Theta_0 $:

\begin{equation}\label{RSC}
\left(\triangledown \mathscr{L}_n(\theta) - \triangledown \mathscr{L}_n(\theta^*)\right)^T\Delta \geq \alpha \norm{\Delta}_2^2 -\tau(\norm{\Delta}_{\mathscr{G},2,1}) 
\end{equation}
where $ \Delta := \theta - \theta^* $ and $ \tau(\norm{\Delta}_{\mathscr{G},2,1}) = \tau_1 \norm{\Delta}_{\mathscr{G},2,1}^2 \dfrac{\log J+m}{n} + \tau_2 \norm{\Delta}_{\mathscr{G},2,1} \sqrt{\dfrac{\log J+m}{n}} $.
\end{definition}
In the special case where $\norm{\Delta}_{\mathscr{G},2,1} = \norm{\Delta}_{1}$ and hence $\tau(\norm{\Delta}_{1}) = \tau_1 \norm{\Delta}_{1}^2 \dfrac{\log p}{n} + \tau_2 \norm{\Delta}_{1} \sqrt{\dfrac{\log p}{n}} $, similar RSC conditions were discussed in~\cite{negahban_unified_2012} and \cite{loh_regularized_2013} with different $ \tau$ and $\Theta_0$. One of the important steps in our proof is to prove that RSC holds for the objective function $\mathscr{L}_n(\theta)$.

\subsection{Guarantee}	

Under Assumptions~\ref{a1}-\ref{a4}, we will show in Theorem \ref{thm:3.2} that the RSC condition holds with high probability over $\{\theta ; \norm{\theta}_2\leq r_0\} $ and therefore over $\Theta_0$, for $\Theta_0$ defined in \eqref{def:Theta0}. Under the RSC assumption, the following proposition, which is a modification of Theorem 1 in \cite{loh_regularized_2013}, provides $ \ell_1/\ell_2 $ and $ \ell_2 $ bounds of an error vector $ \hat{\Delta}:= \hat{\theta} - \theta^*$. Recall that $m = \max_j |g_j|$ (the size of the largest group) and $J$ is the number of groups.
\begin{proposition}\label{prop:3.1}
Suppose the empirical loss $ \mathscr{L}_n $ satisfies the RSC condition \eqref{RSC} with $ \tau(\norm{\Delta}_{\mathscr{G},2,1}) = \tau_1 \norm{\Delta}_{\mathscr{G},2,1}^2 \dfrac{\log J+m}{n} + \tau_2 \norm{\Delta}_{\mathscr{G},2,1} \sqrt{\dfrac{\log J+m}{n}} $ over $ \Theta_0 $ where $ \Theta_0 $ is feasible region for the objective \eqref{eq:objective}, as defined in \eqref{def:Theta0}, and the true parameter vector $\theta^*$ is feasible, i.e. $\theta^* \in \Theta_0$. Consider $ \lambda $ such that
\begin{equation}\label{prop3.1:lam_condition}
4 \max \left\lbrace \norm{\triangledown \mathscr{L}_n(\theta^*) }_{\bar{\mathscr{G}},2,\infty}, \left(\tau_1 \dfrac{2R_n (\log J+m)}{n}+\tau_2\sqrt{\dfrac{(\log J+m)}{n}}\right) \right\rbrace \leq \lambda.
\end{equation}
Let $ \hat{\theta} $ be a stationary point of \eqref{eq:objective}.
Then the following error bounds 
\begin{equation}\label{prop3.1:errbounds}
\norm{\hat{\Delta}}_2 \leq (\max_{j \in S} w_j)\dfrac{3\sqrt{s}\lambda}{2\alpha}\qquad
\text{and}	 \qquad \norm{\hat{\Delta}}_{\mathscr{G},1,2} \leq (\max_{j \in S} w_j)^2\dfrac{6s\lambda}{\alpha}, 
\end{equation} hold where $ S := \{j \in (1,\dots,J); \theta^*_{g_j} \neq 0\} $ and $s := |S|$.
\end{proposition}
The proof for Proposition \ref{prop:3.1} is deferred to the supplementary material S2.3. From~\eqref{prop3.1:errbounds}, we note the squared $\ell_2 $-error to grow proportionally with $s$ and $ \lambda^2 $. If $\theta^* \in \Theta_0$ and the choice of $\lambda = \Theta \left(\sqrt{\frac{\log J+m}{n}}\right)$ satisfies the inequality \eqref{prop3.1:lam_condition}, we obtain squared $\ell_2$ error which scales as $s\frac{\log J+m}{n}$, provided that the RSC condition holds over $\Theta_0$. In the case of lasso we recover $ \frac{s\log p}{n} $ parametric optimal rate since  $J=p, m=1$. 

With the choice of $r_0 \geq \|\theta^*\|_2 $ and $R_n = \Theta\left(\sqrt{\frac{n}{\log J+m}}\right)$\footnote{We note that the group $\ell_1$ constraint is active only if $\sqrt{\frac{n}{\log J+m}} = \mathcal{O}\left((\max_j w_j) r_0 \sqrt{J}\right)$. If $R_n \geq (\max_j w_j) r_0\sqrt{J} $, $\Theta_0 =  \{\theta ; \|\theta\|_2 \leq r_0, \|\theta\|_{\mathscr{G},2,1}\leq R_n\} \supseteq \{\theta ; \|\theta\|_2 \leq r_0, \|\theta\|_{\mathscr{G},2,1} \leq (\max_j w_j) r_0\sqrt{J}\}\supseteq \{\theta ; \|\theta\|_2 \leq r_0\}$ by the $\ell_1$-$\ell_2$ inequality, i.e. if $\|\theta\|_2 \leq r_0$, $\|\theta\|_{\mathscr{G},2,1} \leq (\max_j w_j) r_0\sqrt{J} $. The other direction is trivial, and thus $\Theta_0$ is reduced to $\Theta_0 = \{\theta ; \|\theta\|_2 \leq r_0\}$.}, we ensure $\theta^*$ is feasible and $\lambda = \Theta \left(\sqrt{\frac{\log J+m}{n}}\right)$ satisfies the inequality \eqref{prop3.1:lam_condition} with high probability. Clearly  $\left(\tau_1 \frac{2R_n (\log J+m)}{n}+\tau_2\sqrt{\frac{\log J+m}{n}}\right)$ is of the order $\sqrt{\frac{\log J +m}{n}} $ with the choice of $R_n =\Theta\left(\sqrt{\frac{n}{\log J+m}}\right) $, and following Lemma \ref{lem:3.1}, we have $\norm{\triangledown \mathscr{L}_n(\theta^*) }_{\bar{\mathscr{G}},2,\infty} = \mathcal{O}\left(\sqrt{\frac{\log J+m}{n}}\right)$ with high probability. Thus inequality \eqref{prop3.1:lam_condition} is satisfied with $\lambda = \Theta \left(\sqrt{\frac{\log J+m}{n}}\right)$ w.h.p. as well.

\begin{lemma}\label{lem:3.1}
Under Assumptions \ref{a1}-\ref{a4}, for any given $\epsilon>0$, there is a positive constant $c$ such that 
$$\mathbb{P}\left( \|\triangledown \mathscr{L}_n(\theta^*) \|_{\bar{\mathscr{G}},2,\infty} \geq c\sqrt{\dfrac{\log J+m}{n}} \right)\leq \epsilon$$
given a sample size $n\gtrsim (\log p + m) \vee (1/\epsilon)^{1/\beta}$.
\end{lemma}
The proof for Lemma~\ref{lem:3.1} is provided in the supplement S2.4. Now we state the main theorem of this section which shows that RSC condition holds uniformly over a neighborhood of the true parameter. 

\begin{theorem}\label{thm:3.2} For any given $r>0$ and $\epsilon>0$, there exist strictly positive constants  $\alpha, \tau_1 $ and $ \tau_2$  depending on $ \sigma_x, K_0,K_1^r$ and $K_2$ such that
\begin{equation}\label{thm3.3:main_ineq}
\left(\triangledown \mathscr{L}_n(\theta) - \triangledown \mathscr{L}_n(\theta^*)\right)^T\Delta \geq \alpha \|\Delta\|_2^2 - \tau_1 \|\Delta\|_{\mathscr{G},2,1}^2 \dfrac{\log J+m}{n} - \tau_2 \|\Delta\|_{\mathscr{G},2,1} \sqrt{\dfrac{\log J+m}{n}} 
\end{equation}
holds for all $\theta$ such that $\|\Delta\|_2:=\|\theta-\theta^*\|_2\leq r$ with probability at least $ 1-\epsilon $, given $(n,p) $ satisfying $n \gtrsim (\log J+m) \vee (1/\epsilon)^{1/\beta}$.
\end{theorem}
The proof of Theorem \ref{thm:3.2} is deferred to the supplement S2.5. There are a couple of notable remarks about Theorem~\ref{thm:3.2} and Proposition~\ref{prop:3.1}.
\begin{itemize}
\item The application of the Proposition \ref{prop:3.1} requires for a RSC condition to hold over a feasible region $\Theta_0$. Setting $r= 2r_0$ in Theorem \ref{thm:3.2}, inequality \eqref{thm3.3:main_ineq} holds over $\{ \theta ; \norm{\theta-\theta^*}_2 \leq 2r_0\} $ w.h.p, therefore over $\Theta_0 \subseteq \{ \theta ; \norm{\theta-\theta^*}_2 \leq 2r_0\}$.
\item We discuss how underlying parameters $r_0, \sigma_x$, and constants $K_0$-$K_2$ in Assumptions \ref{a1}-\ref{a3} are related to the $\ell_2$-error bound. From Proposition \ref{prop:3.1}, we see that $\ell_2$-error is proportional to  $\tau_1/\alpha$ and $\tau_2/\alpha $. The proof of Theorem~\ref{thm:3.2} reveals that 
$\tau_1/\alpha \lesssim (\sigma_x K_3/ K_0)^2$ and $\tau_2/\alpha \lesssim  \sigma_x(1+K_1^{2r_0})/K_0 L_0$, where $L_0$ and $K_3$ are also constants defined as $L_0 := \displaystyle \inf_{|u| \leq K_2 +K_1^{2r_0} +2r_0 K_3}(e^u/(1+e^u)^2 )(1+e^{K_1^{2r_0} +2r_0 K_3} )^{-2}$ and $K_3 \lesssim \sigma_x \log (\sigma_x^2/K_0)^{1/2}$.
As $L_0$ is inversely related to $K_2$ and $r_0$, $\ell_2$-error is proportional to the  $r_0$, $\sigma_x$, $K_1^{2r_0}$ and $K_2$ in Assumptions \ref{a2} and \ref{a3}, but inversely related to the minimum eigenvalue bound $K_0$ in Assumption \ref{a1}.

\item The mean-squared error $\frac{s \log p}{n}$ in the case of $J=p$ is verified below in Fig.~\ref{fig:mseplot} and both the mean-squared error and $\ell_1$ errors are minimax optimal for high-dimensional linear regression~(\cite{raskutti_minimax_2011}).

\end{itemize}

To validate the mean-squared error upper bound of $\frac{s \log p}{n}$ in Section~\ref{SecStat}, a synthetic dataset was generated according to the logistic model \eqref{eq:py|x} with $p = 500$ covariates and $X \sim N(0, I_{500\times 500})$. Varying $s$ and $n$ were considered to study the rate of convergence of $\norm{\hat{\theta}-\theta^*}_2$. The ratio $n_\ell/n_u$ was fixed to be $1$. For each dataset, $\hat{\theta}$ was obtained by applying PUlasso algorithm with a lambda sequence $\lambda_n := c_s \sqrt{\frac{\log p}{n}}$ for a suitably chosen $c_s$ for each $s$. We repeated the experiment 100 times and average $\ell_2-$error was calculated.
\begin{figure}
    \centering
    \includegraphics[width=.4\linewidth]{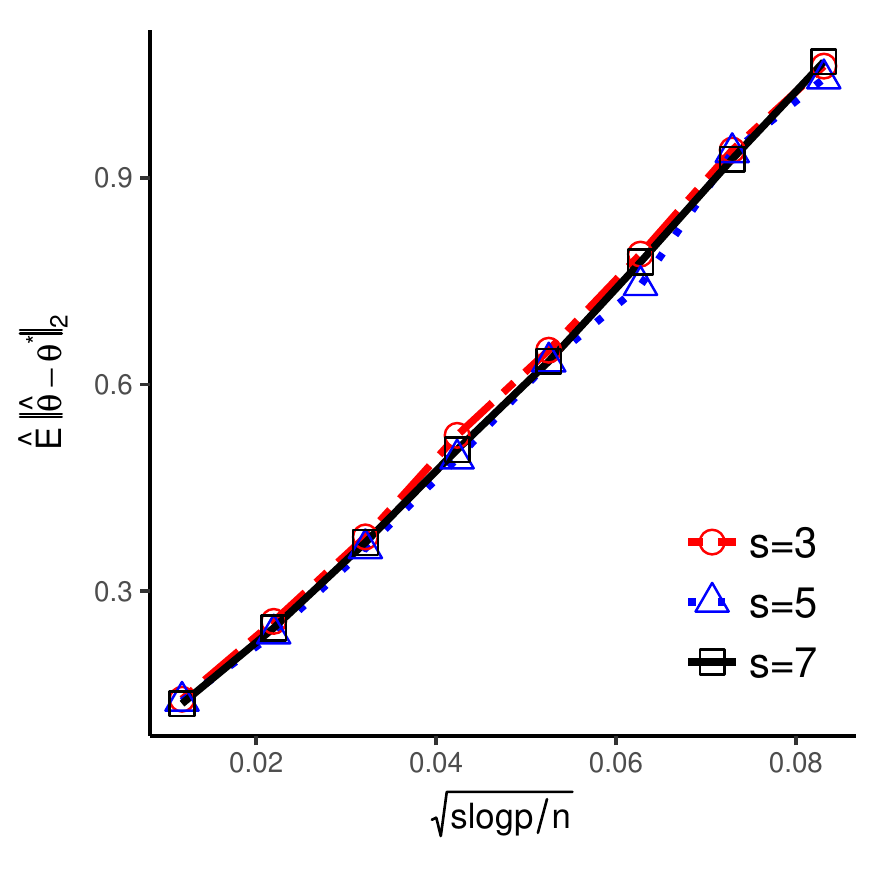}
    \caption{$\hat{E}[\norm{\hat{\theta}-\theta}_2]$  plotted against $\sqrt{s \log p/n}$ with fixed p=500 and varying $s$ and $n$}
    \label{fig:mseplot}
\end{figure}
In Figure $\ref{fig:mseplot}$, we illustrate the rate of convergence of $\norm{\hat{\theta}-\theta^*}_2$. In particular, $\norm{\hat{\theta}-\theta^*}_2$ against $\sqrt{\frac{s\log p}{n}}$ is plotted with varying $s$ and $n$. The error appears to be linear in $\sqrt{\frac{s\log p}{n}}$, and thus we also empirically conclude that our algorithm achieves the optimal $\sqrt{\frac{s\log p}{n}}$ rate.	
\section{Simulation study: Classification performance}
\label{SecSim}
In this section, we provide a simulation study which validates the classification performance for PUlasso. In particular we provide a comparison in terms of classification performance to state-of-the-art methods developed in~\cite{du_marthinus_convex_2015,elkan_learning_2008,liu_building_2003}. The focus of this section is classification rather than variable selection since many of the state-of-the-art methods we compare to are developed mainly for classification and are not developed for variable selection. 

\subsection{Comparison methods}

Our experiments compare six algorithms: (i) logistic regression model assuming we know the true responses (\emph{oracle} estimator); (ii) our PUlasso algorithm; (iii) a bias-corrected logistic regression algorithm in~\cite{elkan_learning_2008}; (iv) a second algorithm from~\cite{elkan_learning_2008} that is effectively a one-step EM algorithm; (v) the biased SVM algorithm from~\cite{liu_building_2003} and (vi) the PU-classification algorithm based on an asymmetric loss from~\cite{du_marthinus_convex_2015}.

The biased SVM from~\cite{liu_building_2003} is based on the supported vector machine~(SVM) classifier with two tuning  parameters which parameterize mis-classification costs of each kind. The first algorithm from~\cite{elkan_learning_2008} estimates label probabilities $ \mathbb{P}(z=1|x) $ and corrects the bias in the classifier via the estimation of $\mathbb{P}(z=1|y=1)$ under the assumption of a disjoint support between $\mathbb{P}(x|y=1)$ and $\mathbb{P}(x|y=0)$. Their second method is a modification of the first method; a unit weight is assigned to each labeled sample, and each unlabeled example is treated as a combination of a positive and negative example with weight $\mathbb{P}(y = 1|x, z = 0)$ and $\mathbb{P}(y = 0|x, z = 0)$, respectively. \cite{du_marthinus_convex_2015} suggests using asymmetric loss functions with $\ell_2$-penalty. Asymmetric loss function is considered to cancel the bias induced by separating positive and unlabeled samples rather than positive and negative samples. Any convex surrogate of 0-1 loss function can be used for the algorithm. There is a publicly available matlab implementation of the algorithm when a surrogate is the squared loss on the author's webpage\footnote{available at http://www.ms.k.u-tokyo.ac.jp/software.html} and since we use their code and implementation, the squared loss is considered.

\subsection{Setup}

We consider a number of different simulation settings: (i) small and large $p$ to distinguish the low and high-dimensional setting; (ii) weakly and strongly separated populations; (iii) weakly and highly correlated features; and (iv) correctly specified (logistic) or mis-specified model. Given dimensions $(n,p)$, sparsity level $s$, predictor auto-correlation $\rho$, separation distance $d$, and model specification scheme (logistic, mis-specified), our setup is the following:
\begin{itemize}
	\item Choose the active covariate set $S \subseteq \{1,2,\dots,p\}$  by taking $s$ elements uniformly at random from $(1,2,\dots,p)$. We let true $\theta^* \in \mathbb{R}^p$ such that $\theta^*_j = \mathbbm{1}_S(j)$.
	
	\item Draw samples $x \in \mathbb{R}^p, i.i.d$ from $\mathbb{P}_X = 0.5 \mathbb{P}_1 + 0.5 \mathbb{P}_0$ where $\mathbb{P}_1 := \mathscr{N}(\mu_1, \Sigma_{\rho})$, $\mathbb{P}_0 := \mathscr{N}(\mu_2, \Sigma_{\rho})$. More concretely, firstly draw $u \sim \mbox{Ber}(0.5)$. If $u=1$, draw $x$ from $\mathbb{P}_1$ and draw $x$ from $\mathbb{P}_0$ otherwise. \begin{itemize}
	    \item Mean vectors $\mu_1,\mu_2\in \mathbb{R}^p$ are chosen so that they are $s$-sparse, i.e. supp($\mu_i$) = $S$,  $E[\norm{\mu_1 -\mu_2}_2^2] = d^2$ and variance of $\mu_i$ does not depend on $d$. Specifically, we sample $\mu_1,\mu_2$ such that for $j \in S$, we let $\mu_{1j} \sim \mathscr{N}(\sqrt{(2d^2-1)/8s},1/\sqrt{8s})$, $\mu_{2j} = -\mu_{1j}$, and for $j \notin S$, $\mu_{ij} =0$ for $i \in (1,2)$. 
	    \item A covariance matrix $\Sigma_{\rho} \in \mathbb{R}^{p \times p}$ is taken to be $\Sigma_{\rho,ij} = K_{\rho } \rho^{|i-j|}$ where $K_\rho$ is chosen so that $\mathbbm{1}_S^T\Sigma_{\rho} \mathbbm{1}_S =s$. This scaling of $\Sigma_{\rho}$ is made to ensure that the signal strength $Var(x^T\theta^*) = \mathbbm{1}_S^T\Sigma_{\rho} \mathbbm{1}_S$ stays the same across $\rho$.
	\end{itemize} 
	\item Draw responses $y\in \{0,1\}$. If scheme = logistic, we draw y such that $y \sim \mbox{Ber}(\mathbb{P}_{\theta^*}(y=1|x))$ where $\mathbb{P}_{\theta^*}(y=1|x) = 1/(1+\exp(-{\theta^*}^Tx))$. In contrast, if scheme = mis-specified, we let $y=1$ if $x$ was drawn from $\mathbb{P}_1$, and zero otherwise; i.e. $y = \mathbbm{1}\{u = 1\}$.
 \end{itemize}
 
To compare performances both in low and high dimensional setting, we consider $ (p=10, s= 5) $ and $ (p=5000,s=5) $. We set the sample size $ n_\ell = n_u = 500 $ in both cases. Auto-correlation level $\rho$ takes values in $(0, 0.2,0.4,0.6,0.8)$. In the high dimensional setting, we excluded algorithm (v), since (v) requires a grid search over two dimensions, which makes the computational cost prohibitive. For algorithms (i)-(iv), tuning parameters $\lambda$ are chosen based on the 10-fold cross validation.

\subsection{Classification comparison}

We use two criteria, mis-classification rate and $F_1 $ score, to evaluate performances.  $ F_1 $ is the harmonic mean of the precision and recall, which is calculated as
$F_1 := 2\cdot \dfrac{\text{precision+recall}}{\text{precision$\cdot$recall}}.$  The $F_1$ score ranges from 0 to 1, where 1 corresponds to perfect precision and recall.
Experiments are repeated $50$ times and the average score and standard errors are reported. The result for the mis-classification rate under correct model specification is displayed in Figure~\ref{fig:classification_comp}. 

\begin{figure}[htbp]
    \centering
    \includegraphics[width=.8\linewidth, height= .4\textheight]{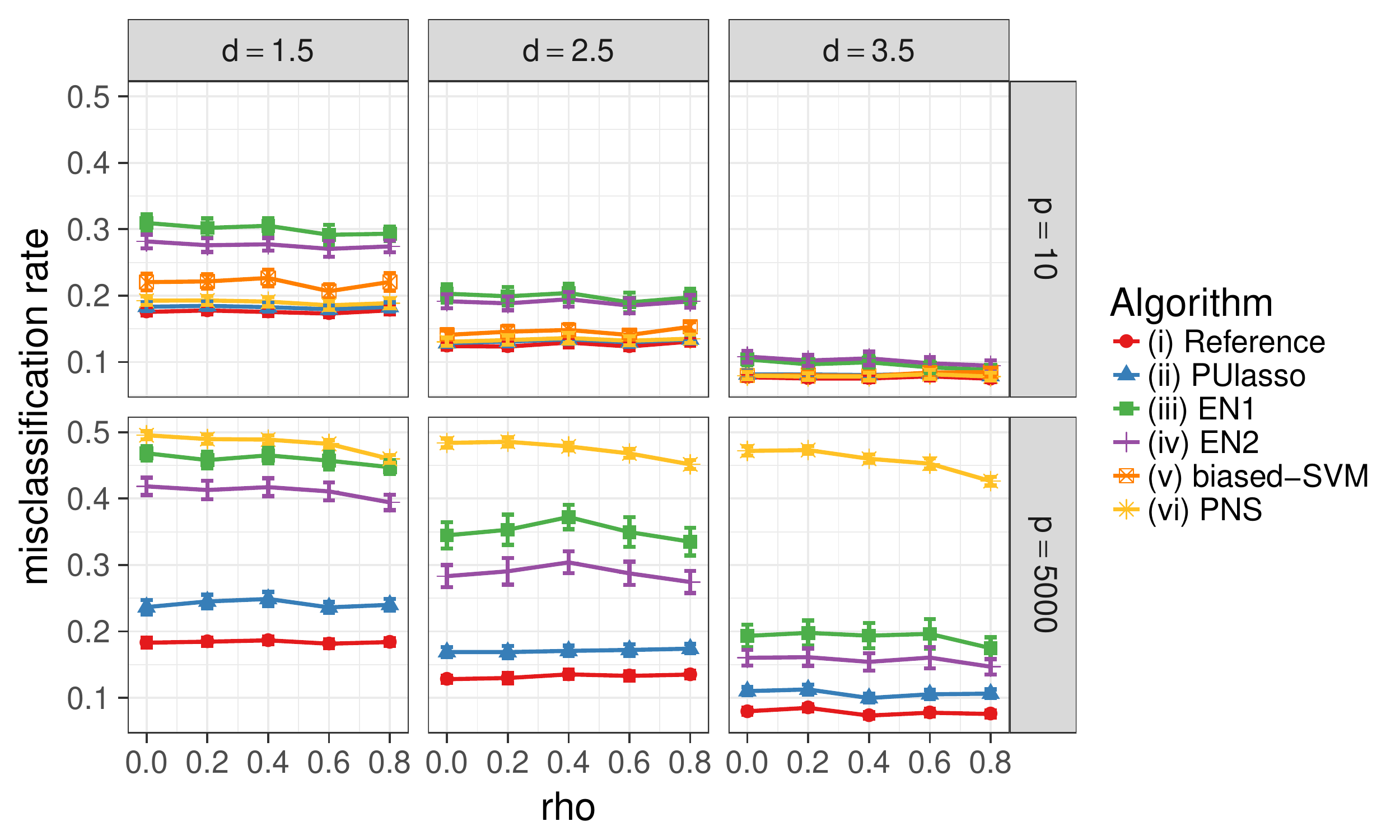}
    \caption{Mis-classification rates of algorithms (i)-(vi) under correct (logistic) model specification. Each error bar represents two standard errors of the mean. }\label{fig:classification_comp}
\end{figure}

Not surprisingly the oracle estimator has the best accuracy in all cases. PUlasso and algorithm (vi) performs almost as well as the oracle in the low-dimensional setting and better than remaining methods in most cases. It must be pointed out that both PUlasso and algorithm (vi) use additional knowledge $\pi$ of the true prevalence in the unlabeled samples. PUlasso performs best in the high-dimensional setting while the performance of algorithm (vi) becomes significantly worse because estimation errors can be greatly reduced by imposing many $0$'s on the estimates in PUlasso due to the $\ell_1$-penalty (compared to $\ell_2$-penalty in algorithm (vi)). The performance of (iii)-(iv) are greatly improved when positive and negative samples are more separated (large $d$), because algorithms (iii)-(iv) assume disjoint support between two distributions. The algorithms show similar performance when evaluated with the $F_1$ score metric and in the mis-specified setting. Due to space constraints, we defer the full set of remaining results in the supplementary material Section S3.

\section{Analysis of beta-glucosidase sequence data}
\label{SecExperiment}

Our original motivation for developing the PUlasso algorithm was to analyze a large-scale dataset with positive and unlabeled responses developed by the lab of Dr.Philip Romero~(\cite{Romero2015}). The prior EM algorithm approach of~\cite{ward_presence-only_2009} did not scale to the size of this dataset. In this section, we discuss the performance of our PUlasso algorithm on a dataset involving mutations of a natural beta-glucosidase~(BGL) enzyme. To provide context, BGL is a hydrolytic enzyme involved in the deconstruction of biomass into fermentable sugars for biofuel production. Functionality of the BGL enzyme is measured in terms of whether the enzyme deconstructs disaccharides into glucose or not. Dr. Romero used a microfluidic screen to generate a BGL dataset containing millions of sequences~(\cite{Romero2015})\footnote{The raw data is available in https://github.com/RomeroLab/seq-fcn-data.git}.

Main effects and two-way interaction models are fitted using our PUlasso algorithm with $\ell_1$ and $\ell_1/\ell_2$ penalties (we discuss how the groups are chosen shortly) over a grid of $\lambda$ values. We test stability of feature selection and classification performance using a modified ROC and AUC approach. Finally a scientific validation is performed based on a follow-up experiment conducted by the Romero lab. The variables selected by PUlasso were used to design a new BGL enzyme and the performance is compared to the original BGL enzyme.

\subsection{Data description}

The dataset consists of $n_\ell = 2647877$ labeled and functional sequences and $n_u = 1567203$ unlabeled sequences where each of the observation  $\sigma = (\sigma_1,\dots,\sigma_{500})$ is a sequence of amino acids of length $d=500$. Each of the position $\sigma_j\in (A,R,\dots,V,*)$ takes one of $M = 21$ discrete values, which correspond to the $20$ amino acids in the DNA code and an extra to include the possibility of a gap($*$). 

Another important aspect of the millions of sequences generated is that a ``base wild-type BGL sequence" was considered and known to be functional~($y = 1$), and the millions of sequences were generated by \emph{mutating} the base sequence. Single mutations~(changing one position from the base sequence) and double mutations~(changing two positions) from the base sequence were common but higher-order mutations were not prevalent using the deep mutational scanning approach in \cite{Romero2015}. Hence the sequences generated were not random samples across the entire enzyme sequence space, but rather very local sequences around the wild-type sequence. Hence the number of possible mutations in each position and consequently the total number of observed sequences is also reduced dramatically. With this dataset, we want to determine which mutations should be applied to the wild-type BGL sequence. 

Categorical variables $\sigma$ are converted into indicator variables: $x = (\mathbbm{1}\{\sigma_j=l\})_{j,l}$ where $ 1\leq j\leq 500$, $l\in (A,R,\dots,V,*) \setminus (\sigma^{WT}_{l})$ for the main-effects model, $x = (\mathbbm{1}\{\sigma_j=l\}, \mathbbm{1}\{\sigma_j=l,\sigma_k=m\})_{j,k,l,m}$ where $ 1\leq j,k \leq 500$, $j \neq k$, $l,m \in (A,R,\dots,V,*) \setminus (\sigma^{WT}_{l\, or\, m})$ for the pairwise interaction models, where  $\sigma^{WT}_l$ represents the amino acid of the wild-type sequence at the $l$th position. In other words, each variable corresponds to an indicator of mutation from the base sequence or interaction between mutations. Although there are in principle $p \approx d(M-1)$ variables for a main-effects model and $p \approx d^2(M-1)^2$ if we include main-effects and two-way interactions, there are many amino acids that never appear in any position or appear only a small number of times. For features corresponding to the main-effects~($\mathbbm{1}\{\sigma_j=l\}$ for some $j$ and $l$), those sparse features are aggregated \textit{within} each position until the number of mutations of the aggregated column reaches 100 or 1\% of the total number of mutations in each position; accordingly, each aggregated column is an indicator of any mutations to those sparse amino acids. For two-way interactions features~($\mathbbm{1}\{\sigma_j=l,\sigma_k=m\}$ for some $j,k,l$, and $m$), sparse features~($\leq 25$ out of $4215080$ samples) are simply removed from the feature space.
Using this basic pre-processing  we obtained only $3075$ corresponding to single mutations and $930$ binary variables corresponding to double mutations. They correspond to $500$ unique positions and $820$ two-way interactions between positions respectively. As mentioned earlier, we consider both $\ell_1$ and group $\ell_1/\ell_2$ penalties. We use the $\ell_1$-penalty for the main-effects model and the $\ell_1/\ell_2$ for the two-way interaction models. For the two-way interaction model each group $g_j$ corresponds to a different position  ($500$ total) and pair of positions ($820$ total) where mutations occur in the pre-processed design matrix and the group size $|g_j|$ corresponds to the number of different observed mutations in each position or pair of mutations in pair of positions~(for this dataset $m = \max_j |g_j| = 8$). Higher-order interactions were not modeled as they did not frequently arise. Hence the main-effects and two-way interaction model we consider have $p  = 3076\,(1 + 3075)$ and $p = 4006 \,(1+3075 + 930)$ and $J = 1320\, (500 + 820)$ groups respectively. In summary, we consider the following two models and corresponding design matrices 
\begin{gather*}
     X_{main} := \text{[Intercept(1)+ main effects(3075)]} \in \{0,1\}^{4215080 \times 3076}\\
    X_{int} : =  \text{[Intercept(1)+ main effects(3075)+ two way interactions(930)]}\in \{0,1\}^{4215080 \times 4006}
\end{gather*}
and the response vector $ z = [1,\dots,1,0,\dots,0] ^T \in \{0,1\}^{4215080}$. 

\subsection{Classification validation and model stability}
Next we validate the classification performance for both the main-effect and two-way interaction models. We fit models using 90\% of the randomly selected samples both from the positive and unlabeled set and use Area Under the ROC Curve~(AUC) to evaluate the classification performance on the 10\% of the hold-out set. Since positive and negative samples are mixed in the unlabeled test dataset this is a non-trivial task with presence-only responses. A naive approach is to treat unlabeled samples as negative and estimate AUC, but if we do so, the AUC is inevitably downward-biased because of the inflated false positive~(FP) rate. We note that a true positive~(TP) rate can be estimated in an unbiased manner using positive samples. To adjust such bias, we follow the methodology suggested in~\cite{jain_recovering_2017} and adjust false positive rate and AUC value using the following equation:
\begin{align*}
	\mbox{FP}^{adj} &= \frac{\mbox{FP}^{naive}-\pi \mbox{TP}}{1-\pi},\qquad \mbox{AUC}^{adj} = \dfrac{\mbox{AUC}^{naive} - \pi/2}{1-\pi}
\end{align*}
where $ \pi $ is the prevalence of positive samples.\\

\begin{figure}[htbp]
\centering
\includegraphics[width=.45\linewidth]{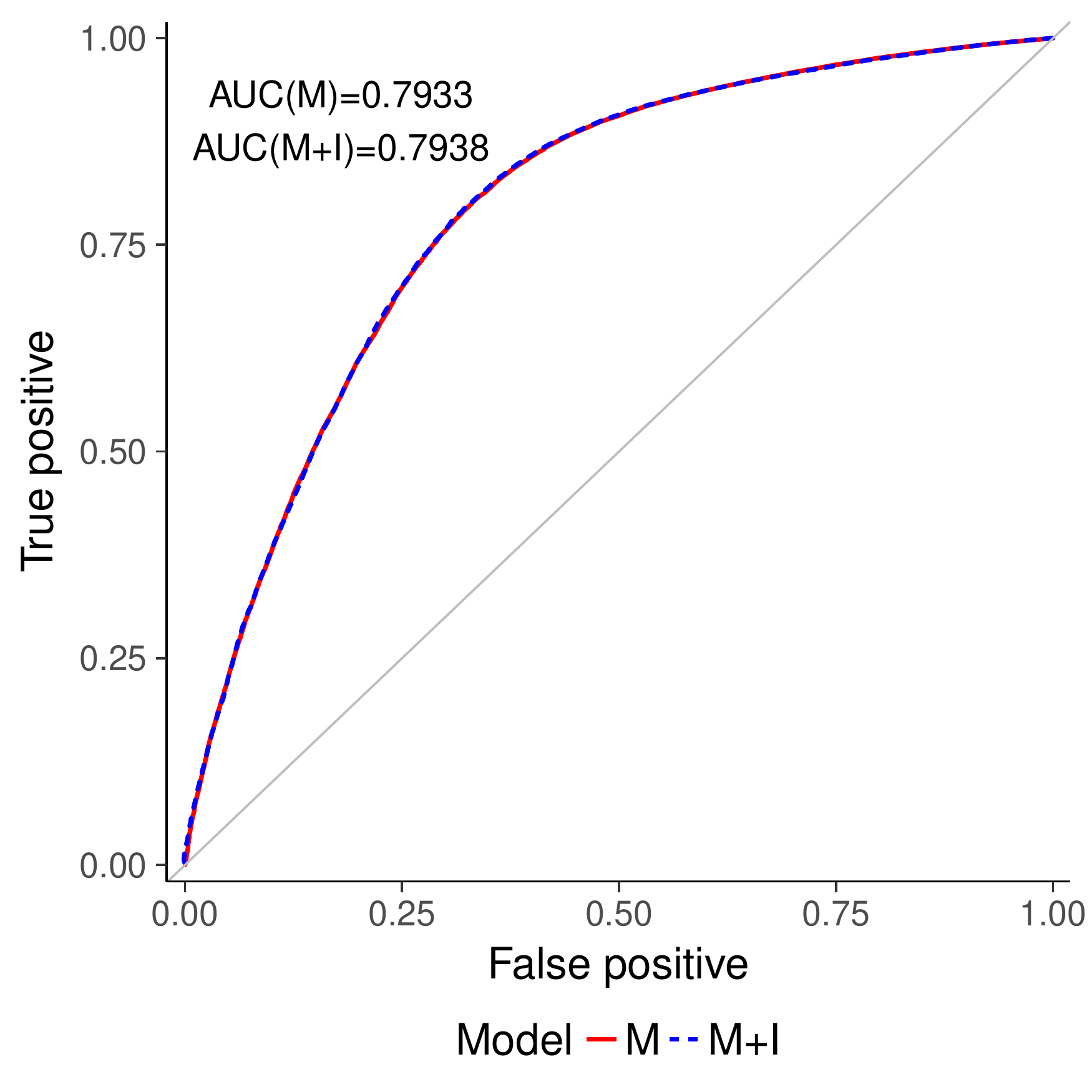}
\caption{ROC curves of main effects~(M) and two-way interaction model~(M+I) with $\lambda $ chosen based on 10-fold cross validation.}
\label{fig:AUC}
\end{figure}

As Fig.~\ref{fig:AUC} shows, we have a significant improvement in AUC over  random assignment~(AUC=$.5$) in both the main effect~(AUC=$.7933$) and two-way interaction~(AUC=$.7938$) models. The performances of the two models in terms of AUC values are very similar at their best $\lambda$ values chosen by 10-fold cross validation. This is not very surprising as only a small number of two-way 
interactions are observed in the experiments.

We also examined the stability of the selected features for both models as the training data changes. Following the methodology of ~\cite{kalousis_stability_2007}, we measure similarity between two subsets of features $s,s'$ using $S_S(s,s')$ defined as
$S_S(s, s') := 1- \dfrac{|s|+|s'|-2|s \cap s'|}{|s|+|s'|-|s \cap s'|}$.
$S_S$ takes values in $[0,1]$, where $0$ means that there is no overlap between the two sets, and $1$ that the two sets are identical. 
$S_s$ is computed for each pair of two training folds~(i.e. we have $\frac{9\cdot10}{2}$ pairs) using selected features and computed values are finally averaged over all pairs. Feature selection turned out to be very stable across all tuning parameter $\lambda$ values: on average we had about $95$\% overlap of selection in main effect model~(M) and about $98$\% overlap in main effect+interaction model~(M+I). Stability score is higher in the latter model since we do a feature selection on groups, whose number is much less than individual variables~($1320$ groups versus $3076$ individual variables). 

\begin{table}[htbp]
\centering
\scalebox{0.85}{
\begin{tabular}{ccccc}
  \hline
 & 1st Qu. & Median & Mean & 3rd Qu. \\ 
  \hline
M & 93.3\% & 94.9\% & 94.9\% & 96.8\% \\ 
  M+I & 97.9\% & 98.8\% & 98.4\% & 99.3\% \\ 
   \hline
\end{tabular}}\caption{Summary of stability scores across all tuning parameter $\lambda$ values }
\end{table}

\subsection{Scientific validation: Designed BGL sequence}
Finally we provide a scientific validation of the mutations estimated by our PUlasso algorithm. In particular, we fit the model with the PUlasso algorithm and selected the best $\lambda=0.0001$ based on the 10-fold cross validation. We use the top 10 mutations based on the largest size of coefficients with positive signs from our PUlasso algorithm because we are interested in mutations that enhance the performance of the sequence. Dr. Romero's lab designed the BGL sequence with the $10$ positive mutations from  Table~\ref{tab:10poscoefs}. This sequence was synthesized, expressed, and assayed for its hydrolytic activity. Hence the designed sequence has $10$ mutations compared to the wild-type~(base) BGL sequence. 

Figure~\ref{Fig:Sci} shows firstly that the designed protein sequence folds which in itself is remarkable given that $10$ positions are mutated. Secondly Figure~\ref{Fig:Sci} shows that the designed sequence decomposes disaccharides into glucose more quickly than the wild-type sequence. These promising results suggest that our variable selection method is able to identify positions of the wild-type sequences with improved functionality. 

\begin{table}[htbp]
\begin{minipage}[b]{.4\linewidth}
\centering
\begin{tabular}{rr}
	\hline
	\multicolumn{2}{l}{Base/Position/Mutated}\\
	\hline
	T197P & E495G \\ 
	K300P & A38G \\ 
	G327A & S486P \\ 
	A150D & T478S \\ 
	D164E & D481N \\ 
	\hline
\end{tabular}
\caption{Ten positive mutations}
\label{tab:10poscoefs}
\end{minipage}\hfill
\begin{minipage}[b]{.6\linewidth}
\centering
\includegraphics[width=0.7\linewidth]{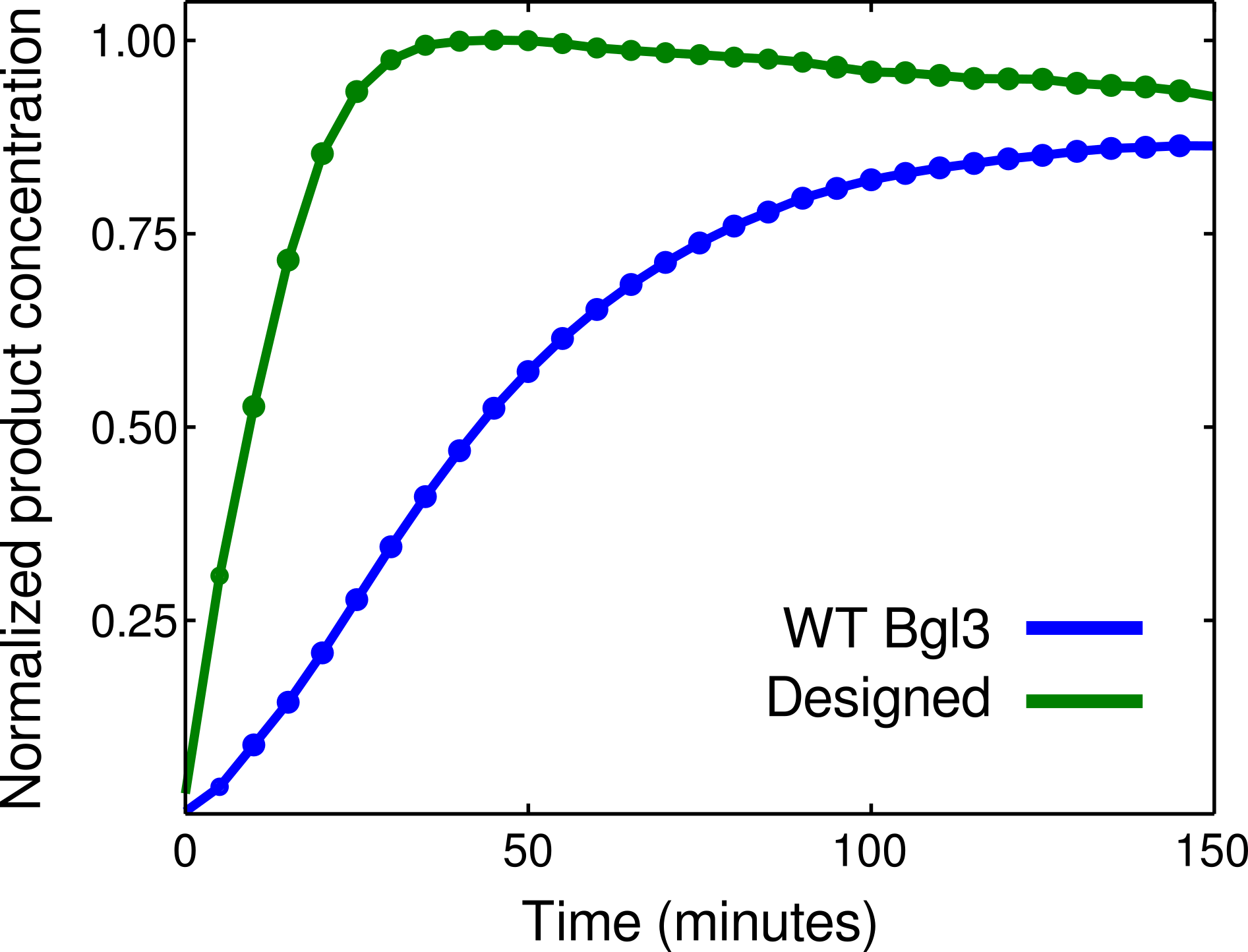}
\captionof{figure}{kinetics}
\label{Fig:Sci}
\end{minipage}
\caption*{10 positive mutations used in the lab(Base state/Position/Mutated state) and kinetics of designed BGL enzyme versus wild-type~(WT) BGL sequence. The designed BGL enzyme based on mutations from Table~\ref{tab:10poscoefs} displays faster kinetics than the WT BGL sequence.}
\end{table}
\section{Conclusion}
In this paper we developed the PUlasso algorithm for both variable selection and classification for high-dimensional classification with presence-only responses. Theoretically, we showed that our algorithm converges to a stationary point and every stationary point within a local neighborhood of $\theta^*$ achieves an optimal mean squared error~(up to constant). We also demonstrated that our algorithm performs well on both simulated and real data. In particular, our algorithm produces more accurate results than the existing techniques in simulations and performs well on a real biochemistry application.
\bibliographystyle{plainnat}
\bibliography{PUlasso}

\begin{thebibliography}{8}
\providecommand{\natexlab}[1]{#1}
\providecommand{\url}[1]{\texttt{#1}}
\expandafter\ifx\csname urlstyle\endcsname\relax
  \providecommand{\doi}[1]{doi: #1}\else
  \providecommand{\doi}{doi: \begingroup \urlstyle{rm}\Url}\fi

\bibitem[Koltchinskii(2011)]{Koltchinskii2011-he}
Vladimir Koltchinskii.
\newblock \emph{Oracle Inequalities in Empirical Risk Minimization and Sparse
  Recovery Problems: {\'E}cole {D'{\'E}t{\'e}} de Probabilit{\'e}s de
  {Saint-Flour} {XXXVIII-2008}}.
\newblock Springer Science \& Business Media, July 2011.

\bibitem[Ledoux and Talagrand(1991)]{LedTal91}
M.~Ledoux and M.~Talagrand.
\newblock \emph{Probability in Banach Spaces: Isoperimetry and Processes}.
\newblock Springer-Verlag, New York, NY, 1991.

\bibitem[Loh and Wainwright(2013)]{loh_regularized_2013}
P-L. Loh and M.~J. Wainwright.
\newblock Regularized {M}-estimators with nonconvexity: {Statistical} and
  algorithmic theory for local optima.
\newblock \emph{Journal of Machine Learning Research}, 1:\penalty0 1--9, 2013.

\bibitem[McDiarmid(1989)]{McDiarmid1989-gy}
Colin McDiarmid.
\newblock On the method of bounded differences.
\newblock \emph{Surveys in combinatorics}, 141\penalty0 (1):\penalty0 148--188,
  1989.

\bibitem[Negahban et~al.(2012)Negahban, Pradeep, Yu, and
  Wainwright]{negahban_unified_2012}
S.~N. Negahban, R.~Pradeep, Bin Yu, and M.~J. Wainwright.
\newblock A {Unified} {Framework} for {High}-{Dimensional} {Analysis} of
  {M}-{Estimators} with {Decomposable} {Regularizers}.
\newblock \emph{Statistica Sinica}, 27\penalty0 (4):\penalty0 538--557, 2012.

\bibitem[van~der Vaart and Wellner(1996)]{van_der_vaart_weak_1996}
A~W van~der Vaart and J~Wellner.
\newblock \emph{Weak {Convergence} and {Empirical} {Processes}: {With}
  {Applications} to {Statistics}}.
\newblock Springer {Series} in {Statistics}. Springer, 1996.

\bibitem[Wu(1983)]{wu_convergence_1983}
C.~F.~Jeff Wu.
\newblock On the {Convergence} {Properties} of the {EM} {Algorithm}.
\newblock \emph{The Annals of Statistics}, 11\penalty0 (1):\penalty0 95--103,
  1983.

\bibitem[Zangwill(1969)]{zangwill_nonlinear_1969}
W~I Zangwill.
\newblock \emph{Nonlinear programming: a unified approach}.
\newblock Prentice-{Hall} international series in management. Prentice-Hall,
  1969.

\end{thebibliography}


\begin{thebibliography}{35}
\providecommand{\natexlab}[1]{#1}
\providecommand{\url}[1]{\texttt{#1}}
\expandafter\ifx\csname urlstyle\endcsname\relax
  \providecommand{\doi}[1]{doi: #1}\else
  \providecommand{\doi}{doi: \begingroup \urlstyle{rm}\Url}\fi

\bibitem[Blazère et~al.(2014)Blazère, Loubes, and
  Gamboa]{blazere_oracle_2014}
M.~Blazère, J.~M. Loubes, and F.~Gamboa.
\newblock Oracle {Inequalities} for a {Group} {Lasso} {Procedure} {Applied} to
  {Generalized} {Linear} {Models} in {High} {Dimension}.
\newblock \emph{IEEE Transactions on Information Theory}, 60\penalty0
  (4):\penalty0 2303--2318, April 2014.

\bibitem[Breheny and Huang(2013)]{breheny_group_2013}
Patrick Breheny and Jian Huang.
\newblock Group descent algorithms for nonconvex penalized linear and logistic
  regression models with grouped predictors.
\newblock \emph{Statistics and Computing}, 25\penalty0 (2):\penalty0 173--187,
  2013.

\bibitem[Du~Marthinus et~al.(2015)Du~Marthinus, Niu, and
  Sugiyama]{du_marthinus_convex_2015}
Plessis Du~Marthinus, Gang Niu, and Masashi Sugiyama.
\newblock Convex {Formulation} for {Learning} from {Positive} and {Unlabeled}
  {Data}.
\newblock \emph{Proceedings of The 32nd International Conference on Machine
  Learning}, pages 1386--1394, 2015.

\bibitem[Elkan and Noto(2008)]{elkan_learning_2008}
Charles Elkan and Keith Noto.
\newblock Learning {Classifiers} from {Only} {Positive} and {Unlabeled} {Data}.
\newblock In \emph{Proceedings of the 14th {ACM} {SIGKDD} {International}
  {Conference} on {Knowledge} {Discovery} and {Data} {Mining}}, {KDD} '08,
  pages 213--220, New York, NY, USA, 2008. ACM.

\bibitem[Elsener and van~de Geer(2018)]{Elsener2018-tm}
Andreas Elsener and Sara van~de Geer.
\newblock Sharp oracle inequalities for stationary points of nonconvex
  penalized m-estimators.
\newblock February 2018.

\bibitem[Fahrmeir and Kaufmann(1985)]{fahrmeir_consistency_1985}
Ludwig Fahrmeir and Heinz Kaufmann.
\newblock Consistency and {Asymptotic} {Normality} of the {Maximum}
  {Likelihood} {Estimator} in {Generalized} {Linear} {Models}.
\newblock \emph{The Annals of Statistics}, 13\penalty0 (1):\penalty0 342--368,
  March 1985.

\bibitem[Fowler and Fields(2014)]{Fowler2014}
Douglas~M Fowler and Stanley Fields.
\newblock {Deep mutational scanning: a new style of protein science}.
\newblock \emph{Nature Methods}, 11:\penalty0 801--807, 2014.

\bibitem[Friedman et~al.(2007)Friedman, Hastie, Höfling, and
  Tibshirani]{friedman_pathwise_2007}
Jerome Friedman, Trevor Hastie, Holger Höfling, and Robert Tibshirani.
\newblock Pathwise coordinate optimization.
\newblock \emph{The Annals of Applied Statistics}, 1\penalty0 (2):\penalty0
  302--332, 2007.

\bibitem[Friedman et~al.(2010)Friedman, Hastie, and
  Tibshirani]{friedman_regularization_2010}
Jerome Friedman, Trevor Hastie, and Robert Tibshirani.
\newblock Regularization {Paths} for {Generalized} {Linear} {Models} via
  {Coordinate} {Descent}.
\newblock \emph{Journal of Statistical Software}, 33\penalty0 (1), 2010.

\bibitem[Hietpas et~al.(2011)Hietpas, Jensen, and Bolon]{Hietpas2011}
Ryan~T Hietpas, Jeffrey~D Jensen, and Daniel N~A Bolon.
\newblock {Experimental illumination of a fitness landscape}.
\newblock \emph{Proceedings of the National Academy of Sciences of the United
  States of America}, 108\penalty0 (19):\penalty0 7896--7901, 2011.

\bibitem[Huang et~al.(2012)Huang, Breheny, and Ma]{huang_selective_2012}
Jian Huang, Patrick Breheny, and Shuangge Ma.
\newblock A {Selective} {Review} of {Group} {Selection} in {High}-{Dimensional}
  {Models}.
\newblock \emph{Statistical Science}, 27\penalty0 (4):\penalty0 481--499,
  November 2012.

\bibitem[Jain et~al.(2017)Jain, White, and Radivojac]{jain_recovering_2017}
Shantanu Jain, Martha White, and Predrag Radivojac.
\newblock Recovering {True} {Classifier} {Performance} in
  {Positive}-{Unlabeled} {Learning}.
\newblock In \emph{Proceedings of the {Thirty}-{First} {AAAI} {Conference} on
  {Artificial} {Intelligence}, {February} 4-9, 2017, {San} {Francisco},
  {California}, {USA}.}, pages 2066--2072, 2017.

\bibitem[Kakade et~al.(2010)Kakade, Shamir, Sindharan, and
  Tewari]{kakade_learning_2010}
Sham Kakade, Ohad Shamir, Karthik Sindharan, and Ambuj Tewari.
\newblock Learning {Exponential} {Families} in {High}-{Dimensions}: {Strong}
  {Convexity} and {Sparsity}.
\newblock In \emph{Proceedings of the {Thirteenth} {International} {Conference}
  on {Artificial} {Intelligence} and {Statistics}}, pages 381--388, March 2010.

\bibitem[Kalousis et~al.(2007)Kalousis, Prados, and
  Hilario]{kalousis_stability_2007}
Alexandros Kalousis, Julien Prados, and Melanie Hilario.
\newblock Stability of {Feature} {Selection} {Algorithms}: {A} {Study} on
  {High}-dimensional {Spaces}.
\newblock \emph{Knowl. Inf. Syst.}, 12\penalty0 (1):\penalty0 95--116, May
  2007.

\bibitem[Krishnapuram et~al.(2005)Krishnapuram, Carin, Figueiredo, and
  Hartemink]{krishnapuram_sparse_2005}
B.~Krishnapuram, L.~Carin, M.~A.~T. Figueiredo, and A.~J. Hartemink.
\newblock Sparse multinomial logistic regression: fast algorithms and
  generalization bounds.
\newblock \emph{IEEE Transactions on Pattern Analysis and Machine
  Intelligence}, 27\penalty0 (6):\penalty0 957--968, June 2005.

\bibitem[Lancaster and Imbens(1996)]{lancaster_case-control_1996}
Tony Lancaster and Guido Imbens.
\newblock Case-control studies with contaminated controls.
\newblock \emph{Journal of Econometrics}, 71\penalty0 (1):\penalty0 145 --160,
  1996.

\bibitem[Lange et~al.(2000)Lange, Hunter, and Yang]{lange_optimization_2000}
Kenneth Lange, David~R. Hunter, and Ilsoon Yang.
\newblock Optimization {Transfer} {Using} {Surrogate} {Objective} {Functions}.
\newblock \emph{Journal of Computational and Graphical Statistics}, 9\penalty0
  (1):\penalty0 1--20, 2000.

\bibitem[Lee et~al.(2006)Lee, Lee, Abbeel, and Ng]{lee_efficient_2006}
Su-in Lee, Honglak Lee, Pieter Abbeel, and Andrew~Y. Ng.
\newblock Efficient l1 regularized logistic regression.
\newblock In \emph{In {Proceedings} of the {Twenty}-first {National}
  {Conference} on {Artificial} {Intelligence} ({AAAI}-06)}, pages 1--9, 2006.

\bibitem[Liu et~al.(2003)Liu, Dai, Li, Lee, and Yu]{liu_building_2003}
Bing Liu, Yang Dai, Xiaoli Li, Wee~Sun Lee, and Philip Yu.
\newblock Building {Text} {Classifiers} {Using} {Positive} and {Unlabeled}
  {Examples}.
\newblock \emph{Proceedings of the Third IEEE International Conference on Data
  Mining (ICDM'03)}, 2003.

\bibitem[Loh and Wainwright(2013)]{loh_regularized_2013}
P-L. Loh and M.~J. Wainwright.
\newblock Regularized {M}-estimators with nonconvexity: {Statistical} and
  algorithmic theory for local optima.
\newblock \emph{Journal of Machine Learning Research}, 1:\penalty0 1--9, 2013.

\bibitem[McCullagh and Nelder(1989)]{mccullagh_generalized_1989}
P.~McCullagh and J.~A. Nelder.
\newblock \emph{Generalized {Linear} {Models}}, volume~28.
\newblock 1989.

\bibitem[Mei et~al.(2018)Mei, Bai, and Montanari]{Mei2018-ec}
Song Mei, Yu~Bai, and Andrea Montanari.
\newblock The landscape of empirical risk for nonconvex losses.
\newblock \emph{Ann. Stat.}, 46\penalty0 (6A):\penalty0 2747--2774, December
  2018.

\bibitem[Meier et~al.(2008)Meier, Geer, Buhlmann, Van De~Geer, and
  Bühlmann]{meier_group_2008}
Lukas Meier, Van De~S Geer, Peter Buhlmann, Sara Van De~Geer, and Peter
  Bühlmann.
\newblock The group lasso for logistic regression.
\newblock \emph{Journal of the Royal Statistical Society, Series B},
  70\penalty0 (1):\penalty0 53--71, 2008.

\bibitem[Negahban et~al.(2012)Negahban, Pradeep, Yu, and
  Wainwright]{negahban_unified_2012}
S.~N. Negahban, R.~Pradeep, Bin Yu, and M.~J. Wainwright.
\newblock A {Unified} {Framework} for {High}-{Dimensional} {Analysis} of
  {M}-{Estimators} with {Decomposable} {Regularizers}.
\newblock \emph{Statistica Sinica}, 27\penalty0 (4):\penalty0 538--557, 2012.

\bibitem[Ortega and Rheinboldt(2000)]{Ortega}
J.~M. Ortega and W.~C. Rheinboldt.
\newblock \emph{Iterative solution of nonlinear equations in several
  variables}.
\newblock Classics in applied mathematics. SIAM, New York, 2000.

\bibitem[Puig et~al.(2011)Puig, Wiesel, Fleury, and
  Hero]{puig_multidimensional_2011}
A.~T. Puig, A.~Wiesel, G.~Fleury, and A.~O. Hero.
\newblock Multidimensional {Shrinkage}-{Thresholding} {Operator} and {Group}
  {LASSO} {Penalties}.
\newblock \emph{IEEE Signal Processing Letters}, 18\penalty0 (6):\penalty0
  363--366, June 2011.

\bibitem[Raskutti et~al.(2010)Raskutti, Wainwright, and Yu]{RasWaiYu10b}
G.~Raskutti, M.~J. Wainwright, and B.~Yu.
\newblock Restricted eigenvalue conditions for correlated {G}aussian designs.
\newblock \emph{Journal of Machine Learning Research}, 11:\penalty0 2241--2259,
  2010.

\bibitem[Raskutti et~al.(2011)Raskutti, Wainwright, and
  Yu]{raskutti_minimax_2011}
Garvesh Raskutti, Martin~J. Wainwright, and Bin Yu.
\newblock Minimax {Rates} of {Estimation} for {High}-{Dimensional} {Linear}
  {Regression} {Over} $\ell_q$-{Balls}.
\newblock \emph{IEEE Transactions on Information Theory}, 57\penalty0
  (10):\penalty0 6976--6994, October 2011.

\bibitem[Romero et~al.(2015)Romero, Tran, and Abate]{Romero2015}
Philip~A Romero, Tuan~M Tran, and Adam~R Abate.
\newblock {Dissecting enzyme function with microfluidic-based deep mutational
  scanning}.
\newblock \emph{Proceedings of the National Academy of Sciences of the United
  States of America}, 112\penalty0 (23):\penalty0 7159--7164, 2015.

\bibitem[Simon and Tibshirani(2012)]{simon_standardization_2012}
Noah Simon and Robert Tibshirani.
\newblock Standardization and the {Group} {Lasso} {Penalty}.
\newblock \emph{Statistica Sinica}, 22\penalty0 (3):\penalty0 1--21, 2012.

\bibitem[Tibshirani et~al.(2012)Tibshirani, Bien, Friedman, Hastie, Simon,
  Taylor, and Tibshirani]{tibshirani_strong_2012}
Robert Tibshirani, Jacob Bien, Jerome Friedman, Trevor Hastie, Noah Simon,
  Jonathan Taylor, and Ryan~J. Tibshirani.
\newblock Strong rules for discarding predictors in lasso-type problems.
\newblock \emph{Journal of the Royal Statistical Society. Series B: Statistical
  Methodology}, 74\penalty0 (2):\penalty0 245--266, 2012.

\bibitem[van~de Geer(2008)]{van_de_geer_high-dimensional_2008}
Sara~A. van~de Geer.
\newblock High-dimensional generalized linear models and the lasso.
\newblock \emph{The Annals of Statistics}, 36\penalty0 (2):\penalty0 614--645,
  April 2008.

\bibitem[Ward et~al.(2009)Ward, Hastie, Barry, Elith, and
  Leathwick]{ward_presence-only_2009}
Gill Ward, Trevor Hastie, Simon Barry, Jane Elith, and John~R. Leathwick.
\newblock Presence-only data and the em algorithm.
\newblock \emph{Biometrics}, 65\penalty0 (2):\penalty0 554--563, 2009.

\bibitem[Wu and Lange(2008)]{wu_coordinate_2008}
Tong~Tong Wu and Kenneth Lange.
\newblock Coordinate descent algorithms for lasso penalized regression.
\newblock \emph{The Annals of Applied Statistics}, 2\penalty0 (1):\penalty0
  224--244, 2008.

\bibitem[Yuan and Lin(2006)]{yuan_model_2006}
Ming Yuan and Yi~Lin.
\newblock Model selection and estimation in regression with grouped variables.
\newblock \emph{J. R. Statist. Soc. B}, 68\penalty0 (1):\penalty0 49--67, 2006.

\end{thebibliography}
 \newpage
\setcounter{section}{0}
\setcounter{assumption}{0}
\setcounter{proposition}{0}
\setcounter{table}{0}
\setcounter{figure}{0}
\setcounter{page}{1}
\setcounter{equation}{0}
\renewcommand{\theequation}{S\arabic{equation}}
\renewcommand{\thefigure}{S\arabic{figure}}
\renewcommand{\thesection}{S\arabic{section}}
\renewcommand{\theassumption}{S\arabic{assumption}}
\begin{center}
{\Large\bf SUPPLEMENTARY MATERIAL}
\end{center}
\section{Proofs for results in Section 2}\label{supp_sec:Section 2}
\subsection{Proof of Proposition 2.1}
\label{supp_sec:pf_prop2.1}
We prove (i) in Proposition 2.1 for both Algorithm 1 and 2.
First we define $Q,\widetilde{Q},H$ as follows:
\begin{align*}
Q(\theta ;\theta^m) &:= n^{-1}E_{\theta^m}[\log L_f (\theta) | z_1^n, x_1^n] \\
\widetilde{Q}(\theta ;\theta^m) &:= -Q(\theta;\theta^m)+P_\lambda(\theta)\\
H(\theta; \theta^m) &:= n^{-1} E_{\theta^m}[\log \mathbb{P}_{\theta}( y_1^n|z_1^n,x_1^n) | z_1^n,x_1^n].
\end{align*} 
Note that for any $ \theta^m $, $\mathscr{F}_n(\theta) = \widetilde{Q}(\theta ;\theta^m)+H(\theta;\theta^m)$ holds and $ H(\theta^m;\theta^m)  \geq H(\theta; \theta^m)$ by Jensen's inequality. Also since $ \theta^{m+1} $ is a minimizer of  $ \widetilde{Q}(\theta ;\theta^m) $, we have
\begin{equation}\label{eq:non-incr}
\mathscr{F}_n(\theta^{m+1})=\widetilde{Q}(\theta^{m+1} ;\theta^m)+H(\theta^{m+1};\theta^m)\leq \widetilde{Q}(\theta^{m} ;\theta^m)+H(\theta^{m};\theta^m) = \mathscr{F}_n(\theta^m).
\end{equation}
To show that the inequality is strict, it suffices to show that if $ \theta^m \not\in \mathscr{S} $, $\theta^m$ is not a stationary point of $\widetilde{Q}$. Since $ \theta^m \not\in \mathscr{S} $, there exists $\theta'$ such that 
\begin{equation}\label{eq:pf_prop1}
	\triangledown\mathscr{F}_n(\theta^m)^T(\theta'-\theta^m)<0,\forall \triangledown\mathscr{F}_n(\theta^m) \in \partial \mathscr{F}_n(\theta^m)
\end{equation}
Since $ \theta^m $ is a maximizer of $H(\cdot; \theta^m)$, $ \triangledown H(\theta^m ; \theta^m)=0$. Then $\partial \mathscr{F}_n(\theta^m) = \partial \widetilde{Q}(\theta^m;\theta^m)$. Thus by~\eqref{eq:pf_prop1}, $ \theta^m $ is not a stationary point of $ \widetilde{Q}(\cdot; \theta^m) $. 

For Algorithm 2 (PUlasso algorithm), since $\overline{Q}$ is a surrogate function of $Q$ which satisfies following two properties
\begin{equation}\label{supp_eq:surrogateProperties}
\overline{Q}(\theta^m;\theta^m) = Q(\theta^m;\theta^m), \quad \overline{Q}(\theta;\theta^m) \leq Q(\theta;\theta^m), \forall \theta
\end{equation}
and $ \theta^{m+1} $ is a minimizer of $ -\overline{Q}(\theta ;\theta^m)+P_\lambda(\theta) $, we have
\begin{align*}
    \mathscr{F}_n(\theta^m) &=-Q(\theta^{m} ;\theta^m)+P_\lambda(\theta^{m})+H(\theta^{m};\theta^m)\\
    &= -\overline{Q}(\theta^{m} ;\theta^m)+P_\lambda(\theta^{m})+H(\theta^{m};\theta^m) \\
    &\geq -\overline{Q}(\theta^{m+1} ;\theta^m)+P_\lambda(\theta^{m+1})+H(\theta^{m};\theta^m) \\
    &\geq-Q(\theta^{m+1} ;\theta^m)+P_\lambda(\theta^{m+1})+H(\theta^{m+1};\theta^m)=\mathscr{F}_n(\theta^{m+1})
\end{align*}
The strict inequality follows from the fact that $ \triangledown Q(\theta^m;\theta^m) = \triangledown \overline{Q} (\theta^m;\theta^m) $.

Now we address (ii) and (iii) in Proposition 2.1. Using the same argument as in \citeSupp{wu_convergence_1983}, we appeal to the global convergence theorem stated below as Theorem~\ref{thm:globalConvThm} in~\citeSupp{zangwill_nonlinear_1969} with $ \Gamma = \mathscr{S}, \alpha = \mathscr{F}_n$, and letting $A$ be a mapping from $\theta^m$ to $\theta^{m+1}$ defined by Algorithm 1 or 2. As stated in \citeSupp{wu_convergence_1983}, condition (iii) in Theorem~\ref{thm:globalConvThm} follows from the continuity of $ -Q(\theta,\theta')+P_\lambda(\theta) $ or $-\bar{Q}(\theta;\theta')+P_\lambda(\theta)$ in both $ \theta,\theta' $. Therefore, if we show that 
$\widetilde{\Theta_0}$ is compact, both (ii) and (iii) follow from the fact that $(\theta^m)_{m=0}^{\infty} $ lie in a compact set. Since $\widetilde{\Theta_0} \subseteq \mathbb{R}^p$ it suffices to show that $\widetilde{\Theta_0}$ is closed and bounded in $\mathbb{R}^p$. $\widetilde{\Theta_0}$ is bounded since $\mathscr{F}_n(\theta) \rightarrow \infty$ whenever $\norm{\theta}_2 \rightarrow \infty$ since $ \norm{\theta}_{\mathscr{G},2,1} \geq \min_j w_j \norm{\theta}_2\rightarrow \infty$. For closedness of the set, consider $(\theta_k)_{k\geq 1}$ such that $\theta_k \in \widetilde{\Theta_0}$ and $\theta_k \rightarrow \theta'$. We have $\mathscr{F}_n(\theta_k) \leq \mathscr{F}_n(\theta_{null})$ for all $k$. Then by the continuity of $\mathscr{F}_n$, $\mathscr{F}_n(\theta') \leq \mathscr{F}_n(\theta_{null})$ thus $\theta' \in \widetilde{\Theta_0}$. 

\begin{theorem}[Global Convergence Theorem,~\citeSupp{zangwill_nonlinear_1969}]
\label{thm:globalConvThm}  Let the sequence $ \{x_k\}_{k=0}^\infty $ be generated by $ x_{k+1} \in A(x_k) $, where $ A $ is a point-to-set map on $ X $. Let a solution set $ \Gamma \in X$ be given, and suppose that:
\begin{enumerate}
\item[(i)] The sequence $ \{x_k\}_{k=0}^\infty \subset S$  for $  S \subset X $ a compact set.
\item[(ii)] There is a continuous function $ \alpha $ on $ X $ such that (a) if $ x \not\in \Gamma $, then $ \alpha(y)<\alpha(x) $ for all $ y \in A(x) $. (b) if $ x \in \Gamma $, then $ \alpha(y)\leq \alpha(x) $ for all $ y \in A(x) $. 
\item[(iii)] The mapping A is closed at all points of $ X \setminus  \Gamma $. 
\end{enumerate}
Then all the limit points of any convergent subsequence of $ (x_k)_{k=0}^\infty $ are in the solution set $ \Gamma $ and $ \alpha(x_k) $ converges monotonically to $ \alpha(x) $ for some $ x \in \Gamma $.
\end{theorem}

\section{Proofs for results in Section 3}\label{supp_sec:Section 3}
\subsection{Derivation of the log-likelihood in the form of GLMs}\label{supp_sec:pf_glm}

\begin{align*}
\log L(\theta;x,z,s=1) &= \log \left(\prod_i \mathbb{P}_\theta(z_i|x_i,s_i=1)\right)\\
& = \sum_i z_i \log \mathbb{P}_\theta(z_i=1|x_i,s_i=1)+(1-z_i) \log \mathbb{P}_\theta(z_i=0|x_i,s_i=1)\\
& = \sum_i z_i \log \dfrac{\mathbb{P}_\theta(z_i=1|x_i,s_i=1)}{\mathbb{P}_\theta(z_i=0|x_i,s_i=1)}+\log \mathbb{P}_\theta(z_i=0|x_i,s_i=1).
\end{align*}

From Lemma 2.1, we have
$\mathbb{P}_\theta(z=1|x,s=1) = \dfrac{\frac{n_l}{\pi n_u} e^{\theta^Tx}}{1+(1+\frac{n_l}{\pi n_u})e^{\theta^Tx}}$.
Then,
\begin{align*}
\log \dfrac{\mathbb{P}_\theta(z=1|x,s=1)}{\mathbb{P}_\theta(z=0|x,s=1)} &=
\log \dfrac{\frac{n_l}{\pi n_u} e^{\theta^Tx}}{1+e^{\theta^Tx}} = \log \frac{n_l}{\pi n_u} +\theta^Tx -\log(1+e^{\theta^Tx}).
\end{align*}	 
and, 
\begin{align*}
\log \mathbb{P}_\theta(z=0|x,s=1)&= - \log \left(\dfrac{1+(1+\frac{n_l}{\pi n_u} )e^{\theta^Tx}}{1+e^{\theta^Tx}}\right)= - \log \left(1+ \dfrac{\frac{n_l}{\pi n_u} e^{\theta^Tx}}{1+e^{\theta^Tx}}\right) \\
&= - \log \left(1+ e^{\log \frac{n_l}{\pi n_u}+\theta^Tx- \log(1+e^{\theta^Tx})}\right).
\end{align*}
Therefore we obtain,
\begin{align*}
  \log \left(\prod_i \mathbb{P}_\theta(z_i|x_i,s_i=1)\right) 
  &= \sum_i z_i \eta_i - \log (1+e^{\eta_i})
    \end{align*}
where $\eta_i = \log \frac{n_l}{\pi n_u}+\theta^Tx- \log(1+e^{\theta^Tx}).$ 

\subsection{Useful inequalities and technical lemmas}\label{supp_sec:lemmas}
In this section, we provide some results that will be useful for our proofs. First we state the symmetrization inequality, which shows relationships between empirical and Rademacher processes.

\begin{theorem}\label{thm:symmetrization}(Symmetrization theorem[\citeSupp{van_der_vaart_weak_1996}])
Let $ U_1,\dots, U_n $ be independent random variables with values in $ \mathscr{U} $ and  $ (\epsilon_i) $ be an i.i.d. sequence of Rademacher variables, which take values $ \pm 1 $ each with probability 1/2. Let $ \Gamma $ be a class of real-valued functions on $ \mathscr{U} $. then
\[ E\left( \sup_{\gamma \in \Gamma} \left\lvert \sum_{i=1}^{n} \{\gamma(U_i) - E(\gamma(U_i))\}\right\rvert \right)  \leq 2E\left(\sup_{\gamma \in \Gamma}\left\lvert\sum_{i=1}^{n}\epsilon_i \gamma(U_i)\right\rvert\right).\]
\end{theorem}

The next theorem is Ledoux-Talagrand contraction theorem. The stated version is Theorem 2.2 in \citeSupp{Koltchinskii2011-he}, which allows $T$ be any subset in $\mathbb{R}^n$, thus slightly more general than the original theorem in \citeSupp{LedTal91} where $T$ needs to be bounded.
\begin{theorem}\label{thm:contraction}(Contraction theorem[\citeSupp{LedTal91}]) Let $T \subset \mathbb{R}^n$ and let $ \varphi_i: \mathbb{R} \rightarrow \mathbb{R}$, $i=1,\dots,n$ be contractions which satisfy $ |\varphi_i(s) - \varphi_i(t)| \leq |s-t| ,s,v \in \mathbb{R}$ and $ \varphi_i(0) =0 $. Let $(\epsilon_i)$ be independent Rademacher random variables. Then 
\[ E\left(\sup_{t\in T} \left| \sum_{i=1}^{n} \epsilon_i \varphi_i(t_i) \right|\right) \leq 2 E\left(\sup_{t\in T}  \left| \sum_{i=1}^{n}\epsilon_i t_i \right|\right). \] 
\end{theorem}

Finally, we state the bounded differences inequality, also sometimes called as Hoeffding-Azuma inequality.
\begin{theorem}\label{thm:BDI}(Bounded difference inequality[\citeSupp{McDiarmid1989-gy}])
Let $X_1 ,\dots, X_n$ be arbitrary independent random variables on set $A$ and $\varphi:A^n \rightarrow \mathbb{R}$ satisfy
the bounded difference assumption: there exists constants $c_i, i=1,\dots,n$ such that for all $i=1,\dots,n$ and all $x_1,x_2,\dots,x_i,x_i',\dots,x_n$,
\begin{equation*}
    \lvert \varphi(x_1,\dots,x_i,\dots,x_n)-\varphi(x_1,\dots,x_i',\dots,x_n)\rvert \leq c_i
\end{equation*}
Then $\forall t > 0$, 
\begin{equation*}
    \mathbb{P}\left(\varphi(X_1,\dots,X_n)-E[\varphi(X_1,\dots,X_n)] \geq t \right) \leq \exp(-2t^2/\sum_{i=1}^n c_i^2)
\end{equation*}
\end{theorem}

Now we state and prove some useful results about sub-Gaussian and sub-exponential random variables.
\begin{lemma}\label{lem:blk_holder}
Let $v,u \in \mathbb{R}^p$ and $(g_1,\dots,g_J)$ be a partition of $(1,\dots,p)$. For $\mathscr{G} = ((g_1,\dots,g_J), (w_j)_1^J)$ and  $\bar{\mathscr{G}} = ((g_1,\dots,g_J), (w_j^{-1})_1^J)$ such that all $g_j$ are non-empty and $w_j>0$, $|v^Tu| \leq \norm{v}_{\mathscr{G},2,1}\norm{u}_{\bar{\mathscr{G}},2,\infty}$.
\end{lemma}
\begin{proof}
We note $\norm{v}_{\mathscr{G},2,1} = \sum_{j=1}^J w_j\norm{v_{g_j}}_2$ and $\norm{u}_{\bar{\mathscr{G}},2,\infty}:= \max_{1\leq j \leq J} \norm{w_j^{-1}u_{g_j}}_2.$
By Cauchy-Schwarz inequality, we have
	$$|v^Tu| \leq \sum_{j=1}^J |w_j v_{g_j}^T w_j^{-1}u_{g_j}|\leq \sum_{j=1}^J \norm{w_j v_{g_j}}_2 \norm{w_j^{-1}u_{g_j}}_2.$$
Taking the maximum of the second quantity, 
	$$ |v^Tu|\leq \max_{1\leq j \leq J} \norm{w_j^{-1}u_{g_j}}_2 \sum_{j=1}^J w_j\norm{ v_{g_j}}_2 = \norm{v}_{\mathscr{G},2,1}\norm{u}_{\bar{\mathscr{G}},2,\infty}. $$
\end{proof}

\begin{lemma}\label{lem:subG_expk}
Let $x \in \mathbb{R}^p$ such that $x^Tv \sim \mbox{subG}(\norm{v}_2^2 \sigma_x^2)$ for any fixed $v \in \mathbb{R}^p$ and $E[x]=0$. For any $i \in (1,\dots,p)$, $k\geq 1$, 
$$E[|x_i|^{k}] \leq k(2\sigma_x^2)^{k/2}\Gamma(k/2).$$
\end{lemma}
\begin{proof}
Taking $v = e_i$ where $e_i$ is an $i$th coordinate vector, we have $E(\exp(tv^Tx)) = E[\exp(tx_i)]\leq \exp( t^2\sigma_x^2/2) $ for $t \in \mathbb{R}$. Then following a standard argument for sub-Gaussian random variables,
\begin{align*}
    E[|x_i|^{k}] & = \int_{s=0}^\infty \mathbb{P}(|x_i| \geq s^{1/k}) ds\\
    &\leq 2 \int_{s=0}^\infty \exp(-s^{2/k}/2\sigma_x^2) ds\\
    &= k(2\sigma_x^2)^{k/2} \int_{s=0}^\infty e^{-u} u^{k/2-1}du = k(2\sigma_x^2)^{k/2}\Gamma(k/2)
\end{align*}
where the third inequality comes from the change of variable $u = s^{2/k}/2\sigma_x^2$.
\end{proof}

The next lemma concerns distribution of $x\circ x = [x_1^2,\dots,x_s^2]$ for independent sub-Gaussian $(x_i)_{i=1}^s$.
\begin{lemma}\label{lem:subG_subExp}
Let $x \in \mathbb{R}^s$ such that $x^Tv \sim \mbox{subG}(\norm{v}_2^2 \sigma_x^2)$ for any fixed $v \in \mathbb{R}^s$ and $E[x]=0$. Also, assume $(x_i)_{i=1}^s$ are independent. Then we have $v^T (x \circ x)  \sim \mbox{subExp}(\nu,b)$ with $\nu =16\sigma_x^2\norm{v}_2$, $b=16\sigma_x^2 \norm{v}_\infty$ for any fixed $v \in \mathbbm{R}^s$.
\end{lemma}

\begin{proof}
Let $z := x \circ x - E[x \circ x]$. For any given $v \in \mathbb{R}^s$ and $t>0$,
\begin{align*}
E[\exp(tv^Tz)] &= E[\exp(tv_1z_1+\dots tv_sz_s)]\\
&= \prod_{i=1}^s E[\exp(tv_iz_i)]
\end{align*}
where we use independence. Then by Taylor series expansion,
\begin{align*}
E[\exp(tv^Tz)] &=\prod_{i=1}^sE\left(1+tv_iz_i +\frac{t^2 (v_iz_i)^2}{2}+\dots \right)\\
&=\prod_{i=1}^s \left(1+\sum_{k=2}^\infty \dfrac{t^k E\left(v_i(x_i^2 - E[x_i^2])\right)^k}{k!}\right)
\end{align*}

By Jensen's inequality, we have,
$$E(v_i x_i^2 - E[v_ix_i^2])^k \leq |v_i|^k 2^{k-1}(E[x_i^{2k}] + E[x_i^2]^k),$$
and by applying Jensen's inequality again, we get
\begin{align}\label{lems2.3:s42}
E[\exp(tv^Tz)] \leq \prod_{i=1}^s \left(1+\sum_{k=2}^\infty \dfrac{t^k |v_i|^k 2^k E[x_i^{2k}] }{k!}\right).
\end{align}
We let $t_i = t|v_i|$. By Lemma \ref{lem:subG_expk}, we have,
\begin{equation}\label{lems2.3:s43}
    E[x^{2k}_{i}] \leq  (2k)(2\sigma_x^2)^k\Gamma(k) = 2(k!)(2\sigma_x^2)^k
\end{equation}

Substituting \eqref{lems2.3:s43} into \eqref{lems2.3:s42},
\begin{align*}
E[\exp(tv^Tz)] &\leq  \prod_{i=1}^s \left( 1+\sum_{k=2}^\infty t_i^k 8^k  (\sigma_x^2)^k \right)\\
&=\prod_{i=1}^s \left(1+ (8t_i \sigma_x^2)^2\sum_{k=0}^\infty (8t_i \sigma_x^2)^k\right)\\
&\leq \prod_{i=1}^s\left(1+128t_i^2 \sigma_x^4\right)
\end{align*}
if $t|v_i| \leq 1/(16\sigma_x^2)$, for all $i$. By the fact that $1+128t_i^2 \sigma_x^4 \leq \exp(128t_i^2 \sigma_x^4)$
\begin{align*}
E[\exp(tv^Tz)] \leq \prod_{i=1}^s \exp( 128t_i^2\sigma_x^4) = \exp( \sum_{i=1}^s 128t^2 v_i^2\sigma_x^4) = \exp(  128t^2 \norm{v}_2^2 \sigma_x^4)
\end{align*}
for $t \leq 1/(16\sigma_x^2 \max_i |v_i|)$. Therefore $v^T x\circ x  \sim \mbox{subExp}(\nu,b)$ with $\nu = 16\sigma_x^2 \norm{v}_2$, $b= 16\sigma_x^2 \norm{v}_\infty)$.
\end{proof}

Also, we have a lemma about maximum of sum of variables with sub-exponential tails.

\begin{lemma}\label{lem:exp_subExp}
	Consider $(u_j)_{j=1}^J$ where $u_j \in \mathbb{R}^{m_j}$ such that $\mathbbm{1}^T u_j \sim \mbox{subExp}(\nu_j,b)$ with $E[u_j]=0$ for $1 \leq j \leq J$. We let $m:= \max_j m_j$. Also, assume $\exists \nu_*>0$ such that $\nu_j \leq \nu_* \sqrt{m}$ for all $j$ and $\exists c>0$ such that $b \leq c\nu_*$.
 Then we have,
	\begin{equation*}
		E[\max_{1\leq j \leq J} \mathbbm{1}^T u_j]  \leq c\nu_* (\log J+ m/(2c^2)).
	\end{equation*}
	In particular, when $c=1$, 
	$\displaystyle E[\max_{1\leq j \leq J} \mathbbm{1}^T u_j]  \leq \nu_* (\log J+ m/2).$
\end{lemma}
\begin{proof}
For $ |t| \leq 1/b$ we have,
\begin{equation}\label{lems2.4:s44}
    E[\exp(t \mathbbm{1}^T u_j)] \leq \exp(t^2\nu_j^2/2)\leq \exp(mt^2\nu_*^2/2)
\end{equation}
Then, 
	\begin{align*}
		E[\max_{1\leq j \leq J} \mathbbm{1}^T u_j] &= \frac{1}{t}E\left(\log e^{\max_{1\leq j \leq J} t(\mathbbm{1}^T u_j)}\right)\\
		&\leq \frac{1}{t}\log E\left( e^{\max_{1\leq j \leq J} t(\mathbbm{1}^T u_j)}\right)\\
		&= \frac{1}{t}\log E\left( \max_{1\leq j \leq J} e^{ t(\mathbbm{1}^T u_j)}\right).
	\end{align*}
where the second inequality comes from Jensen's. Using a union bound,
\begin{align}
	\frac{1}{t}\log E\left( \max_{1\leq j \leq J} e^{ t(\mathbbm{1}^T u_j)}\right) &\leq \frac{1}{t}\log  \left(\sum_{j=1}^J E\left( e^{ t(\mathbbm{1}^T u_j)}\right)\right) \nonumber\\
	&\leq \frac{1}{t}\log  \left(J e^{mt^2\nu_*^2/2}\right) \label{lems2.7:s7}.
\end{align}
where the last inequality uses \eqref{lems2.4:s44}. Since $1/(c\nu_*) \leq 1/b$ by assumption, the inequality \eqref{lems2.7:s7} holds for $t = 1/(c\nu_*)$. Plugging $t = 1/(c\nu_*)$ into \eqref{lems2.7:s7}, we obtain,
\begin{align*}
	E[\max_{1\leq j \leq J}\mathbbm{1}^T u_j] \leq c\nu_* (\log J+ m/(2c^2))
\end{align*}
as claimed.
\end{proof}

Finally, in Lemma \ref{lem:exp_blknorm} and \ref{lem:tailbound_blknorm}, we provide expectation and probability tail bounds of a dual $\ell_1/\ell_2$ norm of a sub-Gaussian vector.
\begin{lemma}\label{lem:exp_blknorm}
Let $\mathscr{G} = ((g_1,\dots,g_J), (w_j)_1^J)$. Consider a random vector $v\in \mathbb{R}^p$ such that for each $j$ and any fixed $u \in \mathbb{R}^{|g_j|}$, $u^Tv_{g_j} \sim \mbox{subG}(\sigma^2\norm{u}_2^2) $ with $E[v_{g_j}]=0$ and $u^T(v_{g_j} \circ v_{g_j}) \sim \mbox{subExp}(\nu \norm{u}_2, \nu \norm{u}_\infty)$. Then, 
\begin{equation*}
	E[\norm{v}_{\bar{\mathscr{G}},2,\infty}] \leq c \sqrt{\log J+ m}
\end{equation*}
for $c = (\min_{1\leq j \leq J} w_j)^{-1}\sqrt{\max(\nu, 8\sigma^2)}$, where we define $\bar{\mathscr{G}} = ((g_1,\dots,g_J), (w_j^{-1})_1^J)$ and $ m := \max_j |g_j|$, the largest group size.
\end{lemma}
\begin{proof}

First we let $m_j = |g_j|$. By Holder's inequality, we have,
	\begin{align*}
		E[\max_{1\leq j \leq J} \norm{ w_j^{-1} v_{g_j}}_2] \leq E[\max_{1\leq j \leq J} \norm{ w_j^{-1} v_{g_j}}_2^2]^{1/2}  = E[\max_{1\leq j \leq J}  w_j^{-2} (v_{g_j,1}^2+\dots v_{g_j,m_j}^2)]^{1/2}
	\end{align*}
Then,
\begin{align*}
		E[\max_{1\leq j \leq J}  w_j^{-2} (v_{g_j,1}^2+\dots v_{g_j,m_j}^2)] &\leq (\max_{1\leq j \leq J}  w_j^{-2})E[\max_{1\leq j \leq J} (v_{g_j,1}^2+\dots v_{g_j,m_j}^2 )]  \\
		& =(\max_{1\leq j \leq J} w_j^{-2})E[\max_{1\leq j \leq J} \sum_{i=1}^{m_j} (u_{g_j,i}+E[v_{g_j,i}^2])  )]\\
		&\leq (\max_{1\leq j \leq J} w_j^{-2})\left(E[\max_{1\leq j \leq J}\mathbbm{1}^Tu_{g_j}] +4m\sigma^2\right)
	\end{align*}
	where $u_{g_j} \overset{d}{:=} v_{g_j} \circ v_{g_j} - E[v_{g_j}\circ v_{g_j}]$ and the last inequality uses Lemma \ref{lem:subG_expk} and $m_j \leq m$, for all $j$. By assumption, we have,  
	$\mathbbm{1}^Tu_{g_j} \sim \mbox{subExp}(\nu \sqrt{m_j},\nu)$ and $E[u_{g_j}]=0$.
	Then, by Lemma \ref{lem:exp_subExp},
\begin{align*}
LHS &\leq (\max_{1\leq j \leq J} w_j^{-2})[\nu(\log J+ m/2)+4m\sigma^2]\\
& \leq (\max_{1\leq j \leq J} w_j^{-2}) 
\max(\nu, 8\sigma^2)(\log J+ m).
\end{align*}
Since $\max_{1\leq j \leq J} w_j^{-2} = 1/(\min_{1\leq j \leq J} w_j)^{2}$, defining $c=(\min_{1\leq j \leq J} w_j)^{-1}\sqrt{\max(\nu, 8\sigma^2)}$, we obtain 
\begin{equation*}
	\norm{v}_{\bar{\mathscr{G}},2,\infty} \leq  c  \sqrt{\log J+ m}
\end{equation*}
as desired.
		 
\end{proof}

\begin{lemma}\label{lem:tailbound_blknorm}
Let $\mathscr{G} = ((g_1,\dots,g_J), (w_j)_1^J)$. Consider a random vector $v\in \mathbb{R}^p$ such that for each $j$ and for any fixed $u \in \mathbb{R}^{|g_j|}$, $u^Tv_{g_j} \sim \mbox{subG}(\sigma^2\norm{u}_2^2) $ with $E[v_{g_j}]=0$ and $u^T(v_{g_j} \circ v_{g_j}) \sim \mbox{subExp}(\nu \norm{u}_2, \nu \norm{u}_\infty)$. Then, 
$$	\mathbb{P} \left( \norm{v}_{\bar{\mathscr{G}},2,\infty} \geq \delta \right)\leq J\exp \left(-\frac{1}{2}\min(C_\delta^2/\nu^2,C_\delta/\nu)\right) $$
where we define $C_\delta := (\min_j w_j^2)\delta^2/m - 4\sigma^2$, $\bar{\mathscr{G}} = ((g_1,\dots,g_J), (w_j^{-1})_1^J)$, and $ m := \max_j |g_j|$, the largest group size.
\end{lemma}
\begin{proof}
By the union bound, we have
	\begin{align*}
		\mathbb{P} \left( \max_{1\leq j \leq J} \norm{ w_j^{-1} v_{g_j}}_2 \geq \delta \right) 
		& \leq \sum_{j=1}^J \mathbb{P} \left( \norm{ w_j^{-1} v_{g_j}}_2^2 \geq \delta^2\right).
		\end{align*}
Defining $u_{g_j} \overset{d}{:=} v_{g_j} \circ v_{g_j} - E[v_{g_j}\circ v_{g_j}]$ and $m_j := |g_j|$. 		\begin{align*}
		\mathbb{P} \left( \max_{1\leq j \leq J} \norm{ w_j^{-1} v_{g_j}}_2 \geq \delta \right) 
		&\leq  \sum_{j=1}^J \mathbb{P} \left( \sum_{k=1}^{m_j} v_{g_j,k}^2 \geq  w_j^2\delta^2\right)\\
		&\leq \sum_{j=1}^J \mathbb{P} \left( \mathbbm{1}^Tu_{g_j} \geq (\min_j w_j^2)\delta^2 - 4m_j\sigma^2\right)
		\end{align*}
where the last inequality uses Lemma \ref{lem:subG_expk}. By assumption, we have $\mathbbm{1}^Tu_{g_j} \sim \mbox{subExp}(\nu \sqrt{m_j},\nu)$ and $E[u_{g_j}]=0$. We use Bernstein type inequality to bound the probability. More concretely for any $s>0$ such that $|s|\leq 1/\nu$,  we have, 
		\begin{align*}
		    \mathbb{P} \left( \mathbbm{1}^Tu_{g_j} \geq  (\min_j w_j^2)\delta^2 - 4m_j\sigma^2\right)&\leq \mathbb{P} \left( s\mathbbm{1}^Tu_{g_j} \geq sm C_\delta \right) \\
		    &\leq \exp(-smC_\delta) E\left[\exp\left(s\mathbbm{1}^Tu_{g_j}\right)\right]\\
		    &\leq \exp(-smC_\delta+s^2m\nu^2/2).
		\end{align*}
In the first and third inequality, the bound $m_j \leq m$ was also used. Optimizing over $s>0$, we take $s = \min\{C_\delta/\nu^2,1/\nu\}$. Hence, we have,
		\begin{align*}
		\mathbb{P} \left( \max_{1\leq j \leq J} \norm{ w_j^{-1} v_{g_j}}_2 \geq \delta  \right)
		&\leq J\exp \left(-\frac{m}{2}\min(C_\delta^2/\nu^2,C_\delta/\nu)\right)
	\end{align*}

	
\end{proof}
\subsection{Proof for Proposition 3.1}
\label{supp_sec:pf_prop3.1}
The proof of this result follows similar lines to the proof of Theorem 1 in \citeSupp{loh_regularized_2013}, which established the result with a different tolerance function and an additive penalty. Since $ \theta^* $ is feasible, by the first order optimality condition, we have the following inequality
\begin{equation*}
(\triangledown \mathscr{L}_n(\hat{\theta}) + \triangledown P_\lambda(\ltheta))^T(\theta^* - \ltheta) \geq 0.
\end{equation*}
Letting $ \hat{\Delta} := \ltheta - \theta^* $, since $ \ltheta \in \Theta_0$ by the setup of the problem, we can apply RSC condition to obtain
\begin{equation}\label{rscP}
\alpha\norm{\hat{\Delta}}_2^2 -\tau (\norm{\hat{\Delta}}_{\mathscr{G},2,1})
\leq (-\triangledown P_\lambda(\ltheta)-\triangledown \mathscr{L}_n(\theta^*))^T\hat{\Delta}.
\end{equation}
On the other hand, convexity of $ P_\lambda(\theta) $ implies
\begin{equation}\label{covP}
P_\lambda(\theta^*) - P_\lambda(\ltheta) \geq -\triangledown P_\lambda(\ltheta)^T\hat{\Delta}.
\end{equation}
Combining~\eqref{rscP} with~\eqref{covP}, we obtain 
\begin{align*}
\alpha\norm{\hat{\Delta}}_2^2 -\tau (\norm{\hat{\Delta}}_{\mathscr{G},2,1})
&\leq (-\triangledown P_\lambda(\ltheta)-\triangledown \mathscr{L}_n(\theta^*))^T\hat{\Delta}\\
&\leq P_\lambda(\theta^*) - P_\lambda(\ltheta)+\norm{\triangledown \mathscr{L}_n(\theta^*)}_{\bar{\mathscr{G}},2,\infty}\norm{\hat{\Delta}}_{\mathscr{G},2,1}.
\end{align*}
by Lemma \ref{lem:blk_holder}.
Since $ \tau (\norm{\hat{\Delta}}_{\mathscr{G},2,1}) = \tau_1 \dfrac{\log J+m}{n}\norm{\hat{\Delta}}_{\mathscr{G},2,1}^2 + \tau_2 \sqrt{\dfrac{\log J+m}{n} }\norm{\hat{\Delta}}_{\mathscr{G},2,1}$,
\begin{equation*}
\alpha\norm{\hat{\Delta}}_2^2 \leq P_\lambda(\theta^*) - P_\lambda(\ltheta)+\norm{\hat{\Delta}}_{\mathscr{G},2,1}\left( \tau_1 \dfrac{\log J+m}{n}\norm{\hat{\Delta}}_{\mathscr{G},2,1} +\tau_2\sqrt{\dfrac{\log J+m}{n}}+\norm{\triangledown \mathscr{L}_n(\theta^*)}_{\bar{\mathscr{G}},2,\infty}\right),
\end{equation*}
By the choice of $ \lambda $,
\begin{equation*}
\tau_1 \dfrac{\log J+m}{n}\norm{\hat{\Delta}}_{\mathscr{G},2,1} +\tau_2\sqrt{\dfrac{\log J+m}{n}}+\norm{\triangledown \mathscr{L}_n(\theta^*)}_{\bar{\mathscr{G}},2,\infty}  \leq \dfrac{\lambda}{2}.
\end{equation*}

Then by using the triangle inequality
\begin{align*}
\alpha\norm{\hat{\Delta}}_2^2 
&\leq P_\lambda(\theta^*) - P_\lambda(\ltheta)+ \dfrac{\lambda}{2}\norm{\hat{\Delta}}_{\mathscr{G},2,1}\\
& = \lambda \sum_{j \in S} w_j \norm{\theta^*_{g_j}}_2 - \lambda \sum_{j \in S} w_j \norm{\ltheta_{g_j}}_2- \lambda \sum_{j \in S^c} w_j \norm{\ltheta_{g_j}}_2+\dfrac{\lambda}{2}\sum_{j=1}^J w_j\norm{\hat{\Delta}_{g_j}}_2\\
& \leq \lambda \sum_{j \in S} w_j \norm{\hat{\Delta}_{g_j}}_2 -\lambda \sum_{j \in S^c} w_j \norm{\ltheta_{g_j}}_2+\dfrac{\lambda}{2}\sum_{j=1}^J w_j\norm{\hat{\Delta}_{g_j}}_2
\end{align*}
where $S := \{j \in (1,\dots,J); \theta^*_{g_j} \neq 0\} $ where the last inequality comes from the triangle inequality. Since for $j \in S^c$, $\ltheta_{g_j} =\ltheta_{g_j} - \theta^*_{g_j}$,
\begin{align*}
\alpha\norm{\hat{\Delta}}_2^2 
&\leq \lambda \sum_{j \in S} w_j \norm{\hat{\Delta}_{g_j}}_2 -\lambda \sum_{j \in S^c} w_j \norm{\hat{\Delta}_{g_j}}_2+\dfrac{\lambda}{2}\sum_{j=1}^J w_j\norm{\hat{\Delta}_{g_j}}_2\\
&= \frac{3\lambda}{2} \sum_{j \in S} w_j \norm{\hat{\Delta}_{g_j}}_2 -\frac{\lambda}{2} \sum_{j \in S^c} w_j \norm{\hat{\Delta}_{g_j}}_2.
\end{align*}
In particular, we have
\begin{equation}
	\sum_{j \in S^c} w_j \norm{\hat{\Delta}_{g_j}}_2 \leq 3 \sum_{j \in S} w_j \norm{\hat{\Delta}_{g_j}}_2
\end{equation}
and 
\begin{equation}
	\alpha\norm{\hat{\Delta}}_2^2  \leq \frac{3\lambda}{2} \sum_{j \in S} w_j \norm{\hat{\Delta}_{g_j}}_2.
\end{equation}
Then,
$$ \alpha\norm{\hat{\Delta}}_2^2 \leq (\max_{j \in S} w_j)\frac{3\lambda}{2}  (\sum_{j \in S} \norm{\hat{\Delta}_{g_j}}_2^2)^{1/2}(\sum_{j \in S} 1)^{1/2} \leq (\max_{j \in S} w_j)\frac{3\lambda}{2} \sqrt{|S|} \norm{\hat{\Delta}}_2.$$

The $ \ell_1/\ell_2 $  upper bound follows from the $ \ell_2 $-bound and \[ \norm{\hat{\Delta}}_{\mathscr{G},2,1} =\sum_{j \in S} w_j \norm{\hat{\Delta}_{g_j}}_2 + \sum_{j \in S^c} w_j \norm{\hat{\Delta}_{g_j}}_2 \leq 4(\max_{j \in S} w_j)\sum_{j \in S}  \norm{\hat{\Delta}_{g_j}}_2  \leq 4(\max_{j\in S} w_j)\sqrt{|S|} \norm{\hat{\Delta}}_2 \].

\subsection{Proof of Lemma 3.1}
\label{supp_sec:pf_lem3.1}
Recalling $\mathscr{L}_n(\theta) = 
\frac{1}{n}\sum_{i=1}^n \left(-z_i f(\theta^T x_i) - A(f(\theta^T x_i))\right)$, we have $$\triangledown \mathscr{L}_n(\theta^*)  = \frac{1}{n}\sum_{i=1}^n \left( -z_i  +\mu(f({\theta^*}^T x_i) )\right) \dfrac{1}{1+e^{{\theta^*}^T x_i}}x_i,$$ where we define $A(\eta) = \log(1+e^\eta)$, $\mu(\eta) = A'(\eta) = e^{\eta}/(1+e^{\eta})$ and $f(\theta^Tx) = \log(n_\ell/\pi n_u) +\theta^Tx -\log (1+e^{\theta^Tx})$. For $ 1\leq i \leq n $ and $1 \leq j \leq p$, define $
V_{ij} :=  \left( -z_i  +\mu(f({\theta^*}^T x_i) )\right)\dfrac{1}{1+e^{{\theta^*}^Tx_i}}x_{ij}$. We note $\triangledown \mathscr{L}_n(\theta^*)_j = \empavg V_{ij}$.

Considering the event, with $C:=36\sigma_x^2$, 
$$\mathscr{E} = \left\lbrace\max_{1\leq j \leq p}\empavg x_{ij}^2 \leq C  \right\rbrace.$$
we have,
\begin{equation*}
\mathbb{P}\left( \norm{\triangledown \mathscr{L}_n(\theta^*) }_{\bar{\mathscr{G}},2,\infty} \geq \delta \right) \leq \mathbb{P}(\mathscr{E}^c)+\mathbb{P}\left( \norm{\triangledown \mathscr{L}_n(\theta^*) }_{\bar{\mathscr{G}},2,\infty} \geq \delta |\mathscr{E} \right)\mathbb{P}(\mathscr{E}).
\end{equation*}
First we show that $\mathbb{P}(\mathscr{E}^c)$ is small. Since each $x_{ij}$ is a sub-Gaussian variable with sub-Gaussian parameter $\sigma_x$, defining $z_{ij} = x_{ij}^2 - E[x_{ij}^2]$,
\begin{equation*}
	\mathbb{P}(\mathscr{E}^{c}) \leq p\mathbb{P}\left(\empavg x_{ij}^2 \geq 36\sigma_x^2 \right)\leq p\mathbb{P}\left(\empavg z_{ij} \geq 32\sigma_x^2 \right)
\end{equation*}
where we use the fact that $E[x_{ij}^2] \leq 4\sigma_x^2$. We note that $(z_{ij})_{i=1}^n$ are i.i.d. samples from mean-zero distribution with sub-Exponential tail with parameter $\nu=b=16\sigma_x^2$ by applying Lemma \ref{lem:subG_subExp} with $s=1$. 
By Bernstein-type tail bound of the sub-exponential random variable,
\begin{equation}\label{eq:pf_lem3.1_1}
	\mathbb{P}(\mathscr{E}^{c}) \leq p\mathbb{P}\left(\empavg z_{ij} \geq 32\sigma_x^2 \right)\leq \exp(-\frac{n}{2}(2-\frac{2\log p}{n}))\leq \exp (-n/2),
\end{equation}
by the sample size condition $n \gtrsim \log J+ m$, assuming sufficiently large $n$. Now we show that $ \empavg V_{ij} $ is a sub-Gaussian variable on $\mathscr{E}$. 
In particular, we show that $E[\exp(t\empavg V_{ij})|\mathscr{E}] \leq \exp(t^2v^2/2)$ for some $v>0$.

Defining $ t_i:=\dfrac{t}{n(1+e^{{\theta^*}^Tx_i})}$, by definition of $V_{ij}$, we have
\begin{align}
E\left[\exp (\frac{t}{n}V_{ij} )|x_i\right] &= E\left[\exp\left( -t_i z_i x_{ij}\right)\cdot \exp\left( t_i\mu(f(x_i^T{\theta^*}))x_{ij}\right)|x_i \right]\nonumber\\
& = E\left[ \exp\left( -t_i z_i x_{ij}\right)|x_i\right]\cdot \exp\left( t_i\mu(f(x_i^T{\theta^*}))x_{ij}\right)\label{eq:pf_lem3.1_2}.
\end{align}
By the property of exponential family, we obtain
\begin{align}
E\left[ \exp\left( -t_i z_i x_{ij}\right)|x_i\right] 
&= \int \exp\left( -t_i z x_{ij}\right)\cdot \exp(zf(x_i^T{\theta^*})-A(f(x_i^T{\theta^*}))dz\nonumber	\\
&=\exp\left\lbrace A(f(x_i^T{\theta^*})-t_i x_{ij})-A(f(x_i^T{\theta^*})) \right\rbrace.\label{eq:pf_lem3.1_3}
\end{align}
Therefore combining \eqref{eq:pf_lem3.1_2} and \eqref{eq:pf_lem3.1_3}, we obtain
\begin{align*}
E\left[\exp (\frac{t}{n}V_{ij} )|x_i\right]
& = \exp\left\lbrace A(f(x_i^T{\theta^*})-t_i x_{ij})-A(f(x_i^T{\theta^*})) +t_i\mu(f(x_i^T{\theta^*}))x_{ij} \right\rbrace \\
& \leq \exp \left\lbrace \frac{1}{8n^2}(t x_{ij})^2 \right\rbrace
\end{align*}
where the second inequality comes from the second order Taylor expansion, $\mu(\cdot) = A'(\cdot)$, $\sup_u A''(u) \leq 1/4$, and $ t_i \leq t/n$.
Therefore
\begin{equation*}
\prod_{i=1}^n E\left[\exp (\dfrac{t}{n}V_{ij} )|x_i\right]  \leq \exp\left( \frac{t^2}{8n^2}\sum_{i=1}^{n}x_{ij}^2 \right),
\end{equation*}
and conditioned on $ \mathscr{E} $, we have the bound
$$  \exp\left( \frac{t^2}{8n^2}\sum_{i=1}^{n}x_{ij}^2 \right) \leq \exp\left( \frac{t^2C}{8n} \right).$$
Therefore, $\empavg V_{ij} \sim \mbox{subG}(C/4n)$, i.e. $\triangledown \mathscr{L}_n(\theta^*)_j \sim \mbox{subG}(C/4n)$ for all $j$.

Now we discuss the distribution of $u^T \triangledown \mathscr{L}_n(\theta^*)_{g_j}$ and $u^T \triangledown \mathscr{L}_n(\theta^*)_{g_j} \circ \triangledown \mathscr{L}_n(\theta^*)_{g_j}$ on $\mathscr{E}$, for any $u \in \mathbb{R}^{|g_j|}$, to apply Lemma \ref{lem:tailbound_blknorm}. By Assumption 1, $(\triangledown \mathscr{L}_n(\theta^*)_j)_{j \in g_j}$ are independent. With independence, it is easy to see for any $j$ and any fixed $u \in \mathbb{R}^{|g_j|}$, $u^T \triangledown \mathscr{L}_n(\theta^*)_{g_j}\sim \mbox{subG}(\norm{u}_2^2 (C/4n))$ and $E[\triangledown \mathscr{L}_n(\theta^*)]=0$. Then Lemma \ref{lem:subG_subExp} gives 
$$u^T (\triangledown \mathscr{L}_n(\theta^*)_{g_j}\circ \triangledown \mathscr{L}_n(\theta^*)_{g_j}) \sim \mbox{subExp}(\norm{u}_2 (4C/n),\norm{u}_\infty (4C/n))$$ for any $j$ and fixed $u \in \mathbb{R}^{|g_j|}$. Therefore the condition of Lemma \ref{lem:tailbound_blknorm} is satisfied with $\sigma^2 = C/4n$ and $\nu = 16\sigma^2 =  4C/n$.

We let $\delta^2 = 16C(\log J+m)/(\min_j w_j^2 n)$ and note that
$$C_\delta = \frac{(\min_jw_j^2)\delta^2}{m}-\frac{C}{n} = \frac{16C(\log J+m)}{mn} -\frac{C}{n} = \frac{4C}{n}\left(\frac{16\log J}{4m}+\frac{15}{4}\right) $$
By Lemma \ref{lem:tailbound_blknorm},
\begin{align*}
	\mathbb{P} \left( \norm{\triangledown \mathscr{L}_n(\theta^*) }_{\bar{\mathscr{G}},2,\infty}\geq \delta  | \mathscr{E}\right)\leq \exp \left( -\frac{m}{2} \min \left( \frac{C_\delta^2}{( 4C/n)^2}, \frac{C_\delta}{4C/n}\right)+\log J\right),
\end{align*}
and because $\log J/m \geq 0$, $C_\delta \geq 4C/n$, and $\min \left( \frac{C_\delta^2}{( 4C/n)^2}, \frac{C_\delta}{4C/n}\right) = \frac{C_\delta}{4C/n}$ if $C_\delta \geq 4C/n$, we have, 
\begin{align}\label{eq:pf_lem3.1_4}
	\mathbb{P} \left( \norm{\triangledown \mathscr{L}_n(\theta^*) }_{\bar{\mathscr{G}},2,\infty}\geq \delta  | \mathscr{E}\right)&\leq \exp \left( -\frac{m}{2}\left(\frac{4\log J}{m}+\frac{15}{4}\right) +\log J\right)\nonumber \\
	&\leq  \exp \left(-\log J -m\right).
\end{align}

Putting \eqref{eq:pf_lem3.1_1} and \eqref{eq:pf_lem3.1_4} together, and noting $\delta = (24\sigma_x/\min_jw_j) \sqrt{\frac{\log J+m}{n}}$, we obtain
\begin{align*}
\mathbb{P}\left( \norm{\triangledown \mathscr{L}_n(\theta^*) }_{\bar{\mathscr{G}},2,\infty} \geq (24\sigma_x/\min_j w_j)\sqrt{\dfrac{\log J+m}{n}} \right)
&\leq  \exp(-0.5n) +\exp(-\log J-m)\leq \epsilon
\end{align*}
where the last inequality follows from the sample size condition $n \gtrsim (\log J+ m) \vee (1/\epsilon)^{1/\beta}$.

\subsection{Proof of Theorem 3.2 }
\label{supp_sec:pf_thm3.2}

\subsubsection{Proof Outline}
Defining $f(\theta^Tx) = \log(n_l / \pi n_u) + \theta^Tx - \log (1+e^{\theta^Tx})$, we recall that 
\begin{equation*}
\mathscr{L}_n(\theta)  = \dfrac{1}{n}\sum_{i=1}^{n}\left( -z_i f(\theta^Tx_i) + \log (1+e^{f(\theta^Tx_i)})\right)  .
\end{equation*}
Taking a derivative with respect to $ \theta $ of $ \mathscr{L}_n(\theta) $, we obtain
\begin{equation*}
\triangledown \mathscr{L}_n(\theta)  = \dfrac{1}{n}\sum_{i=1}^{n}\left( -z_i  +\mu(f(\theta^Tx_i) )\right) f'(\theta^Tx_i)x_i
\end{equation*}
and 
\begin{align}\label{eq:dRn}
&\left(\triangledown \mathscr{L}_n(\theta) - \triangledown \mathscr{L}_n(\theta^*)\right)^T\Delta\nonumber\\
&=  \left( \dfrac{1}{n}\sum_{i=1}^{n}\left( \mu(f(\theta^Tx_i) ) - z_i\right)  f'(\theta^Tx_i) -\left( \mu(f({\theta^*}^Tx_i)) -z_i\right) f'({\theta^*}^Tx_i)  \right)x_i^T\Delta
\end{align}
where $\Delta$ is defined as $\Delta := \theta - \theta^*$, and $A(\cdot), \mu(\cdot)$ defined as $A(\eta):=\log(1+e^{\eta})$, $\mu(\eta) := A'(\eta)= e^{\eta}/(1+e^{\eta})$. Also we let $e_i:=\mu(f({\theta^*}^Tx_i))-z_i$.

To prove that \eqref{eq:dRn} is positive with high probability, we decompose \eqref{eq:dRn} into two terms, whose first term $I$ has a positive expectation and the second term $II$ has an expectation zero. 
To do so, we add and subtract $ \frac{1}{n}\sum_{i=1}^n e_if'(\theta^Tx_i)x_i$ to \eqref{eq:dRn} to obtain 
\begin{equation*}
\eqref{eq:dRn} = \dfrac{1}{n}\sum_{i=1}^{n} \left(\mu(f(\theta^Tx_i)) -\mu(f({\theta^*}^Tx_i)\right)) f'(\theta^Tx_i)x_i^T\Delta+ e_i(f'(\theta^Tx_i)- f'({\theta^*}^Tx_i))x_i^T\Delta.
\end{equation*}
Applying a Taylor expansion around $f({\theta^*}^Tx_i)$, we obtain

\begin{align}
&\left(\triangledown \mathscr{L}_n(\theta) - \triangledown \mathscr{L}_n(\theta^*)\right)^T\Delta \nonumber\\
& = \underbrace{ \dfrac{1}{n}\sum_{i=1}^{n} A''(f({\theta^*}^Tx_i) + v_i(f({\theta}^Tx_i)-f({\theta^*}^Tx_i)))(f({\theta}^Tx_i)-f({\theta^*}^Tx_i))f'(\theta^Tx_i)x_i^T\Delta}_\textrm{\large I}\label{eq:lossDifference1}\\
&\quad\quad +\underbrace{\dfrac{1}{n}\sum_{i=1}^{n} e_i(f'(\theta^Tx_i)- f'({\theta^*}^Tx_i))x_i^T\Delta}_\textrm{\large II}\text{ for }v_i \in [0,1]\label{eq:lossDifference2}
\end{align}
where $A''(\eta) = e^\eta/(1+e^\eta)^2$. 
We will show that the expectation of $I$ is positive. We immediately see $E[e_i(f'(\theta^Tx_i)- f'({\theta^*}^Tx_i))x_i^T\Delta]=0$ because $E[e_i|x_i] = 0$.

We aim to show each inequality
\begin{align}
   I &\geq \kappa_0 \norm{\Delta}_2^2 -\kappa_1 \norm{\Delta}_{\mathscr{G},2,1} \norm{\Delta}_2\sqrt{\frac{\log J +m}{n}}\label{eq:ineq_I}\\
   |II| &\leq \kappa_2 \norm{\Delta}_{\mathscr{G},2,1}\sqrt{\frac{\log J+m}{n}}\label{eq:ineq_II}
\end{align}
holds for all $\Delta \in \{\Delta; \|\Delta\|_2 \leq r\}$ with probability at least $1-\epsilon/2$ for some $\kappa_0, \kappa_1,\kappa_2>0$.

Then 
\begin{align*}
   I+II \geq \kappa_0 \norm{\Delta}_2^2 -\kappa_1 \norm{\Delta}_{\mathscr{G},2,1} \norm{\Delta}_2\sqrt{\frac{\log J +m}{n}}-\kappa_2 \norm{\Delta}_{\mathscr{G},2,1}\sqrt{\frac{\log J+m}{n}}
\end{align*}
holds for all $\Delta \in \{\Delta; \|\Delta\|_2 \leq r\}$ with probability at least $1-\epsilon$. Finally, by the inequality $a^2+b^2 \geq 2ab$, we obtain, 
\begin{align*}
   I+II \geq (\kappa_0/2) \norm{\Delta}_2^2 -(2\kappa_1^2/\kappa_0)  \left(\frac{\log J +m}{n}\right)\norm{\Delta}_{\mathscr{G},2,1}^2-\kappa_2 \sqrt{\frac{\log J+m}{n}}\norm{\Delta}_{\mathscr{G},2,1}
\end{align*}
for all $\Delta \in \{\Delta; \|\Delta\|_2 \leq r\}$ with probability at least $1-\epsilon$.

\subsubsection{Obtaining a lower bound of term \texorpdfstring{$I$}{I}}

We use a similar argument in \citeSupp{negahban_unified_2012}  to obtain a lower bound of the first term. The main difference is that we get the dependence on $ \theta $ for a curvature term, which is not the case for a canonical link $ f(\theta^Tx) = \theta^T x $. Since
$f'(u) = \dfrac{1}{1+e^{u}}$, the first term $I$ becomes
\begin{align*}
I &= \dfrac{1}{n}\sum_{i=1}^{n} A''(f({\theta^*}^Tx_i) + v_i(f({\theta}^Tx_i)-f({\theta^*}^Tx_i)))\dfrac{(x^T\Delta)^2}{(1+e^{x_i^T\theta^*+v_i' x_i^T\Delta})(1+e^{ x_i^T \theta})}.
\end{align*}
for some $ v_i' \in [0,1] $ by Taylor expansion. We note
\begin{align*}
I \geq \dfrac{1}{n}\sum_{i=1}^{n} \dfrac{A''(f({\theta^*}^Tx_i) + v_i(f({\theta}^Tx_i)-f({\theta^*}^Tx_i)))}{(1+e^{x_i^T\theta^*+v_i' x_i^T\Delta})(1+e^{ x_i^T \theta})}(x_i^T\Delta)^2\mathbbm{1}\{ |\Delta^Tx_i| \leq \tau\norm{\Delta}_2\}
\end{align*}
for any $\tau \geq 0$, as $A''(u) = \frac{e^{u}}{(1+e^u)^2} \geq 0 , \forall u$. A suitable $\tau$ will be chosen shortly. Since on the event 
\begin{equation}\label{eq:event}
|\Delta^Tx_i| \leq \tau\norm{\Delta}_2,
\end{equation}
we have
$\theta^Tx_i \leq |{\theta^*}^Tx_i| + |\Delta^Tx_i| \leq K_1^r+\tau r$ and 
\begin{align*}
|f({\theta^*}^Tx_i) + v_i(f({\theta^*}^Tx_i)-f({\theta}^Tx_i))|& \leq |f({\theta^*}^Tx_i) |+ |f({\theta^*}^Tx_i)-f({\theta}^Tx_i)| \\
& \leq \left\lvert \log\dfrac{n_l}{\pi n_u}\right\rvert  + |{\theta^*}^Tx_i| + |\Delta^Tx_i|,
\end{align*}
by Assumption 3 and the fact that $x^T\theta-\log(1+e^{x^T\theta}) $ is 1-Lipschitz in $x^T\theta$, $I$ can be further lower-bounded by 
\begin{equation*}
I \geq \dfrac{L_0(\tau)}{n}\sum_{i=1}^{n}(x_i^T\Delta)^2\mathbbm{1}\{ |\Delta^T x_i| \leq \tau\norm{\Delta}_2\}, 
\end{equation*}
where $L_0(\tau)$ is defined as $\displaystyle L_0(\tau) := \inf_{|u| \leq K_2+K_1^r+\tau r}\dfrac{A''(u)}{(1+e^{K_1^r+\tau r })^2}$. Finally, we truncate each term $(x_i^T\Delta)^2\mathbbm{1}\{ |\Delta^T x_i| \leq \tau\norm{\Delta}_2\}$ so that each term is Lipschitz in $(x_i^T\Delta)$. For a truncation level $ \tau>0 $, we define the following function:
\begin{equation*}
\varphi_{\tau}(u) = 
\begin{cases}
u^2& \text{if }|u|\leq \frac{\tau}{2}\\
(\tau - u)^2 & \text{if } \frac{\tau}{2} \leq |u| \leq \tau\\
0& \text{otherwise}
\end{cases}
\end{equation*}
and note that $I \geq  \frac{1}{n}\sum_{i=1}^{n} L_0(\tau)  \varphi_{\tau\norm{\Delta}_2 }(\Delta^Tx_i) $,
since if the event~\eqref{eq:event} holds, $(\Delta^Tx_i)^2  \geq \varphi_{\tau\norm{\Delta}_2 }(\Delta^Tx_i )$, and both left and right-hand sides are $0$ if the event does not hold.

Defining $I_\ell$ as 
\begin{align}\label{eq:lossdiff1_lowerbound}
I_\ell := \frac{L_0(\tau)}{n}\sum_{i=1}^{n} \varphi_{\tau\norm{\Delta}_2 }(\Delta^Tx_i),
\end{align}
we note that it is sufficient to show the inequality
\begin{equation}\label{eq:ineq_1}
    I_\ell \geq \kappa_0 \norm{\Delta}_2^2 -\kappa_1 \norm{\Delta}_{\mathscr{G},2,1} \norm{\Delta}_2\sqrt{\frac{\log J +m}{n}}
\end{equation}
holds with high probability for all $\Delta \in \{\Delta; \|\Delta\|_2\leq r\} $ to prove \eqref{eq:ineq_I}. To do so, first we will show the inequality \eqref{eq:ineq_1} is true for $\Delta \in \mathbb{S}(\delta,t)$, where we define
\begin{equation}\label{def:S(delta,t)}
  \mathbb{S}(\delta,t):= \{ \Delta \in \mathbb{R}^p ; \norm{\Delta}_2 = \delta, \norm{\Delta}_{\mathscr{G},2,1}/\norm{\Delta}_2 \leq t\}.
\end{equation}
If $\Delta =0$, the inequality \eqref{eq:ineq_1} is trivially true. Otherwise, we show that 
\begin{align}\label{eq:not_uniform_ineq}
   \frac{L_0(\tau) }{n\delta^2}\sum_{i=1}^{n}  \varphi_{\tau\norm{\Delta}_2 }(\Delta^Tx_i)\geq \kappa_0 -\kappa_1 t \sqrt{\frac{\log J+m}{n}},
\end{align}
is true for all $\Delta \in \mathbb{S}(\delta,t)$ with high probability. Then we will use a homogeneity property of $\varphi$ and peeling argument to obtain a uniform result over ($\delta,t$).

\subsubsection{Bounding Expectation of Term \texorpdfstring{$I$}{I}}

We note that $I_\ell$ is lower bounded by,
\begin{align*}
    I_\ell = E[I_\ell] + (I_\ell -E[I_\ell])
     \geq E[I_\ell] -\sup_{\Delta \in \mathbb{S}(\delta,t)} |I_\ell -E[I_\ell]|.
\end{align*}

In this sub-section, we obtain the lower bound of $E[I_\ell]$, which is strictly positive with a suitably chosen $ \tau$. In the next sub-section, we will control the deviation term $\sup_{\Delta \in \mathbb{S}(\delta,t)} |I_\ell -E[I_\ell]|$.
First we have $ E[I_\ell] = L_0(\tau) E\left[\varphi_{\tau\norm{\Delta}_2} (\Delta^Tx ) \right]$ where $x\overset{d}{=} x_i$, and 
\begin{equation*}
E\left[\varphi_{\tau\norm{\Delta}_2} (\Delta^Tx) \right] = E[(\Delta^Tx)^2] - E[(\Delta^Tx)^2 -\varphi_{\tau\norm{\Delta}_2 }(\Delta^Tx) ].
\end{equation*}
We lower and upper bound each two terms on the right-hand side by
\begin{equation*}
E[(\Delta^Tx)^2] \geq  K_0\norm{\Delta}_2^2 
\end{equation*} and 
\begin{align*}
&E[(\Delta^Tx)^2 -\varphi_{\tau\norm{\Delta}_2} (\Delta^Tx) ]  
\leq  E\left[(\Delta^Tx)^2\mathbbm{1}\left\lbrace  |\Delta^Tx| \geq \frac{\tau\norm{\Delta}_2}{2}\right\rbrace \right] 
\end{align*}
Applying the Cauchy-Schwarz inequality, we obtain
\begin{align*}
E\left[(\Delta^Tx)^2\mathbbm{1} \left\lbrace |\Delta^Tx| \geq \frac{\tau\norm{\Delta}_2}{2}\right\rbrace\right] 
&\leq \sqrt{E(\Delta^Tx)^4}\sqrt{\mathbb{P}\left(|\Delta^Tx|\geq \frac{\tau\norm{\Delta}_2}{2}\right)}\\
&\leq 4\sqrt{2}\sigma_x^2 \exp\left(-\frac{\tau^2}{16 \sigma_x^2}\right)\|\Delta\|_2^2
\end{align*}
by using expectation and tail-bound of sub-Gaussians, since $ \Delta^Tx \sim \mbox{subG} (\norm{\Delta}_2^2 \sigma_x^2)$. As
$ 4\sqrt{2} \sigma_x^2 \left(\exp\left(-\frac{\tau^2}{16 \sigma_x^2}\right)\right)\leq \dfrac{K_0}{4}$ for $ \tau^2 \geq 16\sigma_x^2 \log \dfrac{16\sqrt{2} \sigma_x^2}{K_0} $, we take $\tau = K_3 := 4 \sigma_x \left(\log \frac{16 \sqrt{2} \sigma_x^2}{K_0}\right)^{1/2}$ to have
\begin{align}\label{eq:expectation_of_lbI}
E[I_\ell]&= L_0(K_3) E\left[\varphi_{K_3\norm{\Delta}_2} (\Delta^Tx) \right]\nonumber \\
& \geq L_0(K_3) \|\Delta\|_2^2\left(K_0 -4\sqrt{2}\sigma_x^2 \exp\left(-\frac{\tau^2}{16 \sigma_x^2}\right) \right) \nonumber \\
& \geq \norm{\Delta}_2^2\dfrac{3L_0(K_3) K_0}{4}.
\end{align}
For simplicity, we write $L_0 := L_0(K_3)$ for future references.


\subsubsection{Controlling the difference of Term \texorpdfstring{$I$}{I} from its expectation}
We now bound the term $\displaystyle \sup_{\Delta \in \mathbb{S}(\delta,t) } |I_\ell-E[I_\ell]|$ using the concentration property of an empirical process. 
We have $\displaystyle \sup_{\Delta \in \mathbb{S}(\delta,t) } |I_\ell-E[I_\ell]| = \delta^2 L_0 U_1(t)$, where we define $U_1(t)$ as
%
\begin{align*}
U_1(t) &: = \sup_{\Delta\in \mathbb{S}(\delta,t) } \left\lvert \frac{1}{n\norm{\Delta}_2^2}\sum_{i=1}^{n}\varphi_{K_3\norm{\Delta}_2} (\Delta^Tx_i) -E\left[\varphi_{K_3\norm{\Delta}_2 }(\Delta^Tx) \right]\right\rvert,
\end{align*} 
since $\norm{\Delta}_2 = \delta$ for all $\Delta \in \mathbb{S}(\delta,t)$. Since we have $ \norm{\varphi_{K_3\norm{\Delta}_2} }_\infty \leq \dfrac{K_3^2\norm{\Delta}_2^2}{4} $ by definition of $\varphi_\tau(\cdot)$, we apply bounded difference inequality with $c_i = K_3^2/2n$~(Theorem \ref{thm:BDI}) to obtain
\begin{equation*}
\mathbb{P}(U_1(t) \geq EU_1(t) + u_1) \leq  2\exp\left(-\dfrac{8 n u_1^2}{ K_3^4}\right).
\end{equation*}
Setting $u_1 = K_0/4  $, 
\begin{equation}\label{eq:U1empbound}
\mathbb{P}(U_1(t) \geq \mathbb{E}[U_1(t)] + \dfrac{ K_0 }{4}  ) \leq  2\exp(-c_1 n)
\end{equation}
where $ c_1 = K_0^2/2K_3^4 $ is a constant depending on $ K_0$ and $K_3$.
Now we calculate $ EU_1(t) $. 
By symmetrization and contraction inequalities (Theorems \ref{thm:symmetrization},~\ref{thm:contraction}), we have
\begin{align}\label{eq:U1expectationbound}
E[U_1(t)] &\leq 2 E\left[\sup_{\Delta\in \mathbb{S}(\delta,t) }\left\lvert\ \frac{1}{n\norm{\Delta}^2_2}\sum_{i=1}^{n}\epsilon_i \varphi_{K_3\norm{\Delta}_2} (\Delta^Tx_i )\right\rvert\right]\nonumber\\
&\leq \frac{8 K_3 \delta}{\delta^2}E\left[\sup_{\Delta\in \mathbb{S}(\delta,t) }\left\lvert\ \frac{1}{n}\sum_{i=1}^{n}\epsilon_i \Delta^Tx_i \right\rvert\right]\nonumber\\
&\leq 8 K_3\delta^{-1} \left(\sup_{\Delta\in \mathbb{S}(\delta,t) } \norm{\Delta}_{\mathscr{G},2,1} \right)E\left[
\norm{ \frac{1}{n}\sum_{i=1}^{n}\epsilon_i x_i}_{\bar{\mathscr{G}},2,\infty}\right]\nonumber\\
&\leq 8 K_3K_4 t  \sqrt{\dfrac{\log J+m}{n}}
\end{align}
where $(\epsilon_i)_{i=1}^n$ are i.i.d Rademacher variables and $K_4:= 20 \sigma_x (\min_j w_j)^{-1}$. Note that $ \varphi_{K_3\norm{\Delta}_2} $ is a Lipschitz function with the Lipschitz constant = $ 2 K_3 \norm{\Delta}_2 = 2 K_3 \delta$ for $\Delta \in \mathbb{S}(\delta,t)$ which allows us to apply the Ledoux-Talagrand contraction theorem. The second last inequality is from Lemma \ref{lem:blk_holder} and the last inequality follows from $E\left[
\norm{ \frac{1}{n}\sum_{i=1}^{n}\epsilon_i x_i}_{\bar{\mathscr{G}},2,\infty}\right]\leq K_4 \sqrt{\frac{\log J+m}{n}}$, which will be proven shortly in Lemma \ref{lem:2.9}.

Therefore, combining \eqref{eq:expectation_of_lbI}, \eqref{eq:U1empbound} and \eqref{eq:U1expectationbound}, we have
\begin{equation}\label{eq:non_unif_ineq_1}
    \inf_{\Delta \in \mathbb{S}(\delta,t)}\frac{L_0}{n\|\Delta\|_2^2} \sum_{i=1}^{n} \varphi_{K_3\norm{\Delta}_2 }(\Delta^Tx_i) \geq \kappa_0  -\kappa_1' t\sqrt{\frac{\log J +m}{n}} 
\end{equation}
with probability at least $1-\exp(-c_1 n)$ where $\kappa_0 = K_0 L_0/2$ and $\kappa_1' = 8 L_0 K_3 K_4$. It remains to prove Lemma \ref{lem:2.9}.

\begin{lemma}\label{lem:2.9}
\begin{equation}\label{lem2.9:eq1}
	E\left[\norm{ \frac{1}{n}\sum_{i=1}^{n}\epsilon_i x_i }_{\bar{\mathscr{G}},2,\infty}\right] \leq  c \sqrt{\frac{\log J + m}{n}} 
\end{equation}
 for $n\geq \log p$, where $c:=20\sigma_x (\min_j w_j)^{-1}$ is a constant depending on $\sigma_x, (w_j)_1^J$.
\begin{proof}
Conditioned on $  x_1^{n} , \empavg\epsilon_i x_{ij}$ is a sub-Gaussian with a parameter $ \frac{1}{n^2}\sum_i x_{ij}^2 $, since $ \epsilon_i \sim \mbox{subG}(1) $. Then $\empavg\epsilon_i x_{ij} \sim \mbox{subG} (C(x)/n)$, where we define $C(x) = \max_{1\leq j \leq p} \frac{1}{n}\sum_i x_{ij}^2$ conditioned on $x_1^n$.
%
%
Defining $u:=[u_1,\dots,u_p]^T \in \mathbb{R}^p$ as $u_j = \empavg\epsilon_i x_{ij}$, we have independence of $(u_j)_{j\in g_j}$ by Assumption 1. Following similar arguments as in the proof of Lemma 3.1, we obtain for any $j$ and $v \in \mathbb{R}^{|g_j|}$, $v^T u_{g_j} \sim \mbox{subG}((C(x)/n)\norm{v}_2^2)$ and $v^T (u_{g_j}\circ u_{g_j}) \sim \mbox{subExp}(\nu\norm{v}_2,\nu\norm{v}_\infty)$ with $\nu = 16C(x)/n$.
Then Lemma \ref{lem:exp_blknorm} gives, 
\begin{align*}
E\left[\norm{ \frac{1}{n}\sum_{i=1}^{n}\epsilon_i u_{i}}_{\bar{\mathscr{G}},2,\infty}|x_1^n \right] \leq
 4(\min_j w_j)^{-1}\sqrt{C(x)}\sqrt{\frac{\log J + m}{n}} .
\end{align*}
Therefore,
\begin{equation*}
E\left[\norm{ \frac{1}{n}\sum_{i=1}^{n}\epsilon_i x_i }_{\bar{\mathscr{G}},2,\infty}  \right] \leq 4(\min_j w_j)^{-1}\sqrt{\frac{\log J + m}{n}}E[\sqrt{C(x)}]
\end{equation*}
Now we upper-bound $E[\sqrt{C(x)}]$. By Holder's inequality, 
\begin{align*}
    E\left[\sqrt{\max_{1\leq j \leq p}\frac{1}{n}\sum_{i=1}^n x_{ij}^2}\right]\leq E\left[\max_{1\leq j \leq p}\frac{1}{n}\sum_{i=1}^n x_{ij}^2\right]^{1/2}
\end{align*}
Now we define $z_{ij} := x_{ij}^2 - E[x^2_{ij}]$ for each $1 \leq i\leq n$ and $1\leq j \leq p$ and $z_j = [z_{1j},\dots,z_{nj}]^T$. Using Lemma \ref{lem:subG_expk}, we have,
\begin{align*}
    E\left[\max_{1\leq j \leq p}\frac{1}{n}\sum_{i=1}^n x_{ij}^2\right]\leq E\left[\max_{1\leq j \leq p}\frac{1}{n}\sum_{i=1}^n z_{ij} \right]+ 4\sigma_x^2.
\end{align*}
Since $\mathbbm{1}^Tz_j \sim \mbox{subExp}(16\sigma_x^2 \sqrt{n}, 16\sigma_x^2)$ by Lemma \ref{lem:subG_subExp}, we apply Lemma \ref{lem:exp_subExp} with $\nu_* = 16\sigma_x^2\sqrt{n}$, $c = 1/\sqrt{n}$ (taking $m_j=1,\forall j$) to obtain
\begin{align*}
    n^{-1}E[\max_{1\leq j\leq p} \mathbbm{1}^T z_j] \leq n^{-1}16\sigma_x^2(\log p+n/2) = 16\sigma_x^2 \frac{\log p}{n} + 8\sigma_x^2,
\end{align*}
Hence, 
$$E[\sqrt{C(x)}] \leq 4\sigma_x \sqrt{\log p/n +1/2} \leq 5\sigma_x$$
by the condition of $\log p/n \leq 1$, and thus,
\begin{equation*}
E\left[\norm{ \frac{1}{n}\sum_{i=1}^{n}\epsilon_i x_i }_{\bar{\mathscr{G}},2,\infty}  \right] \leq 20 \sigma_x(\min_j w_j)^{-1}\sqrt{\frac{\log J + m}{n}}
\end{equation*}
\end{proof}
\end{lemma}

\subsubsection{Extending the inequality \texorpdfstring{\eqref{eq:non_unif_ineq_1}}{(S29)} for all \texorpdfstring{$\Delta \in \mathbb{B}_2(r)$}{Delta in B2(r)} }

In this section, we show 
\begin{equation}\label{eq:unif_ineq_1}
\frac{L_0}{n\|\Delta\|_2^2} \sum_{i=1}^{n} \varphi_{K_3\norm{\Delta}_2 }(\Delta^Tx_i ) \geq \kappa_0  -\kappa_1 \left(\frac{\|\Delta\|_{\mathscr{G},2,1} }{\|\Delta\|_2}\right)\sqrt{\frac{\log J +m}{n}} 
\end{equation}
holds for all $\|\Delta\|_2 = \delta$ with probability at least $1- \epsilon/2$ where $\kappa_1 = 2\kappa_1'$.  
Note if \eqref{eq:unif_ineq_1} holds, for any $\Delta' $ such that $\|\Delta'\|_2 = \delta'  \neq \delta$,  we can apply \eqref{eq:unif_ineq_1} to $\Delta = \Delta '(\delta/\delta')$ to obtain
\begin{equation*}
\frac{L_0}{n\|\Delta'\|_2^2} \sum_{i=1}^{n} \varphi_{K_3\norm{\Delta'}_2 }(\Delta'^Tx_i ) \geq \kappa_0  -\kappa_1 \left(\frac{\|\Delta'\|_{\mathscr{G},2,1} }{\|\Delta'\|_2}\right)\sqrt{\frac{\log J +m}{n}}
\end{equation*}
by using homogeneity property of $\varphi$,i.e. $ \varphi_\tau(x) =c^{-2}\varphi_{c\tau}(cx) $ for any $c > 0$. Thus proving that \eqref{eq:unif_ineq_1} holds for all $\|\Delta\|_2 = \delta$ with probability at least $1- \epsilon/2$ is enough to prove that the same inequality holds for all $\|\Delta\|_2 \leq r$ with the same high probability. We let $\mathbb{S}_2(\delta) := \{ \Delta \in \mathbb{R}^p; \|\Delta\|_2 = \delta \} $ and $K_w>0$ be a constant such that $\min_j w_j \geq K_w$, where the existence of $K_w$ is guaranteed by Assumption 4.
\begin{align}
&\mathbb{P}\left(\exists \Delta \in \mathbb{S}_2(\delta) \mbox{ such that }\mbox{inequality \eqref{eq:unif_ineq_1} fails } \right)\nonumber\\
&\leq \sum_{l=1}^{N_L} \mathbb{P}\left(\exists \Delta \in \mathbb{S}_2(\delta) ; K_w 2^{l-1} \leq \frac{\|\Delta\|_{\mathscr{G},2,1} }{\|\Delta\|_2} \leq K_w 2^{l} \mbox{ s.t } \mbox{inequality \eqref{eq:unif_ineq_1} fails } \right)\label{eq:fail_ineq_1}
\end{align}
where $2^{N_L} \leq (\max_j w_j/K_w) \sqrt{J}$, i.e. $N_L := \left\lceil \log_2\left(\max_j w_j \sqrt{J}/K_w\right)\right\rceil$, by the inequality $K_w\|\Delta\|_2 \leq (\min_j w_j)\|\Delta\|_2\leq \|\Delta\|_{\mathscr{G},2,1} \leq (\max_j w_j) \sqrt{J} \|\Delta\|_2$.

\begin{align*}
&\sum_{l=1}^{N_L}\mathbb{P}\left(\exists \Delta \in \mathbb{S}_2(\delta) ; K_w 2^{l-1} \leq \frac{\|\Delta\|_{\mathscr{G},2,1} }{\|\Delta\|_2} \leq K_w 2^{l} \mbox{ such that } \mbox{inequality \eqref{eq:unif_ineq_1} fails } \right)\\
&\leq \sum_{l=1}^{N_L}\mathbb{P}\left( \inf_{\Delta \in \mathbb{S}_2(\delta) ; \frac{\|\Delta\|_{\mathscr{G},2,1} }{\|\Delta\|_2} \leq (K_w 2^{l})}\frac{L_0}{n\|\Delta\|_2^2} \sum_{i=1}^{n} \varphi_{K_3\norm{\Delta}_2 }(\Delta^Tx_i ) < \kappa_0  -\kappa_1 (K_w 2^{l-1})\sqrt{\frac{\log J +m}{n}} \right)\\
& = \sum_{l=1}^{N_L}\mathbb{P}\left( \inf_{\Delta \in \mathbb{S}(\delta, (K_w 2^{l}))}\frac{L_0}{n\|\Delta\|_2^2} \sum_{i=1}^{n} \varphi_{K_3\norm{\Delta}_2 }(\Delta^Tx_i ) <\kappa_0  -\kappa_1' (K_w 2^{l})\sqrt{\frac{\log J +m}{n}} \right)\\
&\leq \exp(-c_1n + \log N_L) 
\end{align*}
by $\kappa_1 = 2\kappa_1'$ and the inequality \eqref{eq:non_unif_ineq_1}. Finally,
\begin{align*}
    \exp(-c_1n + \log N_L) 
    \leq \exp\left(-c_1n +  \log \log_2 (J^{3/2}/K_w)\right)
    \lesssim \exp\left(-c_1n + \log \log J \right)\leq \epsilon/2
\end{align*}
 by the sample size condition $n\gtrsim (\log J+m) \vee (1/\epsilon)^{1/\beta}$ and $\max_j w_j /J \leq 1$.

\subsubsection{Controlling the difference of Term \texorpdfstring{$II$}{II} from its expectation}
For the second term, we recall the definition :
\[ \large II =  \dfrac{1}{n}\sum_{i=1}^{n} e_i(f'(\theta^Tx_i)- f'({\theta^*}^Tx_i))x_i^T\Delta,\]
and note that $E[II]=0$ by $E[e_i|x_i] = 0$.
Similar to $U_1(t)$, we define a following quantity,
\[U_2(t) : = \sup_{ (1/2)t\leq \|\Delta\|_{ \mathscr{G},2,1} \leq t} \left\lvert \dfrac{1}{n\|\Delta\|_{\mathscr{G},2,1}}\sum_{i=1}^{n} e_i(f'(\theta^Tx_i)- f'({\theta^*}^Tx_i))x_i^T\Delta. \right\rvert \]
, and bound $E(U_2(t))$ using symmetrization and contraction theorem. First we define
\begin{align*}
    g_i(\Delta^Tx_i) :=  e_i \left(f'({\theta^*}^Tx_i+\Delta^Tx_i)- f'({\theta^*}^Tx_i)\right)\Delta^Tx_i.
\end{align*}
and prove that $ g_i/L_g$ is a contraction map where  $L_g := 3+(K_1^r/4)$.

\begin{lemma}\label{lem:2.10}
$ g_i(s)/L_g$ is a contraction map with $ g_i(0) = 0 $. 
\end{lemma}
\begin{proof}
We consider the first derivative of $g_i$. For ease of notation, we let $u^*_i := {\theta^*}^Tx_i$. We note $f'(u) =1/(1+e^u)$,$f'(u) = -e^u/(1+e^u)^2$. Thus $\sup_u |f'(u)| \leq 1, \sup_u |f''(u)|\leq 1/4$. Also, elementary calculation shows that $\sup_u|uf'(u)|,\sup_u|uf''(u)| \leq 1/2$. Since,
\begin{align*}
    g_i(u) = e_i (f'(u^*_i+u)-f'(u^*_i))u
\end{align*}
we have,
\begin{align*}
    |g_i'(u)| &= |e_i (f''(u^*_i+u)u+f'(u^*_i+u)-f'(u^*_i))|\\
    &\leq |f''(u^*_i+u)(u^*_i+u) - f''(u^*_i+u)u^*_i+f'(u^*_i+u)-f'(u^*_i))|\\
    &\leq 3+(1/4)|u^*_i|
\end{align*}
where $|e_i|\leq 1$ was used in the first inequality.  By Assumption 2, 
$u^*_i:=|{\theta^*}^Tx_i|\leq K_1^r $, thus we can take $L_g := 3+(1/4)K_1^r$.
\end{proof}

Back to $E(U_2(t))$, by symmetrization and contraction theorem (Theorems \ref{thm:symmetrization},\ref{thm:contraction}),
\begin{align}\label{eq:U2expectationbound}
E(U_2(t)) &\leq 4 L_g E\left[ \sup_{ (1/2)t\leq \|\Delta\|_{ \mathscr{G},2,1} \leq t}\left\lvert \dfrac{1}{n\|\Delta\|_{\mathscr{G},2,1}} \sum_{i=1}^{n} \epsilon_i  \Delta^Tx_i \right\rvert\right]\nonumber\\
&\leq 4L_g E\left[ \sup_{ (1/2)t\leq \|\Delta\|_{ \mathscr{G},2,1}\leq t}\dfrac{1}{(1/2)t} \norm{\Delta}_{\mathscr{G},2,1}\norm{ \frac{1}{n}\sum_{i=1}^{n} \epsilon_i  x_i}_{\bar{\mathscr{G}},2,\infty}  \right]\nonumber\\
&\leq  8K_4L_g \sqrt{\dfrac{\log J+m}{n}} .
\end{align}
where the second inequality uses the fact that $(1/2)t\leq \|\Delta\|_{ \mathscr{G},2,1} \leq t$ and Lemma \ref{lem:blk_holder}, and the last inequality comes from Lemma \ref{lem:2.9}.

Now, we apply bounded difference inequality to show that $U_2(t)$ is close to $E(U_2(t))$ with probability at least $1-\exp(-c'n)$. 
We have, 
\begin{align*}
	\sup_{i,\theta}\frac{1}{n\|\Delta\|_{\mathscr{G},2,1}}|g_i(\Delta^Tx_i)| &=\sup_{i,\theta}\frac{1}{n\|\Delta\|_{ \mathscr{G},2,1}}\left\lvert e_i\left( f'({\theta^*}^Tx_i+\Delta^Tx_i)-f'({\theta^*}^Tx_i)\right) \Delta^Tx_i	\right\rvert\\
	&\leq \sup_{i,\theta}\frac{2}{n\|\Delta\|_{\mathscr{G},2,1}}|\Delta^Tx_i|\leq \frac{2}{n}\max_{i,j} w_j^{-1}\|(x_i)_{g_j}\|_{2}
\end{align*}
by Lemma \ref{lem:blk_holder}. We note for any $w \in \mathbb{R}^{p}$ such that $w_{g_j^c}=0$ and $\|w\|_2 = 1$, $u\in\mathbb{R}^p$, defined as $u := \theta^*+r w$, satisfies $\|u-\theta^*\|_2 \leq r$ and $supp(u-\theta^*) \subseteq g_j$. By Assumption 2, $|x_i^Tu| \leq K_1^r$ a.s. for all $i$. Then $|x_i^T w| = |x_i^T(u-\theta^*)|/r \leq 2 K_1^r/r$ a.s., which implies $\|(x_i)_{g_j}\|_2 \leq 2K_1^r/r$ since $\displaystyle\|(x_i)_{g_j}\|_2 = \sup_{v\in \mathbb{R}^{|g_j|};\|v\|_2=1}|(x_i)_{g_j}^T v|=\sup_{w\in \mathbb{R}^{p};\|w\|_2=1, w_{g_j^c}=0}|x_i^T w|$. As the bound holds for any $i,j$, we have $\max_{i,j}  \|(x_i)_{g_j}\|_2 \leq  2K_1^r/r$.

Hence by applying Theorem \ref{thm:BDI} with $c_i = (8K_1^r/K_w r)n^{-1}$, we obtain
\begin{equation*}
\mathbb{P}(U_2(t) \geq EU_2(t) + u_2) \leq  \exp(-2 u_2^2/\sum_{i=1}^n c_i^2 )
\end{equation*}
Taking $u_2 = K_4L_g\sqrt{\dfrac{\log J+m}{n}}$, we get
\begin{align*}
\mathbb{P}\left( U_2(t) \geq 9 K_4L_g \sqrt{\dfrac{\log J+m}{n}}\right)\leq  \exp (- c_2(\log J+m))
\end{align*}
where $ c_2 := (K_w rK_4L_g)^2/32(K_1^r)^2$. In other words, we have shown, for any $t>0$,
\begin{align}\label{eq:non_unif_ineq_2}
&\mathbb{P}\left(\left\lvert \dfrac{1}{n}\sum_{i=1}^{n} e_i(f'({\theta^*}^Tx_i+\Delta^Tx_i)- f'({\theta^*}^Tx_i))x_i^T\Delta \right\rvert \leq  \kappa_2' \|\Delta\|_{\mathscr{G},2,1} \sqrt{\dfrac{\log J+m}{n}},\forall (1/2)t\leq \|\Delta\|_{ \mathscr{G},2,1} \leq t \right)\nonumber\\
&\geq 1-\exp(-c_2(\log J+m))
\end{align}
where we define $\kappa_2' := 9 K_4L_g$.
\subsubsection{Extending the inequality \texorpdfstring{\eqref{eq:non_unif_ineq_2}}{(S34)} for all \texorpdfstring{$\Delta \in \mathbb{B}_2(r)$}{Delta in B2(r)}}
In this section, we obtain a uniform result for term II. More concretely, we consider the following inequality:
\begin{align}\label{eq:unif_ineq_2}
\left\lvert \dfrac{1}{n}\sum_{i=1}^{n} e_i(f'({\theta^*}^Tx_i+\Delta^Tx_i)- f'({\theta^*}^Tx_i))x_i^T\Delta \right\rvert \leq  \kappa_2 \|\Delta\|_{\mathscr{G},2,1} \sqrt{\dfrac{\log J+m}{n}}
\end{align}
where $\kappa_2 := 10K_4L_g$. Equivalently, defining
\begin{equation*}
    \phi(\Delta; x_1^n,z_1^n) := \dfrac{1}{n\|\Delta\|_{\mathscr{G},2,1}}\sum_{i=1}^{n} e_i(f'({\theta^*}^Tx_i+\Delta^Tx_i)- f'({\theta^*}^Tx_i))x_i^T\Delta
\end{equation*}
for $\Delta \neq 0$,
we aim to establish the result,
\begin{align*}
\mathbb{P}\left(|\phi(\Delta; x_1^n,z_1^n)| \leq  \kappa_2  \sqrt{\dfrac{\log J+m}{n}},  \forall \Delta \in \mathbb{B}_2(r)\right)\geq 1-\epsilon/2.
\end{align*}
We first define
\begin{align*}
\mathbb{A}(r_1,r_2) := \{ \Delta \in \mathbb{R}^p ; r_1 <\|\Delta\|_{\mathscr{G},2,1} \leq r_2\}
\end{align*}
and decompose $\mathbb{B}_2(r)$ into different regions. We have,
\begin{align}
&\mathbb{P}\left(\exists \Delta \in \mathbb{B}_2(r) \mbox{ such that }\mbox{inequality \eqref{eq:unif_ineq_2} fails } \right)\nonumber\\
&\leq \mathbb{P}\left(\exists \Delta \in \mathbb{A}(0,C_n) \mbox{ such that }\mbox{inequality \eqref{eq:unif_ineq_2} fails } \right)\label{eq:fail_ineq_2_1}\\
&+\sum_{k=1}^{N_K} \mathbb{P}\left(\exists \Delta \in \mathbb{A}(r_{k-1},r_k) \mbox{ such that }\mbox{inequality \eqref{eq:unif_ineq_2} fails } \right)\label{eq:fail_ineq_2_2}
\end{align}
where we define 
\begin{align*}
C_n&:= K_4L_g\left(\frac{(\min_jw_j)r}{K_1^r}\right)^2 \sqrt{\frac{\log J+m}{n}}\\
r_k&:= C_n2^k.
\end{align*}
Here $C_n$ is chosen to ensure the probability \eqref{eq:fail_ineq_2_1} to be small enough, which will be shown shortly. We take $N_K$ such that $r_{N_K} = C_n2^{N_K} \geq r \max_{j}w_j \sqrt{J}$ since $\|\Delta\|_{\mathscr{G},2,1} \leq (\max_j w_j) \sqrt{J} \|\Delta\|_2\leq r(\max_j w_j) \sqrt{J}$. Then we can let,
\begin{equation*}
N_K  := \left\lceil\log_2 \left(c \max_j w_j \sqrt{\frac{n J}{\log J + m}}\right)\right\rceil
\end{equation*}
for $c := ({K_1^r})^2/(r K_w^2 K_4L_g) \vee 1$.
By the sample size assumption,  $\max_j w_j /n \leq 1$ and $J \gtrsim n^{\beta}$, thus
\begin{align*}
N_K \leq \log_2 \left(c \max_j w_j \sqrt{\frac{n J}{\log J + m}}\right) \leq 2 \log \left(c'  n^{(3+\beta)/2}\right).
\end{align*}for some $c'>1$. Since $\mathbb{P}\left(\exists \Delta \in \mathbb{A}(r_{k-1},r_k) \mbox{ such that }\mbox{inequality \eqref{eq:unif_ineq_2} fails } \right) \leq \exp(-c_2(m+\log J))$ for any $k$ by \eqref{eq:non_unif_ineq_2}, we have for \eqref{eq:fail_ineq_2_2},
\begin{align*}
 \eqref{eq:fail_ineq_2_2} &\leq  \exp (-c_2(m+\log J) + \log N_K) \\
 & \leq 2\exp\left(-c_2(m+\log J)+\log \log c' n^{(3+\beta)/2}\right)\\
 & \leq c_3 \exp (-c_2(m+\log J) + \log \log n )
\end{align*}
for $c_3= 2((3+\beta)/2+\log c') >1$, as $\log \log c' n^{(3+\beta)/2}\leq \log \log n + \log ((3+\beta)/2+\log c')$.

Now we address \eqref{eq:fail_ineq_2_1}:
\begin{align*}
   \eqref{eq:fail_ineq_2_1}= \mathbb{P}\left(\exists \Delta \in \mathbb{A}(0,C_n);|\phi(\Delta; x_1^n,z_1^n)| >  \kappa_2  \sqrt{\dfrac{\log J+m}{n}} \right)
\end{align*}
For  $s \in (0,C_n]$, we define a function $\widetilde{\phi}:\mathbb{R}_+ \times \mathbb{R}^p \rightarrow \mathbb{R}$ , whose first argument takes the size (measured in $\|\cdot\|_{\mathscr{G},2,1}$ norm ), second argument takes normalized direction (i.e. $\|d\|_{\mathscr{G},2,1}=1$) such that
\begin{align*}
    \widetilde{\phi}(s,d;x_1^n,z_1^n):=\dfrac{1}{n}\sum_{i=1}^{n} e_i(f'({\theta^*}^T x_i +s x_i^Td)- f'({\theta^*}^Tx_i))x_i^Td = \phi(sd; x_1^n,z_1^n)
\end{align*}
In particular, for any $\Delta\in \mathbb{A}(0,C_n)$, we have  $\widetilde{\phi}(\|\Delta\|_{\mathscr{G},2,1},\Delta/\|\Delta\|_{\mathscr{G},2,1};x_1^n,z_1^n) = \phi(\Delta;x_1^n,z_1^n).$

Now we calculate how much $\phi$ changes when the size of the input vector varies while fixing the direction. In other words, we calculate the rate of change of $\widetilde{\phi}$ with respect to its first argument. To ease the notation, we suppress the dependence of $\phi, \widetilde{\phi}$ on $(x_1^n, z_1^n)$.
\begin{align*}
    |\frac{d}{ds}\widetilde{\phi}(s,d)|&= \left\lvert\frac{d}{ds}\left(\dfrac{1}{n}\sum_{i=1}^{n} e_i f'({\theta^*}^T x_i +sx_i^Td)- f'({\theta^*}^Tx_i))x_i^Td\right)\right\rvert\\
    &\leq \dfrac{1}{n}\sum_{i=1}^{n}\left\lvert e_i f''({\theta^*}^T x_i +sx_i^Td)\right\rvert (x_i^Td)^2\\
    &\leq \dfrac{1}{4}\|x_i\|_{\bar{\mathscr{G}},2,\infty}^2\|d\|_{\mathscr{G},2,1}^2 \leq \left(\frac{K_1^r}{(\min_jw_j)r}\right)^2
\end{align*}
by $|e_i|\leq 1$ and $\|f''\|_\infty \leq (1/4)$. Then for any normalized direction $d \in \mathbb{R}^p$ such that $\|d\|_{\mathscr{G},2,1} = 1$, we have,
\begin{equation*}
    |\widetilde{\phi}(s,d) - \widetilde{\phi}(u,d)| \leq \left(\frac{K_1^r}{(\min_jw_j)r}\right)^2|s-u|
\end{equation*}
In particular, for any $0<s \leq C_n$,
\begin{equation*}
    |\widetilde{\phi}(s,\Delta/\|\Delta\|_{\mathscr{G},2,1})| \leq   |\widetilde{\phi}(C_n,\Delta/\|\Delta\|_{\mathscr{G},2,1})|+\left(\frac{K_1^r}{(\min_jw_j)r}\right)^2 C_n
\end{equation*}

Therefore,
\begin{align}
    \eqref{eq:fail_ineq_2_1}&=\mathbb{P}\left(\exists \Delta \in \mathbb{A}(0,C_n);|\widetilde{\phi}(\|\Delta\|_{\mathscr{G},2,1} ,\Delta/\|\Delta\|_{\mathscr{G},2,1}) |>  \kappa_2  \sqrt{\dfrac{\log J+m}{n}} \right)\nonumber\\
    &\leq \mathbb{P}\left(\exists \Delta \in \mathbb{A}(0,C_n);|\widetilde{\phi}(C_n,\Delta/\|\Delta\|_{\mathscr{G},2,1}) |>  \kappa_2  \sqrt{\dfrac{\log J+m}{n}}-\left(\frac{K_1^r}{(\min_jw_j)r}\right)^2  C_n\right)\nonumber\\
    & = \mathbb{P}\left(\exists \Delta \in \mathbb{A}(0,C_n);|\widetilde{\phi}(C_n,\Delta/\|\Delta\|_{\mathscr{G},2,1}) |> 9K_4L_g  \sqrt{\dfrac{\log J+m}{n}}\right)\nonumber
\end{align}
where the last line uses the fact $\left(\frac{K_1^r}{(\min_jw_j)r}\right)^2C_n =K_4L_g\sqrt{\frac{\log J+m}{n}}$.
Since $\widetilde{\phi}(C_n,\Delta/\|\Delta\|_{\mathscr{G},2,1}) = \phi(C_n\Delta/\|\Delta\|_{\mathscr{G},2,1})$ and $C_n\Delta/\|\Delta\|_{\mathscr{G},2,1} \in \{ \Delta'\in \mathbb{R}^p ; \|\Delta'\|_{\mathscr{G},2,1} =C_n\}$, we have,
\begin{align*}
    \eqref{eq:fail_ineq_2_1}&\leq \mathbb{P}\left( \sup_{\|\Delta\|_{\mathscr{G},2,1} = C_n} |\phi(\Delta) |>  9K_4L_g \sqrt{\dfrac{\log J+m}{n}}\right)\\
    &\leq \mathbb{P}\left( \sup_{(1/2)C_n\leq \|\Delta\|_{\mathscr{G},2,1} \leq C_n} |\phi(\Delta) |>  9K_4L_g \sqrt{\dfrac{\log J+m}{n}}\right)\\
    &\leq \exp(-c_2(\log J +m))
\end{align*}
by \eqref{eq:non_unif_ineq_2}.
Therefore,
\begin{align*}
\eqref{eq:fail_ineq_2_1}+\eqref{eq:fail_ineq_2_2}&\leq \exp(-c_2(\log J+m))+ c_3 \exp(-c_2(m+\log J)+\log \log n )\\
&\leq 2c_3 \exp(-c_2(m+\log J)+\log \log n ) \leq \epsilon/2
\end{align*}
 by the sample size condition $n\gtrsim (\log J+m) \vee (1/\epsilon)^{1/\beta}$, noting $\log \log n = o(\log J)$.

\clearpage

\section{Supplementary simulation results in Section 4}\label{supp_sec:extra_simul_results}

In this section, we display additional classification performance results. We recall the simulation setting: dimension of features $p \in (10,5000)$, auto-correlation level among features $\rho \in (0, 0.2, 0.4, 0.6, 0.8)$, separation distance $ d\in (1.5,2.5,3.5)$, and the model specification scheme (logistic, misspecified). The sample size is $n_\ell = n_u =500$ in all setting and experiments are repeated 50 times. 
\subsection{The logistic model scheme}
\subsubsection{\texorpdfstring{$F_1$}{F1} scores under the logistic model scheme}

\begin{figure}[htbp]
    \centering
    \includegraphics[width=1\linewidth,height=.5\textheight]{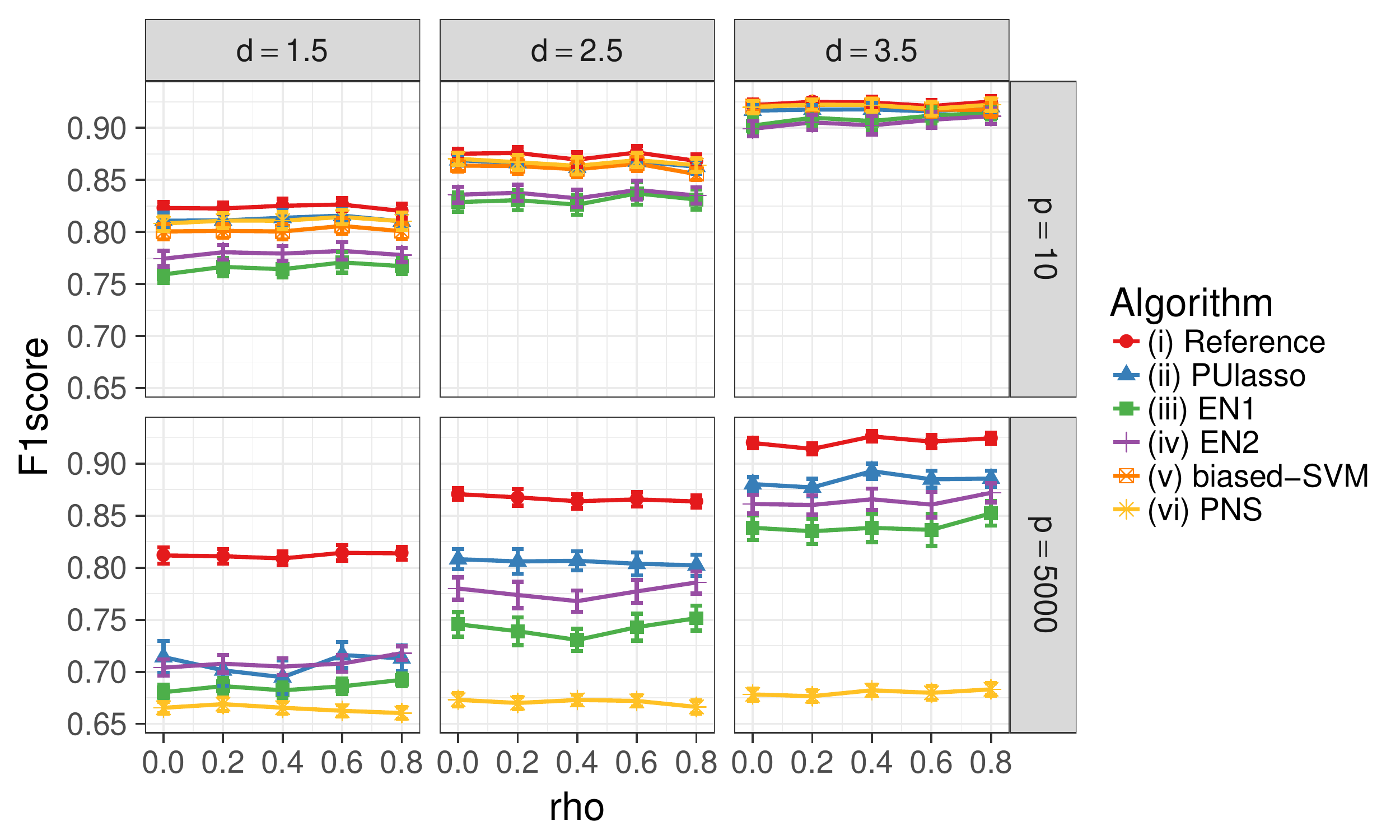}
    \caption{$F_1$ scores of algorithms (i)-(vi) under correct (logistic) model specification }
    \label{fig:S3.1-f1}
\end{figure}

\clearpage
\subsection{The misspecified model scheme}
Heavy-tailed distribution tends to generate more separated samples, leading to better classification performance. The scaling of $\Sigma_\rho$, which sets $Var(x_i^T\theta^*)$ the same across $\rho$, indirectly changes the separation between the two classes. As a result, we observe improved classification performance with higher $\rho$ in the misspecified setting. PUlasso algorithm continues to out-perform other algorithms in most cases, but performance difference among algorithms decreases under the model misspecification scheme.
\subsubsection{Mis-classification rates under the misspecified model}

\begin{figure}[h]
\centering
\includegraphics[width=1\linewidth,height=.5\textheight]{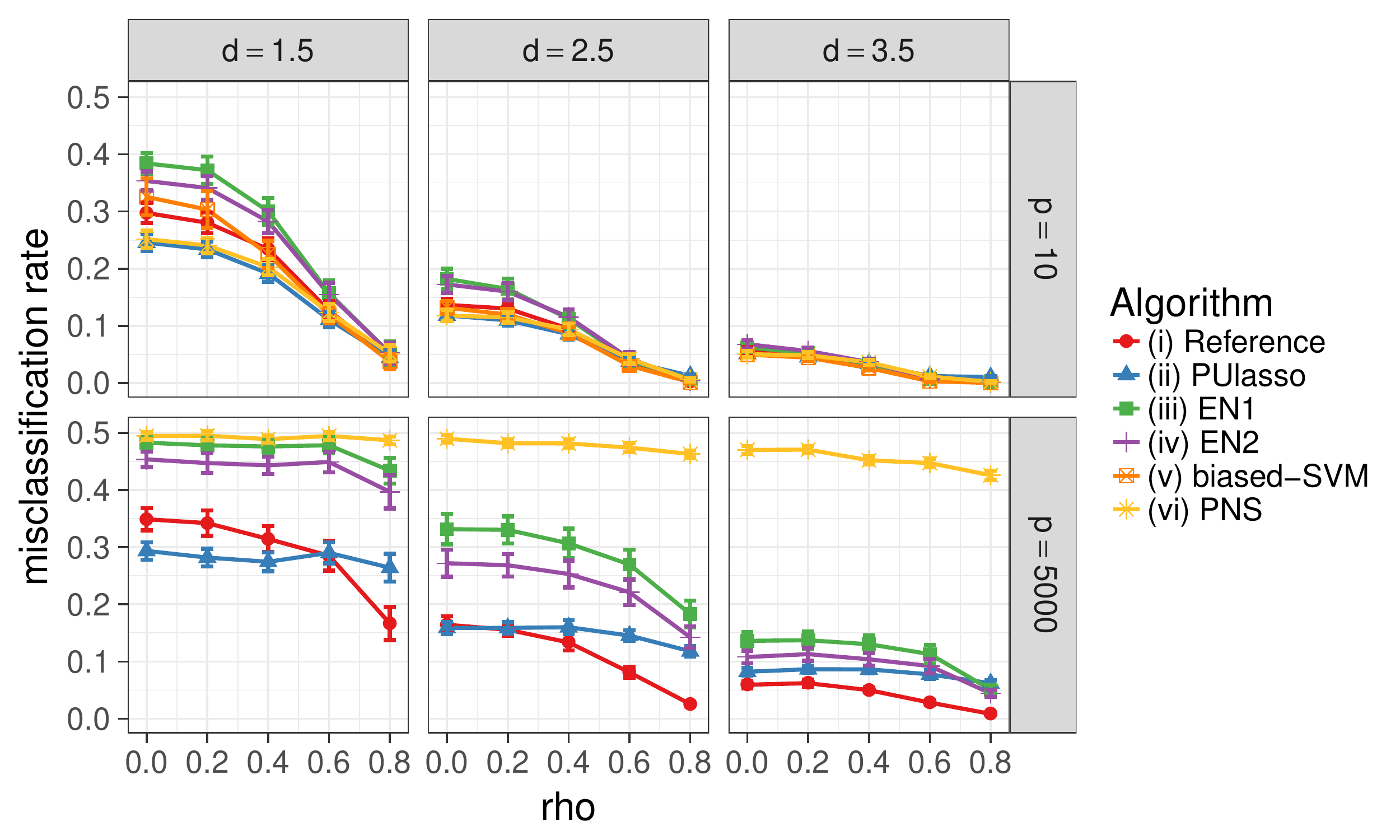}
    \captionof{figure}{Mis-classification rates of algorithms (i)-(vi) under model misspecification. }
     \label{fig:S3.2-mclr}
\end{figure}

\clearpage
\subsubsection{\texorpdfstring{$F_1$}{F1} scores under the misspecified model}
\vspace{.5in}
\begin{figure}[h]
\includegraphics[width=1\linewidth,height=.5\textheight]{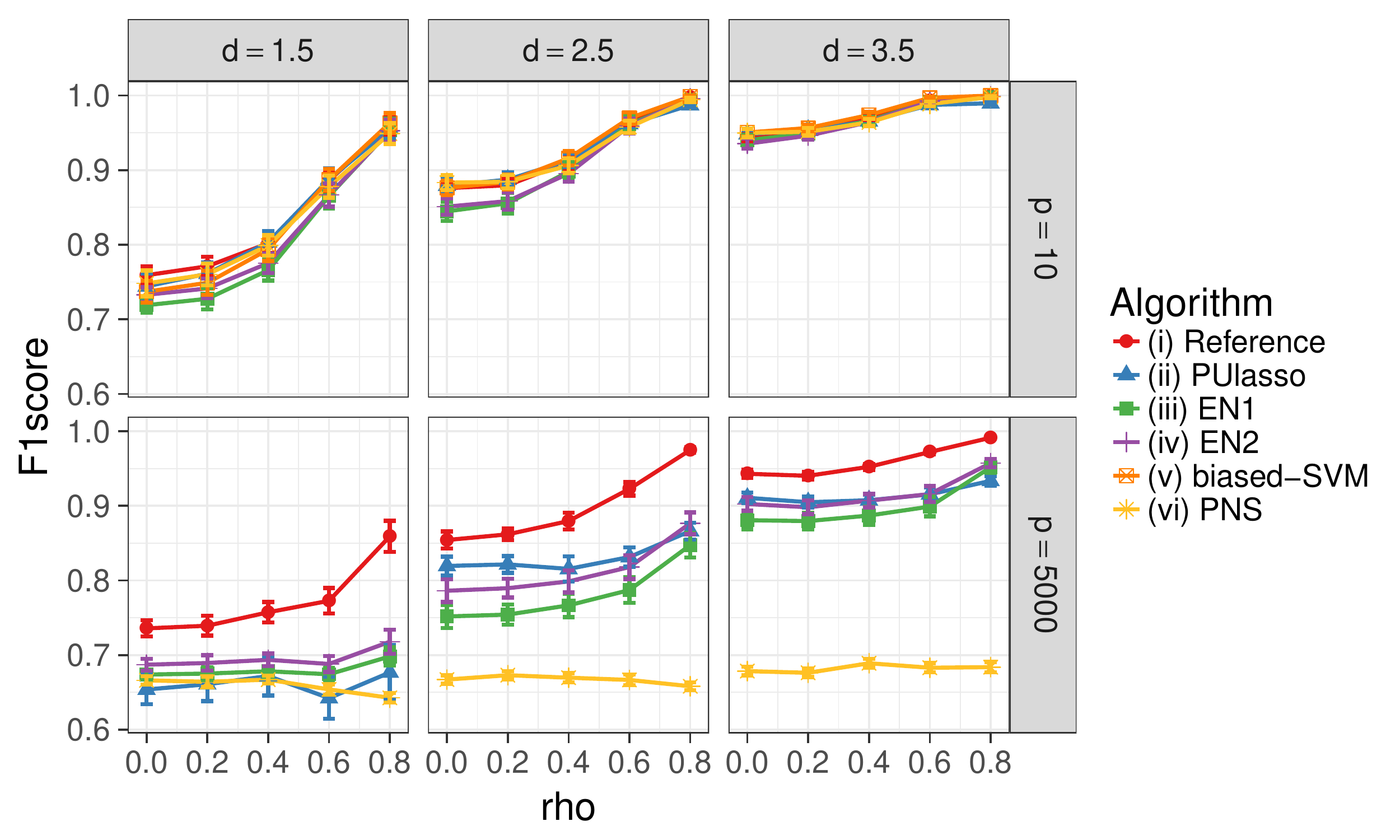}
    \captionof{figure}{$F_1 $ scores of algorithms (i)-(vi) under model misspecification}
    \label{fig:S3.2-f1}
    \end{figure}
\bibliographystyleSupp{plainnat}
\bibliographySupp{PUlasso}
\end{document}